\documentclass[a4paper,11pt,reqno]{amsart}
          \usepackage{amssymb}
	  \usepackage{amsmath}
	  \usepackage{amsthm}
          \usepackage{amsfonts}
          \usepackage[english]{babel}
          \usepackage[utf8]{inputenc}
          \usepackage{enumitem}

\usepackage[margin=2.5cm]{geometry}

\usepackage{tikz}
\usetikzlibrary{arrows,shapes}

 \usepackage[unicode,colorlinks,plainpages=false,hyperindex=true,bookmarksnumbered=true,bookmarksopen=false,pdfpagelabels]{hyperref}
 \hypersetup{urlcolor=cyan,linkcolor=blue,citecolor=red,colorlinks=true}
\newtheorem{thm}{Theorem}[section]
\newtheorem{cor}[thm]{Corollary}
\newtheorem{lem}[thm]{Lemma}
\newtheorem{prop}[thm]{Proposition}

\newtheorem{defn}[thm]{Definition}
\newtheorem{assume}[thm]{Assumption}
\newtheorem{scen}[thm]{Scenario}
\theoremstyle{definition}
\newtheorem{rem}[thm]{Remark}
\newtheorem*{conv*}{Conventions}
\newtheorem*{conte*}{Content}

\numberwithin{equation}{section}

\newcommand{\be}{\begin{equation}}
\newcommand{\ee}{\end{equation}}
\newcommand{\bea}{\begin{eqnarray}}
\newcommand{\eea}{\end{eqnarray}}
\newcommand{\ba}{\begin{aligned}}
\newcommand{\ea}{\end{aligned}}

 \usepackage{color}
\definecolor{lila}{rgb}{1,0.2,0.9}
\definecolor{purple}{rgb}{0.4,0,1}
\definecolor{darkgreen}{rgb}{0,0.5,0}

\def\bI{\mathbb{I}}                         %
\def\span{\mathrm{span}}                    %
\def\1{{\boldsymbol 1}}                     %
\def\cD{{\mathcal D}}                       %
\def\cH{{\mathcal H}}                       %
\def\tr{\mathrm{tr}}                        %
\def\ri{{\rm i}}                            %
\def\bC{\mathbb{C}}                         %
\def\bR{\mathbb{R}}                         %
\def\bZ{\mathbb{Z}}                         %
\def\cF{{\mathcal F}}                       %
\def\reg{\mathrm{reg}}                      %
\def\id{{\mathrm{id}}}                      %
\def\dt {\left.\frac{d}{dt}\right|_{t=0}}   %
\def\cG{{\mathcal G}}                       %
\def\Dress{{\mathrm{Dress}}}                %
\def\red{{\mathrm{red}}}                    %
\def\fP{\mathfrak{P}}                       %
\def\ad{\mathrm{ad}}                        %
\def\Ad{\mathrm{Ad}}                        %
\def\cA{\mathcal{A}}                        %
\def\cM{\mathcal{M}}                          %
\def\fR{\mathfrak{R}}                       %
\def\cC{\mathcal{C}}                        %
\def\cO{\mathcal{O}}                        %
\def\cW{\mathcal{W}}                        %
\def\cU{\mathcal{U}}                        %
\def\sl2z{\mathrm{SL}(2,\bZ)}               %
\def\cS{{\mathcal S}}                       %
\def\bT{\mathbb{T}}                         %
\def\bR{\mathbb{R}}                         %
\def\fH{\mathfrak{H}}                         %
\def\fF{\mathfrak{F}}                         %
\def\bG{\mathbb{G}}
\def\fk{\mathfrak{k}}
\def\ft{\mathfrak{t}}                         %
\def\fb{\mathfrak{b}}                         %
\def\princ{\mathrm{princ}}
\def\taU{\underline{\tau}}
\def\cT{\mathcal{T}}
\def\cZ{\mathcal{Z}}
\def\cP{\mathcal{P}}

\def\fus{{\mathrm{fus}}}                  %
\def\bD{\mathbb{D}}                         %
\def\Gone{G^1}
\def\Gtwo{G^2}
\def\ddim{{\mathrm{ddim}}}                %
\def\ccM{\widehat {\mathcal{M}}}                        %
\def\cV{\mathcal{V}}
\def\proj{{\mathrm{proj}}}

\begin{document}

\title{Integrable systems from Poisson reductions \\
of generalized Hamiltonian torus actions}

\maketitle

\begin{center}

L. Feh\'er${}^{a,b}$
and M. Fairon${}^c$

\medskip
${}^a$Department of Theoretical Physics, University of Szeged\\
Tisza Lajos krt 84-86, H-6720 Szeged, Hungary\\
e-mail: lfeher@physx.u-szeged.hu

\medskip
${}^b$Institute for Particle and Nuclear Physics\\
Hun-Ren Wigner Research Centre for Physics\\
 H-1525 Budapest, P.O.B.~49, Hungary

\medskip
${}^c$Université Bourgogne Europe, CNRS, IMB UMR 5584\\
F-21000 Dijon, France\\
e-mail: maxime.fairon@u-bourgogne.fr
\end{center}

\begin{abstract}
We develop a set of sufficient conditions for guaranteeing that an integrable
system with a symmetry group $K$ on a manifold $M$ descends to an integrable
system on a dense open subset of the quotient Poisson space $M/K$.
The higher dimensional phase space $M$ carries a bivector $P_M$ yielding a bracket
on $C^\infty(M)$ such that $C^\infty(M)^K$ is a Poisson algebra.
The unreduced system on $M$  is supposed to possess `action variables'
that generate a proper, effective action of a group of the form $\mathrm{U}(1)^{\ell_1} \times \mathbb{R}^{\ell_2}$
and descend to action variables of the reduced system.
In view of the form of the group and since $P_M$ could be  a quasi-Poisson bivector,
we say that we work with a generalized Hamiltonian torus action.
The reduced systems are in general superintegrable owing to the large set of
invariants of the proper Hamiltonian action of $\mathrm{U}(1)^{\ell_1} \times \mathbb{R}^{\ell_2}$.
We present several examples
and apply  our  construction for solving open problems
regarding the integrability of systems obtained previously by reductions of master
systems on doubles of compact Lie groups: the cotangent bundle, the
Heisenberg double and the quasi-Poisson double.
Furthermore, we offer numerous applications to integrable systems living on moduli spaces of flat connections,
using the quasi-Poisson approach.
\end{abstract}

 \setcounter{tocdepth}{2}

 \newpage

 \tableofcontents

 \newpage

\section{Introduction}

Integrable systems have found important applications in several branches of physics,
and their rich mathematical structures continue to attract attention.
Some of those systems enjoy exceptionally  large symmetries that render them superintegrable (degenerately integrable),
as is exemplified by the well-known $\mathrm{SO}(4)/\mathrm{SO(3,1)}$ symmetry of the Kepler--Coulomb  problem.
 In this paper we are interested in (super)integrable Hamiltonian
systems on finite dimensional Poisson (including symplectic) manifolds.
Many useful reviews of such systems are available \cite{Fas,MPW,Re2,SInt}, and from the various possible
 definitions  \cite{FeCONF,LGMV,MF,Nek,Re2}  we  adopt the following one.

\begin{defn}\label{defn:int}
Let $(\cM,P_{\cM})$ be a finite dimensional, connected, smooth Poisson manifold,
and $\fH$ an Abelian Poisson subalgebra of $C^\infty(\cM)$ subject to the conditions:
\begin{enumerate}[label=\arabic*$)$]
\item As a commutative algebra of functions, $\fH$ has functional dimension  $\ddim(\fH) = \ell\geq 1$. \label{it:int1}
\item The Hamiltonian vector fields of the elements of $\fH$ are complete and span an
$\ell$-dimensional subspace of the tangent space over a dense open subset of $\cM$. \label{it:int2}
\item The commutant $\fF$ of
$\fH$ in $C^\infty(\cM)$, which contains the joint constants of motion of the Hamiltonians $\cH \in \fH$,
has functional dimension $\ddim(\fF)=\dim(\cM) - \ell$.  \label{it:int3}
\end{enumerate}
We refer to the quadruple $(\cM,P_{\cM},\fH,\fF)$, or simply  $\fH$, as a \emph{(degenerate) integrable system of rank $\ell$}.
\end{defn}

In general, the completeness of the flows of $\fH$
is automatic if
the joint level surfaces of the constants of motion are compact.
On symplectic manifolds the second part of condition \ref{it:int2} follows from condition \ref{it:int1}.
 The standard notion of Liouville integrability results if  $\cM$ is a symplectic manifold
 and $\dim(\cM)=2\ell$.
In this case condition \ref{it:int3} is  implied by condition \ref{it:int1}.
Liouville integrability on Poisson manifolds is the case for which $\ell = k$, where $k$ is
half the dimension
of the maximal symplectic leaves, which are supposed to fill a dense subset of $\cM$.
When $\ell < k$, both on Poisson and symplectic manifolds, then one obtains the case
of degenerate integrability, alternatively called superintegrability.

We are witnesses  to intense research activity on superintegrable
Hamiltonian systems
obtained, e.g., in the presence of a magnetic field \cite{DGJ,HS},
by an algebraic characterization of first integrals \cite{BGSH,Ts2,RZ},
or by prescribing geometric properties of tensor fields \cite{MMcLSV,V}.
Of particular importance for the present study are examples  related to (spin) many-body systems of
Toda \cite{ReS}, Calogero-Moser-Sutherland \cite{Fe24,FoH,Re3,Re4} and Ruijsenaars-Schneider \cite{AO,CF,FFM} type,
as well as their recent generalizations \cite{ARe,Fa,Fe25}.
For concrete systems, either the Abelian algebra $\fH$ is defined, or $\ell$ independent Hamiltonians in involution are usually given,
and one may apply various methods for exhibiting the right number of constants of motions.
Such methods include, for example,  explicit constructions based on some artfully chosen ansatz \cite{MPW}
and Hamiltonian or Poisson reduction methods \cite{Aru,KKS,OR,Per}. In the latter approach, one starts with an
unreduced `master integrable system' on a higher dimensional phase space
that has  a large symmetry group and explicitly known constants of motion.
Then one takes a suitable (symplectic or Poisson) quotient by the symmetry group and investigates if enough constants of motion survive the
reduction. By using this method, the superintegrability of several interesting many-body systems
of spin Calogero--Moser--Sutherland and Ruijsenaars--Schneider type
has been shown on generic symplectic leaves of  Poisson quotient spaces
(see e.g. \cite{FF,Fe23,FeCONF,Re1} in addition to references cited previously).

The goal of the present paper is to further develop the Poisson reduction approach to the construction
of integrable Hamiltonian systems.
We will apply our formalism for demonstrating integrability
and exhibiting (generalized) action-angle variables in interesting examples.
We were influenced mainly by the works of Reshetikhin \cite{Re1,Re2}, Zung \cite{Zu1},
and others \cite{BJ,J} dealing with the question of integrability of reduced systems.
A fundamental underpinning for our investigations will be provided   by the following simple observation.
 Let $G:=\mathrm{U}(1)^{\ell_1} \times \bR^{\ell_2}$, where the superscripts denote the dimensions,
with $\ell=\ell_1+\ell_2>0$.
Suppose that we have a proper and effective  Hamiltonian action of the $\ell$-dimensional
`generalized torus' $G$ on a connected Poisson manifold, now denoted $Y$.
Then invariant theory guarantees the equality
\be
\ell + \ddim(C^\infty(Y)^{G}) = \dim(Y).
\label{Ieq}\ee
This is obvious if the $G$-action is free, in which case $Y/G$ is a smooth manifold.
Based on \eqref{Ieq}, we obtain an integrable system by
taking $\fH$ to be the set of functions depending only on the components of the momentum map generating
the torus action. In short,
if such a generalized torus action is given, then an integrable system in the sense of Definition \ref{defn:int} results immediately.

A substantial part of our results will be formulated using the term
 \emph{generalized action variables} of integrable systems.
To make  our statements precise, we now make a formal definition that will fit our examples.

\begin{defn}\label{defn:genact}
Consider an integrable system of rank $\ell$ on a connected Poisson manifold $(\cM, P_\cM)$ given by the Abelian Poisson algebra $\fH$.
 Suppose that we have $\ell$ smooth functions $H_1,\dots, H_\ell$ on a connected dense open submanifold $  \widetilde \cM \subset \cM$
verifying the properties:
\begin{enumerate}[label=\roman*$)$]
\item
The map $(H_1,\dots, H_\ell): \widetilde \cM \to \bR^\ell$ is  an equivariant momentum map for a proper and effective action of an $\ell$-dimensional `generalized torus'
$\mathrm{U}(1)^{\ell_1} \times \bR^{\ell_2}$ on $\widetilde\cM$.
\item
The restriction of the elements of $\fH$ on $\widetilde \cM$ can be expressed in terms of $H_1,\dots, H_\ell$
and the span of the exterior derivatives of the elements of $\fH$ coincides
with the span of the exterior derivatives $dH_1,\dots, dH_\ell$ at every point of $\widetilde \cM$.
\end{enumerate}
 Then, we say that the functions $H_1,\dots, H_\ell$ are \emph{generalized action variables} on $\widetilde \cM$ for the integrable system $\fH$.
\end{defn}
 The generalized action variables can be extended to action-angle and transversal coordinates, at least in an invariant
 open neighbourhood of a principal orbit.
 On symplectic manifolds, this can be deduced from results of
\cite{FS,Nek}.
 On Poisson manifolds, and assuming $\ell_2=0$,  it follows from an action-angle theorem proved in \cite{LGMV}.
We show in Corollary \ref{Cor:AA} of Theorem \ref{thm:AA} that such coordinates also exist in our general case.

Motivated by the above remarks, our approach for constructing integrable systems will be
based on (generalized) Poisson reduction applied to (generalized) Hamiltonian torus actions.
We shall start with a (generalized) Hamiltonian torus action that
commutes with the action of a symmetry group, say $K$,  and then go to the quotient space by
the symmetry group.   What does the term `\emph{generalized}' indicates here?
Firstly, it will be allowed that the unreduced manifold $Y$  is not necessarily a Poisson manifold, but
 $C^\infty(Y)$ will always be equipped with a bracket $\{-,-\}$, coming from  a bivector $P_Y$,
such that $C^\infty(Y)^K$ is a genuine Poisson algebra. This generality is needed in order
to cover quasi-Poisson manifolds \cite{AKSM,AMM} whose reductions underlie many interesting
integrable systems.
Secondly, as mentioned above,
we will work with a `generalized torus'
$G=\mathrm{U}(1)^{\ell_1} \times \bR^{\ell_2}$ whose action on $Y$ is engendered by vector fields
associated with functions $H_i \in C^\infty(Y)^K$, for $i=1,\dots, \ell_1 + \ell_2$.
Under favourable circumstances described in the text,
our construction gives rise to a true  integrable system on $Y_0/K$, where $Y_0 \subset Y$
is a certain dense open submanifold of $Y$.
Actually we shall prove much more,
see Theorem \ref{thm:G7}  and its subsequent corollaries.
The proofs will rely on the inspection of  the ring of invariant functions for the action
of the direct product group\footnote{In Section 2, it will be convenient to denote the symmetry group as $G^1$ and
the generalized torus as $G^2$.}
$\bG = K \times G$ and of its subgroups: the symmetry group  $K$
and the `generalized torus' $G$.
We shall suppose that the $\bG$-action is proper and the conditions
formulated in Assumption \ref{assume:G5} hold.

\medskip

An advice to our reader could be useful at this point.
Namely, since  the assumptions adopted in Section \ref{Sec:2} may appear complicated or \emph{ad hoc} at first sight,
it should be helpful to study the simplest examples of Section \ref{Sec:Cot} in parallel
with the digestion of the general construction.

After focusing on reductions of (generalized) torus actions,
we shall formulate the general Scenario \ref{scen:G11}
in which  the unreduced (generalized) torus action arises from  generalized action variables of an unreduced
integrable system  with \emph{compact} symmetry group $K$ on a suitable
Poisson or quasi-Poisson manifold $(M, P_M)$.
In particular, the previously sketched situation can be applied to a dense open submanifold $Y \subset M$.
This will lead to our main result, given
by Theorem \ref{thm:G12}.
Here, it should be noted that Definition \ref{defn:int} can be maintained
on quasi-Poisson $K$-manifolds, too, with the additional requirement that $\fH$
consists of $K$-invariant functions.

Our abstract results are tailor-made for the analysis of
interesting examples. We illustrate this  by proving integrability statements for all
systems that result by the previously  studied \cite{Fe23,FeCONF,Fe24,Re1} Poisson reductions of
 integrable master systems on doubles of compact Lie groups.
 The doubles in question are the cotangent bundle, the Heisenberg double
 and the internally fused quasi-Poisson double.
 We complete and considerably enhance the known results concerning these reduced systems.
 For example, we prove the integrability of those
 systems, on every symplectic leaf of a dense open subset of the reduced phase space,
  whose Hamiltonians stem from the class functions of the relevant compact Lie group.
This solves a problem  left open in \cite{Fe23}, which actually provided
one of our original motivations for the present work.
Our method also gives us a convenient way to understand the integrability
of several Hamiltonian systems on moduli spaces of flat connections,
which cover real forms of some
systems investigated by Arthamonov and Reshetikhin \cite{ARe}.
It is worth stressing that,  in all of our examples, the integrability of the reduced systems will be demonstrated by
utilizing explicitly identified generalized action variables.
This may be deemed an important result in itself, since the construction of action-angle variables is  in general a subtle problem.

\begin{conte*}
Subsection \ref{ss:Proper} is devoted to preliminaries on proper group actions.
The core general results of the paper are contained in Subsections \ref{ss:genIS} and \ref{ss:AA}.
 Our principal achievement is Theorem \ref{thm:G12} that establishes
how integrable systems result via reduction starting from any realization
 of Scenario \ref{scen:G11}.
 This is complemented by Proposition \ref{prop:intM} and Corollary \ref{Cor:AA} of Theorem \ref{thm:AA},
dealing with the integrability of the unreduced system
of Scenario \ref{scen:G11} and with generalized action-angle coordinates.
The usefulness of our general construction is shown by exhibiting applications.
Examples arising from cotangent bundles and Heisenberg doubles are presented, together with
necessary background material, in Sections \ref{Sec:Cot} and \ref{Sec:Heis}.
Then, we treat in Section \ref{Sec:qPoiss} numerous examples of integrable systems on moduli spaces of flat connections using the framework of quasi-Poisson geometry.
The main text ends with Section \ref{S:concl}, where we briefly recapitulate the results and discuss open problems.
Some auxiliary technical material  and further examples are presented in the appendices.

To facilitate the reading of the paper, we note that our main results
pertaining to examples are summarized  by Propositions \ref{pr:J10} (for the cotangent bundle),
\ref{pr:H3} (for the Heisenberg double) and \ref{pr:Z2} (for the quasi-Poisson double).
Furthermore, Proposition  \ref{pr:Z11} deals with a
rich collection of integrable systems on moduli spaces of flat connections over a surface of
arbitrary genus $m$ with $n+1$ boundary components.

\end{conte*}

\begin{conv*}
 First, \emph{smooth}  means $C^\infty$ and we only work with smooth manifolds, functions, actions, etc.  throughout this manuscript.
Second, our assumptions about the underlying Lie groups will `evolve' as follows.
In Section \ref{Sec:2}, we start with a proper action of a direct product Lie group $\bG=\Gone\times \Gtwo$ \eqref{bG12}, whose factors
$\Gone$ and $\Gtwo$ later become the symmetry group of an unreduced system and a `generalized torus'.
The compactness of $G^1$ is first assumed in Scenario \ref{scen:G11}.
This assumption is needed for us in order to cover the quasi-Poisson case, but otherwise all
the subsequent general results of
Section \ref{Sec:2}  require only the properness of the action of $\bG$.
In the examples, $G^1$ will be chosen as a compact connected and simply connected Lie group $K$ corresponding to a simple Lie algebra.
\end{conv*}

\section{A general mechanism  for reduced  integrability} \label{Sec:2}

Here, we first present a collection of results about proper actions of Lie groups on smooth manifolds.
Our presentation is based mainly on the reviews found in the books \cite{DK,MMOPR,Mi,OR}, and no originality is claimed.
These results will be crucial for the general construction of integrable systems discussed
in the following two Subsections, which will be applied throughout the rest of the manuscript.
All our manifolds are assumed to be finite dimensional and satisfy the second axiom of countability.
When we talk about a submanifold without qualification, we always mean an \emph{embedded} submanifold.

\subsection{Background material on proper actions} \label{ss:Proper}

Suppose that a  Lie group $\bG$ acts on a smooth manifold $Y$ and denote the smooth action map
 $\bG\times Y \to Y$ by juxtaposition $(g,y) \mapsto gy$. This  is a \emph{proper} action if for any sequence $(g_n, y_n)$  ($n\in \mathbb{N}$) for which
 both $y_n$ and $g_n y_n$ converge, there exists a convergent subsequence of the sequence $g_n$.
 Actions of compact Lie groups are automatically proper.
 For a proper action, every orbit $\bG y$ is a closed submanifold of $Y$. The orbit $\bG y$ is diffeomorphic to the quotient manifold   $\bG/\bG_y$, where $\bG_y$
 is the isotropy (stabilizer) group of $y \in Y$, which is a closed, compact subgroup of $\bG$.

For any subgroup $L < G$ denote by $(L)$ the collection of subgroups conjugate to $L$. By definition,
the point $y$ has isotropy type $(L)$ if $\bG_y\in (L)$.
 The points having mutually conjugate isotropy groups have the same isotropy type,
 and the corresponding orbits  are $\bG$-equivariantly diffeomorphic to each other.
 The smooth $\bG$-invariant functions separate the points of the quotient topological space $Y/\bG$.

 Let us suppose that we have a proper action of a  Lie group $\bG$ on a \emph{connected} manifold $Y$.
 According to a fundamental result \cite{DK},
 there exists a unique collection of pairwise conjugate isotropy groups, denoted  $(L)_{\princ}$,
 such that for every $y \in Y$ there exists $L\in (L)_{\princ}$ having the subgroup property $L < \bG_y$.
 The subset $Y_0 \subset Y$ of elements whose isotropy groups belong to $(L)_{\princ}$ is a \emph{dense, open} submanifold of $Y$.
The elements $\bG_y \in (L_{\princ})$ are called principal isotropy groups, $Y_0$ is
 called the submanifold of principal isotropy type (also called submanifold of principal orbit type), and
  the $\bG$-orbits contained in $Y_0$ are called principal orbits.
  Intuitively speaking, $(L)_{\princ}$ contains the generic isotropy groups that occur for the $\bG$-action.
  Since there is only a single isotropy type for the $\bG$-action on $Y_0$, the quotient space $Y_0/\bG$ can be made into
 a smooth manifold in such a way that the canonical projection $\pi_0: Y_0 \to Y_0/\bG$ is a smooth submersion.
 Moreover, the canonical projection $\pi: Y \to Y/\bG$ is continuous and open with respect to the
  quotient topology\footnote{It is well-known that $Y/\bG$  has a
 stratification induced by
 the various orbit types, but we shall not use this.}, and $\pi(Y_0)$ is a dense,
  open, \emph{connected} subset of the orbit space $Y/\bG$.
 If $\bG$ and $Y$ both  are connected, then $Y_0$ is also connected.

 \begin{lem}\label{lem:G1}
 Consider a smooth, proper action of a Lie group $\bG$ on a connected smooth manifold $Y$.
 Choose an arbitrary point $y$ of the submanifold $Y_0$ of principal isotropy type and denote $D:=\dim(Y_0/\bG)$.
  Then, there exist $D$ elements of $C^\infty(Y)^\bG$   that   are independent at $y$.
 \end{lem}
 \begin{proof}
 Consider a coordinate system $(U, x_1,\dots, x_D)$ on the manifold $Y_0/\bG$  that is centered at $[y]$ and the range of the coordinates is the open ball $B_1(0)\subset \bR^D$ of radius $1$.
 As is standard, we can construct $D$ functions $f_i\in C^\infty(Y_0/\bG)$ that coincide with the coordinate functions $x_i$  on the subset of $U$ corresponding to $B_{1/2}(0)$ and vanish identically
 on the complement in $Y_0/\bG$ of the subset corresponding to the closure of  $B_{3/4}(0)$.
Then,  we extend the functions $f_i$ to functions $F_i$ on $Y/\bG$ that are identically zero on the complement of $Y_0/\bG$ in $Y/\bG$.
 Finally, with the  canonical projection $\pi: Y \to Y/\bG$, we introduce the functions $F_i \circ \pi$ ($i=1,\dots, D$) and show that they verify the properties claimed by the lemma.

 It is obvious that $F_i\circ \pi$ is smooth at every point of $Y_0$ and its support   is contained in $Y_0$.
 It is also plain that these $D$ functions are $\bG$-invariant and are independent at the arbitrarily chosen point $y\in Y_0$ that we started with.
 If $\hat y \in Y\setminus Y_0$, then there exists a coordinate system $(V, \xi_1,\dots, \xi_n)$ (with $n:= \operatorname{dim}(Y)$) so that $\hat y \in V$ and
 $V\cap \operatorname{supp}(F_i\circ \pi) = \emptyset$. This holds  since $\operatorname{supp}(F_i\circ \pi)$ is a closed subset of $Y$, which follows from the fact
 that it is a closed subset of $Y_0$.
 (For the existence of $V$, note that every manifold topology is metrizable, and in a metric space there is a positive distance between a closed set
 and a point outside it.)
 The upshot is that the function $F_i\circ \pi$ is identically zero on $V$, and is thus smooth at $\hat y$.
 \end{proof}

Let us now suppose that the Lie group $\bG$ acting properly on the connected manifold $Y$ is
a direct product
\be
\bG = \Gone \times \Gtwo.
\label{bG12}
\ee
It is plain that the restricted actions of the factors $\Gone$ and $\Gtwo$ are also proper.
Moreover, if $Y_0$, $Y_0^1$ and $Y_0^2$ denote the principal isotropy type submanifolds
for the actions of $\bG$, $\Gone$ and $\Gtwo$, respectively, then
\be
Y_0 \subset Y_0^1 \cap Y_0^2.
\ee
This follows since the respective isotropy groups satisfy
\be
\Gone_y \times \Gtwo_y < \bG_y, \qquad \forall y\in Y.
\ee
Thus, if $\bG_y$ is conjugate to a subgroup of $\bG_z$, then $\Gone_y \times \Gtwo_y$ is
conjugate to  a subgroup of $\Gone_z \times \Gtwo_z$, for any $y,z\in Y$.
Notice that the action of $\Gtwo$  preserves all of the three dense open submanifolds
$Y_0^1$, $Y_0^2$ and $Y_0$ of $Y$. The same holds for the actions of $\Gone$ and $\bG$.
For later purpose, note also that
\be
Y_0/\Gone \subset Y_0^1/\Gone
\ee
\emph{is a dense open submanifold}.

Equality does not hold in the above relations in general. For example, take
$\bG = \operatorname{SU}(n) \times \operatorname{SU}(n)$ acting on $Y = \operatorname{SU}(n)$
according to the rule $(g_1,g_2,y) \mapsto g_1 y g_2^{-1}$.
In this example, the $\Gone$- and $\Gtwo$-actions are free, thus $Y_0^1 = Y_0^2 = Y$.
In contrast, the principal isotropy groups for the $\bG$-action
 are given by the maximal tori of $\operatorname{SU}(n)$ diagonally embedded into
the direct product, and  $Y_0$ is the set of regular elements having $n$ distinct eigenvalues.

\begin{lem}\label{lem:G2}
Suppose that we have a smooth  proper action of $\bG = \Gone \times \Gtwo$ on the connected manifold $Y$.
Let $Y_0$ and  $Y_0^1$ denote the principal isotropy type submanifolds
for the actions of $\bG$ and of the subgroup $\Gone$, respectively.
Then, the action of the subgroup $\Gtwo$ descends to smooth  and proper actions on the quotient manifolds
$Y_0/\Gone$ and $Y_0^1/\Gone$. If the principal isotropy groups of the $\bG$-action are of the form
$\bG_y = \Gone_y \times \{e_2\}$, then
the induced $\Gtwo$-action on $Y_0/\Gone$ is free.
\end{lem}
\begin{proof}
The action of $(g_1,g_2) \in \Gone\times \Gtwo$ is given by $(g_1,g_2) y = g_1(g_2 y) = g_2 (g_1 y)$ for any $(g_1,g_2,y) \in \bG \times Y$,
and the action of $\Gtwo$ on $Y_0^1/\Gone$ is given by  $g_2 [y]:= [ g_2 y]$, for any $[y] \in Y_0^1/\Gone$.
Because the actions of the two factors commute, the latter $\Gtwo$-action is smooth.
 To show that it is a proper action, consider convergent sequences satisfying
\be
\lim_{n\to \infty} [y_n] = [y_*] \quad \hbox{and} \quad
\lim_{n\to \infty} g_{2,n} [y_n] = [g_{2,n} y_n] =  [\hat y].
\ee
By the local triviality of the bundle $Y_0^1 \to Y_0^1/\Gone$, we can choose the representatives so that $y_n$ converges to $y_*$,
and we can also choose a sequence $g_{1,n}$ in $\Gone$ so that $g_{1,n} g_{2,n} y_n$ converges to $\hat y$.
Then, by the properness of the original $\bG$-action, there must exist a convergent subsequence of $g_{1,n} g_{2,n}$ in $\bG$.
Because $\bG$ is a direct product group, the corresponding subsequences of $g_{1,n}$ in $\Gone$ and $g_{2,n}$ in $\Gtwo$ both are convergent.
Therefore, the induced $\Gtwo$-action is proper; and it clearly restricts to a smooth and proper action on the dense
open submanifold $Y_0/\Gone \subset Y_0^1/\Gone$.

Next, we show that the $\Gtwo$-action on $Y_0/\Gone$  is free if for any $y\in Y_0$ the relation $\bG_y = \Gone_y \times \{e_2\}$ holds.
For this, notice that the equality $g_2 [y] = [y]$ requires the existence of an element $g_1\in \Gone$ so that $g_2 y = g_1 y$.  This is equivalent to
$g_1^{-1} g_2 y =y$, which now implies that $g_2$ is the unit element of $\Gtwo$.
\end{proof}

The above proof was inspired by a similar proof of Lemma 4.1.2 in \cite{MMOPR}.
We remark in passing that the existence of $y\in Y$ for which
$\bG_y = \Gone_y \times \{e_2\}$ entails
that the $\Gtwo$-action is free on the dense open submanifold $Y_0^2$ of $Y$.

 \begin{lem}\label{lem:G3}  Suppose that we have commuting actions of two Lie groups $\Gone$ and $\Gtwo$ on
 a connected manifold $Y$ and $\Gone$ is compact.
Then, the combined action of $\bG := \Gone \times \Gtwo$ is proper if and only if the action of $\Gtwo$ is proper.
 \end{lem}
 \begin{proof}
   Consider sequences $y_n$  in $Y$ and $(g_{1,n}, g_{2,n})$  in $\bG$  for which $y_n$ and
 $(g_{1,n}, g_{2,n}) y_n$  both are convergent.  Since $\Gone$ is compact, we may assume that $g_{1,n}$ is a convergent sequence.
 Then $\hat y_n:= g_{1,n} y_n$ is convergent and $g_{2,n} \hat y_n$ is also convergent.
 Thus, if the $\Gtwo$-action is proper, then we can choose a convergent subsequence of $(g_{1,n}, g_{2,n})$, guaranteeing the properness of the $\bG$-action.
Of course, the $\Gtwo$-action is proper if the $\bG$-action is proper.
\end{proof}

If a  smooth proper action preserves an embedded submanifold, then the restricted action on the submanifold is obviously
smooth and proper.
This remains true for the so called initial submanifolds \cite{Mi,OR}, which contain also certain immersed submanifolds, including
the symplectic leaves of smooth  Poisson manifolds and leaves of integrable singular distributions  in general.
The statement of the next lemma, that we shall need,  is a direct consequence of \S2.3.7 and Theorem 4.1.28 in \cite{OR}.

\begin{lem}\label{lem:G4}
Suppose that we have a smooth proper action of a Lie group on a Poisson manifold and
this action preserves a symplectic leaf. Then, restriction yields a smooth proper action on
the symplectic leaf in question.
\end{lem}

\subsection{Integrable systems from generalized Poisson reduction} \label{ss:genIS}

The assumptions that we make below
may appear complicated at first sight, but we shall see later
that they hold in interesting examples.

We begin with a \emph{connected} smooth manifold $Y$ equipped with a smooth bivector $P_Y$
and a smooth \emph{proper} action of a \emph{connected} Lie
group  of the form
\be
\bG = \Gone \times \Gtwo
\quad \hbox{with}\quad
\Gtwo=  \mathrm{U}(1)^{\ell_1} \times \bR^{\ell_2}  \quad\hbox{for some}\quad \ell_1, \ell_2\in \bZ_{\geq 0},\, \ell:=\ell_1 + \ell_2 >0.
\label{G1G2}\ee
The bivector is used to define the `bracket' of any two functions $F,H\in C^\infty(Y)$
as well as the vector field $V_H$ by the formula
\be
(dF \otimes dH)(P_Y) =: \{ F, H\} =: dF(V_H).
\ee
This bracket is anti-symmetric and satisfies the Leibniz property, but is not necessarily a Poisson bracket.
The restriction of the bivector defines a bracket for functions defined on any open subset of $Y$, too.

\begin{assume}\label{assume:G5}
 We require the following properties:
\begin{enumerate}[label=\alph*$)$]
\item \label{itG5:1}
For any $\Gone$-invariant open submanifold $Y'\subseteq Y$, $C^\infty(Y')^{\Gone}$ becomes a Poisson
algebra with the bracket $\{-,-\}$ defined by  the bivector $P_Y$ and its corresponding restrictions.
\item \label{itG5:2}
Let $T_i^Y$ $(i=1,\dots, \ell)$ denote the vector fields on $Y$ associated with the natural basis $T_i$
of the Lie algebra of $\Gtwo$ \eqref{G1G2}. We suppose that we have $\ell$ functions $H_i \in C^\infty(Y)^{\Gone}$ satisfying
\be
T_i^Y = V_{H_i}
\quad \hbox{and}\quad
\{ H_i, H_j\} =0, \qquad \forall i,j=1,\dots, \ell.
\label{Hicond}\ee
\item \label{itG5:3}
The principal isotropy group for the $\Gtwo$-action on $Y$ is trivial (this action is effective)
 and on the principal isotropy type submanifold $Y_0$ of the
$\bG$-action we have $\bG_y = \Gone_y \times \{e_2\}$.
\end{enumerate}
\end{assume}

\begin{rem}\label{rem:G6}
In general, assuming property \ref{itG5:1} on $Y$ may imply that
it is also  valid on any $\Gone$-invariant open submanifold $Y'$.
This property holds automatically, for example, if $(Y,P_Y)$ is a Poisson manifold with a $\Gone$-invariant Poisson bivector.
In the presence of the first equality in equation \eqref{Hicond}, the second equality is equivalent
to the $\Gtwo$-invariance of the  functions $H_i$.
\end{rem}

\begin{thm}\label{thm:G7}
Let us consider a connected, proper $\bG$-manifold $Y$ equipped with a bivector $P_Y$
subject to Assumption \ref{assume:G5}.
Let $H_i$ ($i=1,\dots, \ell$) be functions satisfying \eqref{Hicond} and
denote by $Y_0^1$ and $Y_0$ the principal isotropy type submanifolds of the $\Gone$- and $\bG$-actions,
respectively.
Then, the natural identifications
\be
C^\infty(Y_0^1/\Gone) \simeq C^\infty(Y_0^1)^{\Gone}
\quad \hbox{and} \quad C^\infty(Y_0/\Gone) \simeq C^\infty(Y_0)^{\Gone}
\label{209}\ee
induce Poisson structures on the smooth manifolds $Y_0^1/\Gone$ and  $Y_0/\Gone$.
The $\Gtwo$-action on $Y_0^1$ descends to a  smooth, proper and effective Hamiltonian action
on the connected Poisson manifold $Y_0^1/\Gone$, which is generated by the momentum map
\be
(\cH_1,\dots, \cH_\ell): Y_0^1/\Gone \to \bR^\ell
\label{210}\ee
defined by the relation $\cH_i \circ \pi_0^1 = H_i $ on $Y_0^1$, with the  projection
$\pi_0^1: Y_0^1 \to Y_0^1/\Gone$.
This  `reduced $\Gtwo$-action' restricts to a smooth, proper and free Hamiltonian action on the connected, dense open submanifold
$Y_0/\Gone \subset Y_0^1/\Gone$.
\end{thm}
\begin{proof}
It follows from Lemma \ref{lem:G2} that we obtain a smooth  proper action of  $\Gtwo$ on $Y_0^1/\Gone$.
The generating vector fields of this  action  are the projections of the vector fields $T_i^Y$ \eqref{Hicond} restricted on $Y_0^1$.
Item \ref{itG5:2} of  Assumption \ref{assume:G5} and the definition of the Poisson bracket on  $C^\infty(Y_0^1/\Gone)$ guarantee that
these generating vector fields are the Hamiltonian vector fields of the component functions $\cH_i$ in \eqref{210}.
The same arguments prove that the restricted action on the dense open submanifold $Y_0/\Gone \subset Y_0^1/\Gone$ is Hamiltonian.
The action of $\Gtwo$ on $Y_0/\Gone$ is free  by Lemma \ref{lem:G2} and item \ref{itG5:3} of  Assumption \ref{assume:G5}, thus it is effective on $Y_0^1/\Gone$.
 \end{proof}

Below, the restriction of the function  $\cH_i$  onto $Y_0/\Gone \subset Y_0^1/\Gone$ will be denoted by the same letter.

\begin{cor}\label{cor:G8}
Consider the  smooth, proper and free  Hamiltonian action of $\Gtwo$ \eqref{G1G2} on $Y_0/\Gone$ and its
associated momentum map $(\cH_1,\dots, \cH_\ell)$. This restricts to a smooth, proper and free Hamiltonian
action on every symplectic leaf $S \subset Y_0/\Gone$.
Consequently, at every point $[y] \in Y_0/\Gone$ there exists $\dim(Y_0/\Gone) - \ell$ functionally
independent elements of $C^\infty(Y_0/\Gone)^{\Gtwo}$, and
at every point $[y] \in S$ there exist
$\dim(S) - \ell$ functionally independent elements of $C^\infty(S)^{\Gtwo}$.
This means that the Hamiltonians $(\cH_1,\dots, \cH_\ell)$ yield
an  integrable system of rank $\ell$ on the Poisson manifold $Y_0/\Gone$ and on
its every symplectic leaf.
\end{cor}
\begin{proof}
We have a smooth, proper and free Hamiltonian action of $\Gtwo$ on $Y_0/\Gone$ by Theorem \ref{thm:G7}, and see from Lemma  \ref{lem:G4}
that the restricted free Hamiltonian action remains  smooth and proper on every symplectic leaf $S$. The claims about the functionally independent
invariants are direct consequences of Lemma \ref{lem:G1} applied to the pertinent actions of
$\Gtwo$, since $\dim(\Gtwo) = \ell$.
These invariants are the joint constants of motion for the Hamiltonians $\cH_i$ \eqref{210} and for their restrictions on $S$.
They automatically form Poisson subalgebras in $C^\infty(Y_0/\Gone)$ and in $C^\infty(S)$, respectively.
In fact, the invariants in $C^\infty(Y_0/\Gone)$ form the commutant  of the Abelian Poisson algebra of functional dimension $\ell$ given by
\begin{equation*}
 \fH_{Y_0/\Gone}:=\vec{\cH}^\ast\bigl(C^\infty(\bR^\ell)\bigr)\,, \quad \vec{\cH}:=(\cH_1,\dots, \cH_\ell),
\end{equation*}
 and the flow of  any  $\Psi\in \fH_{Y_0/\Gone}$ is complete. We can write
 $\Psi=\psi\circ (\cH_1,\ldots,\cH_\ell)\circ \pi_0$ for some $\psi \in C^\infty(\bR^\ell)$, and the corresponding
 integral curve starting at $\pi_0(y_0)\in Y_0/\Gone$ is provided by
 the $\pi_0$-projection\footnote{Here, $\pi_0: Y_0\to Y_0/\Gone$ is the canonical projection;
 the same method gives the integral curves
 on $Y_0^1/G^1$.}
 of the curve
\be
\tau \mapsto \exp_{\Gtwo}(\tau \, T_\psi(y_0))\, y_0,
\qquad T_\psi(y_0):=\sum_{j=1}^\ell \,\partial_j \psi(H_1(y_0),\ldots,H_\ell(y_0))\, T_j \in \operatorname{Lie}(\Gtwo)
\ee
for an arbitrary lift $y_0\in Y_0$ of $\pi_0(y_0)$, cf. item \ref{itG5:2} of Assumption \ref{assume:G5}.
% $Y_0/\Gone$.
The analogous statement regarding the symplectic leaves of $Y_0/\Gone$ is verified in the same way.
Therefore, according to Definition \ref{defn:int}, we obtain an integrable system on the Poisson manifold $Y_0/\Gone$ and
on its every symplectic leaf.
\end{proof}

\begin{cor}\label{cor:G9}
The proper Hamiltonian action of $\Gtwo$ \eqref{G1G2} is generically free on the Poisson manifold $Y_0^1/\Gone$
and is also generically free on every such  symplectic leaf $\mathcal{S} \subset Y_0^1/\Gone$
that has non-empty intersection with $Y_0/\Gone$.
This implies the following formulae for the functional dimensions of invariants
\be
\operatorname{ddim}\left( C^\infty(Y_0^1/\Gone)^{\Gtwo}\right) = \dim\left(Y_0^1/\Gone\right) - \ell
\quad \hbox{and}\quad
\operatorname{ddim}\left( C^\infty(\mathcal{S})^{\Gtwo}\right) = \dim\left(\mathcal{S}\right) - \ell,
\label{211} \ee
 for every leaf for which $\mathcal{S}\cap Y_0/\Gone \neq \emptyset$.
 This means that the Hamiltonians $(\cH_1,\dots, \cH_\ell)$  define
 integrable systems of rank $\ell$ on $Y_0^1/\Gone$ and
 on its symplectic leaves with the specified property.
 In contrast to the case of $Y_0/\Gone$, these systems can have fixed points.
\end{cor}
\begin{proof}
By the same token as in the proof of Corollary \ref{cor:G8},
we have a proper Hamiltonian action of $\Gtwo$ on $Y_0^1/\Gone$ and on any of its symplectic leaves $\cS$, for which
the functions $\cH_i$ \eqref{210} and their restrictions yield the momentum map.
The principal isotropy group is just the identity subgroup of $\Gtwo$ both for the action on $Y_0^1/\Gone$ and for its
restriction on those symplectic leaves that intersect $Y_0/\Gone$.
Consequently,  Lemma \ref{lem:G1} implies the equalities \eqref{211}.
The completeness of the flows of the Hamiltonian vector fields of the functions  of the momentum map \eqref{210}
is seen similarly as in the proof of Corollary \eqref{cor:G8}.
\end{proof}

\begin{rem}\label{rem:G10}
We obtain a proper Hamiltonian action of $\Gtwo$ \eqref{G1G2} on every symplectic leaf $\mathcal{\mathcal{S}}$ of
$Y_0^1/\Gone$.   If the dimension of the pertinent principal $\Gtwo$-orbits in $\mathcal{\mathcal{S}}$ is $0< d\leq \ell$, this gives an integrable
system of rank $d$ associated with $d$ linear combinations of the Hamiltonians $(\cH_1,\dots, \cH_\ell)$, whose
Hamiltonian flows engender the principal orbits in question.
\end{rem}

In Sections \ref{Sec:Cot}, \ref{Sec:Heis} and \ref{Sec:qPoiss} we will present
interesting realizations of the following scenario.

\begin{scen}\label{scen:G11}
We have a manifold $M$ equipped with a bivector $P_M$ and  an action of
a  connected \emph{compact} Lie group $\Gone$ fitting into  one of the following three types:
\begin{itemize}
\item $(M,P_M)$ is a Poisson manifold with a $\Gone$-invariant Poisson bivector.
\item $(M,P_M)$ is a Poisson manifold, $\Gone$ is a Poisson--Lie group and the action map $\Gone\times M \to M$ is Poisson.
\item  $(M, P_M, \Gone)$ is a quasi-Poisson manifold in the sense of \cite{AKSM}.
\end{itemize}
The following four conditions hold:
\begin{enumerate}[label=\arabic*)]
\item \label{itScen:1}
The manifold $M$ carries an Abelian Poisson subalgebra $\fH\subset C^\infty(M)^{\Gone}$
for which the vector fields $V_\cH$, $\cH\in \fH$ are all complete.
\item \label{itScen:2}
 On a dense, open, connected $\Gone$-invariant submanifold $Y\subset M$ we have
$\ell$ smooth $\Gone$-invariant functions $H_1,\dots, H_\ell \in  C^\infty(Y)^{\Gone}$ that are independent at every point of $Y$,
satisfy  $\{H_i, H_j\}=0$ for all $1\leq i,j\leq \ell$,  and
 the vector fields $V_{H_i}$ (defined with the aid of the restriction
$P_Y$ of $P_M$ to $Y$) generate a  proper and effective, smooth action of a Lie group $\Gtwo= \mathrm{U}(1)^{\ell_1} \times \bR^{\ell_2}$ ($ \ell_1 + \ell_2 =\ell$)
on $Y$.
\item \label{itScen:3}
The restriction of the elements of $\fH$ on $Y$ can be expressed in terms of $H_1,\dots, H_\ell$
and the span of the exterior derivatives of the elements of $\fH$ coincides
with the span of the exterior derivatives $dH_1,\dots, dH_\ell$ at every point of $Y$.
\item \label{itScen:4}
Denote by $Y_0\subset Y$ the principal isotropy type submanifold for the combined action of the direct product
group $\bG = \Gone \times \Gtwo$ on $Y$. Then, for any $y \in Y_0$ the isotropy group $\bG_y$ is equal to  $\Gone_y \times \{e_2\}$.
\end{enumerate}
\end{scen}

 A key feature of Scenario \ref{scen:G11} is that it ensures the validity of the previous
Assumption \ref{assume:G5}. This is obvious if $(M,P_M)$ is a Poisson manifold,  and directly follows
from the definitions for Poisson--Lie actions \cite{STS} and for quasi-Poisson manifolds \cite{AKSM} as well.
In particular, the  $G^1$- and $G^2$-actions  commute since vector fields $V_{H_i}$ are $G^1$-invariant because $H_i \in C^\infty(Y)^{G^1}$.
It is also easily seen for all three types of cases that
the algebra of first integrals,
\be
\fF:= \{ \cF  \in C^\infty(M) \mid \{ \cF, \cH \} = 0, \,\, \forall \cH \in \fH\},
\label{fF213}\ee
is closed under the bracket defined by $P_M$.

\begin{prop}\label{prop:intM}
Under Scenario \ref{scen:G11},  the quadruple $(M, P_M, \fH, \fF)$ represents an integrable system of rank $\ell$
in the sense of  Definition \ref{defn:int} (also extended to quasi-Poisson manifolds).
\end{prop}
\begin{proof}
It is immediate from conditions \ref{itScen:1}, \ref{itScen:2} and \ref{itScen:3} in Scenario \ref{scen:G11} that $\fH$ satisfies the requirements  in  Definition \ref{defn:int},
so it only remains to demonstrate that
\be
\operatorname{ddim}(\fH) + \operatorname{ddim}(\fF) = \dim(M).
\label{214}\ee
Using also condition \ref{itScen:4}, the Scenario ensures that we have a proper effective action of the group $\Gtwo$  on the connected
dense open submanifold $Y\subset M$.
Recall that $Y_0^2\subset Y$  denotes the dense open subset of principal orbit type for the $\Gtwo$-action.
We know from  Lemma \ref{lem:G1}, applied  to $\Gtwo$ instead of $\bG$, that $C^\infty(Y)^{\Gtwo}$ has functional dimension
$\dim(Y_0^2)-\dim(\Gtwo) = \dim(M) - \ell$.
Moreover,  the proof of Lemma \ref{lem:G1} shows that
there exist
$\dim(M)-\ell$ elements of  $C^\infty(Y)^{\Gtwo}$ that are independent at an arbitrarily chosen point
$y\in Y_0^2$ and their support lies in a proper open  subset of $Y_0^2$.
These functions can be smoothly  extended to $M$ by setting them to zero on $M\setminus Y$.
Condition \ref{itScen:3} of Scenario \ref{scen:G11} implies that the resulting  functions belong to $\fF$ \eqref{fF213}.
Since $Y_0^2$ is dense and open also in $M$, the equality \eqref{214} is proved.
\end{proof}

Comparing  Definition \ref{defn:genact} with  Scenario \ref{scen:G11},
we may now say that the functions $H_1,\dots, H_\ell$ are
 \emph{generalized action variables} on $Y\subset M$ for the integrable system engendered by the Abelian Poisson algebra $\fH$.
In the examples presented later, the functions $H_i\in C^\infty(Y)$ can be extended to continuous functions on $M$,
but not to smooth functions
(essentially since they arise from eigenvalues of matrices).

\medskip

To proceed,
we note that  the $\Gone$ isotropy group is constant along every integral curve
of every $F\in C^\infty(M)^{\Gone}$ for any of the three types of $(M,P_M)$ in Scenario \ref{scen:G11}.
This is  easily demonstrated and is well known.
In particular, the flows of $\fH$  can be restricted to each isotropy type submanifold
of the $\Gone$-action on $M$ and it follows from conditions \ref{itScen:2} and \ref{itScen:3} in Scenario \ref{scen:G11}
that they can also be  restricted on every submanifold in the chain
\be
Y_0 \subset Y_0^1 \subset Y \subset M.
\label{chainI}\ee
(We remind that $Y_0^1\subset Y$ is the principal orbit type submanifold of $Y$ for the $G^1$-action
and $Y_0$  appears in item \ref{itScen:4} of Scenario \ref{scen:G11}.)
Plainly,  $Y_0^1$ is also contained in the principal orbit type submanifold  of $M$ with respect to the $G^1$-action,
which we henceforth  denote by $M_*$.
Indeed, there must exist an element $z\in Y_0^1 \cap M_*$, since  $M_*$ and $Y_0^1$ both are dense open subsets of $M$.
But $Y_0^1$ and $M_*$ are the subsets of points of $Y$ and $M$, respectively,  whose isotropy subgroups are
conjugate to $\Gone_z$.  Therefore, $Y_0^1 = Y \cap M_*$ holds.
A notable consequence is  that
\be
Y_0/\Gone \subset Y_0^1/\Gone \subset  M_*/\Gone
\label{proper?}\ee
is a chain of \emph{dense open} submanifolds. These are expected to be proper subsets in general.

Below, we focus on the `smooth component' of the quotient space  $M/\Gone$ given by the principal orbit type,
denoted as
\be
M_*^\red := M_*/\Gone.
\label{M*red}\ee
This is a Poisson manifold due to the natural identification
\be
C^\infty(M_*^\red) \equiv C^\infty(M_*)^{\Gone}.
\ee
Let $\fH_*$ be the  restriction of $\fH$ on $M_*$ and consider
its commutant  in the Poisson algebra $C^\infty(M_*)^{\Gone}$,
\be
\fF_*^{\Gone}:= \{ \cF  \in C^\infty(M_*)^{\Gone} \mid \{ \cF, \cH \}_* = 0, \,\, \forall \cH \in \fH_*\},
\label{215}\ee
where $\{-,-\}_*$ is the bracket arising from the bivector $P_M$ restricted on $M_*$.

\begin{lem}\label{lem:214}
In any realization of Scenario \ref{scen:G11}, we have the equality
\be
\operatorname{ddim}(\fH_*) + \operatorname{ddim}(\fF^{\Gone}_*) = \dim(M_*/\Gone).
\label{prop213}\ee
\end{lem}
\begin{proof}
Consider $\bG = \Gone \times \Gtwo$ and choose an arbitrary $y \in Y_0$.
In the proof of Lemma \ref{lem:G1} we constructed
$D:= \dim(Y_0/\bG)$ elements of $C^\infty(Y_0)^\bG$  that are independent at $y$ and are supported
inside a proper $\bG$-invariant open subset of $Y_0$.
A quick look at the proof shows that extending those functions by zero we obtain $\Gone$-invariant functions
on $M_*$. The extended functions belong to $\fF^{\Gone}_*$ \eqref{215}, because on $Y_0$ the span
of the  vector fields $V_\cH$ for $\cH\in \fH$ is the same as the span of vector fields
generating the $\Gtwo$-action (see item \ref{itScen:3} in Scenario \ref{scen:G11}).
In our case
\be
D = \dim(Y_0/\bG)=\dim(M_*/\Gone) - \ell = \dim(M_*^\red) - \ell,
\ee
and we obviously have $\ddim(\fH_*) = \ell$.
The equality \eqref{prop213} is proved since at any point $y$ from the dense open subset $Y_0$ of $M_*$
we exhibited $D$ independent elements  of $\fF^{\Gone}_*$,
and there cannot be more such elements because the restriction of any $\cF  \in C^\infty(M_*)^{\Gone}$ on $Y_0$ belongs
to $C^\infty(Y_0)^\bG$.
\end{proof}

The following theorem is the main result of the present paper.
After restriction on $M_*$,
it characterizes  the Hamiltonian reductions of the integrable systems obeying Scenario \ref{scen:G11}.

\begin{thm}\label{thm:G12}
Suppose that  Scenario \ref{scen:G11} is verified. Then,
the Abelian Poisson algebra $\fH$  descends to
an integrable system of rank $\ell$ on
the Poisson manifold $M_\ast^\red$ \eqref{M*red}.
The restrictions of this system to the  dense open submanifolds  $Y_0/\Gone$ and $Y_0^1/\Gone$ \eqref{proper?} of $M_*^\red$
possess generalized action variables given by\footnote{
To be precise, we here should use the projected Hamiltonians $(\cH_1,\ldots,\cH_\ell)$ satisfying $\cH_i\circ \pi_0^1=H_i$, cf. \eqref{210}. By abuse of notation, we identify
the functions $H_i$ and $\cH_i$, and we also adopt this convention in the next sections.}
$(H_1,\ldots,H_\ell)$, and the corresponding Hamiltonian $\Gtwo$-action  is free on $Y_0/\Gone$.
As a result, $\fH$ induces integrable systems of rank $\ell$ with generalized action variables given by the restriction of $(H_1,\ldots,H_\ell)$ on
\begin{itemize}
\item every symplectic leaf of $Y_0/\Gone$;
\item every such symplectic leaf of $Y_0^1/\Gone$ that intersects $Y_0/\Gone$.
\end{itemize}
\end{thm}
\begin{proof}
Since the map $M_* \to M_*/\Gone$ is a surjective submersion,
the rings of $\Gone$-invariant functions $\fH_*$ and $\fF_*^{\Gone}$ in \eqref{prop213} descend to
 rings of functions on  $M_*^\red$ \eqref{M*red} of the same functional dimensions,
  $\ell$ and  $\dim(M_*/\Gone) - \ell$.
 The restrictions of the Hamiltonian vector fields of $\fH$  on $M_*$ are  $\Gone$-invariant, and their projections
  yield
 the Hamiltonian vector fields of the reduced Hamiltonians on $M_*^\red$.     These vector field span an $\ell$-dimensional linear space
 at every point of $Y_0/\Gone$, since their span
  coincides with the span of the fundamental vector fields of the reduced $\Gtwo$-action,
  which  is free on $Y_0/\Gone$ by
 Theorem \ref{thm:G7}.  Thus, we conclude  that the reduction results in an integrable system on the Poisson manifold $M_*^\red$ \eqref{M*red}.

The conditions \ref{itScen:2} and \ref{itScen:3} in Scenario \ref{scen:G11} imply that the dense open submanifolds $Y_0/\Gone$ and $Y_0^1/\Gone$ \eqref{proper?} are invariant
with respect to the flows of the integrable system living on $M_*^\red$.
The remaining claims about generalized action variables and integrability on symplectic leaves
are then direct consequences of
Theorem \ref{thm:G7} and
 its  Corollaries \ref{cor:G8} and \ref{cor:G9}.
Because of conditions \ref{itScen:2} and \ref{itScen:3}, the $\Gtwo$-invariant functions mentioned in
the two corollaries are the joint constants of motion
for the reduced Hamiltonians that arise from the elements of $\fH$ on the relevant spaces.
The functional dimension of the ring of functions descending from $\fH$
is equal to $\ell$ on each of the pertinent symplectic leaves, since  the exterior derivatives of those functions span the
same linear space as do the components of the momentum map \eqref{210} that generates a proper and effective
action of $\Gtwo$ on $Y_0^1/\Gone$, which is a free action on $Y_0/\Gone$.
\end{proof}

The reductions of the integrable system  $(M,P_M,\fH)$ outside the principal orbit type $M_*\subset M$
should be studied as well, but  this  is outside
the scope of the present paper.
We merely mention in passing that the
following equality should be useful in such a study:
\be
\operatorname{ddim}(\fH) + \operatorname{ddim}(\fF^{\Gone}) = \dim(M_*/\Gone),
\label{212}\ee
where
\be
\fF^{\Gone}:= \fF \cap C^\infty(M)^{\Gone}
\label{212+}\ee
with $\fF$ in \eqref{fF213}.
This basically follows from the same arguments that led to  \eqref{prop213}.

In our general considerations,
 the constants of motion arose from the invariants for the action of the group $\Gone\times \Gtwo$.
 In specific examples,  one may wish to describe the Poisson algebra $\fF^{\Gone}$ \eqref{212+} and its reductions as explicitly as possible.
Finally, it is worth noting that
the $\Gone$-action on $M$ is usually generated by a momentum map, which is constant along the flows
of the elements of $C^\infty(M)^{\Gone}$.
As a  result,  at least generically, one can obtain symplectic leaves in the reduced phase space  by
fixing  the $\Gone$-invariant functions depending only on the momentum map to constant values.

\subsection{Action-angle coordinates from generalized Hamiltonian torus actions} \label{ss:AA}
Before proceeding with the examples, let us explain an auxiliary result about the existence of action-angle coordinates, that fits our mechanism.
This justifies the terminology of \emph{generalized action variables} used in Theorem \ref{thm:G12}.
Our proof is strongly inspired by \cite[Thm.~3.6]{LGMV}.

Hereafter, $C_r^k = (-r,r)^k$ denotes the open hypercube in $\bR^k$ centred at the origin, for which
all coordinates vary  between $-r$ and $r$:
\be
 C_r^k:= \{ x \in \bR^k \mid  -r < x_i < r,\,\, \forall i=1,\dots, k\}.
\ee

\begin{thm} \label{thm:AA}
Assume that $(\cM,P_{\cM})$ is a connected smooth Poisson manifold endowed with a proper and effective, smooth Hamiltonian action of the Lie group
$\Gtwo=  \mathrm{U}(1)^{\ell_1} \times \bR^{\ell_2}$  where $\ell_1, \ell_2\in \bZ_{\geq 0}$, $\ell:=\ell_1 + \ell_2 >0$.
Write $(P_1,\ldots,P_\ell):\cM\to \bR^\ell$ for a corresponding momentum map, and set $d=\dim(\cM)$.
Take any $y_0 \in \cM$ with trivial isotropy group and put $p_i := P_i - P_i(y_0)$.
Then,  there exist a $\Gtwo$-invariant open neighbourhood $\cU\subset \cM$ around $y_0$ and functions
\be
\theta_1,\ldots,\theta_{\ell},z_1,\ldots,z_{d-2\ell}:\cU\to\bR,
\label{app:AA1}
\ee
such that
\begin{enumerate}
 \item the functions $(e^{\ri \theta_{1}},\ldots,e^{\ri \theta_{\ell_1}},\theta_{\ell_1+1},\ldots,\theta_{\ell},p_1,\ldots,p_\ell,z_1,\ldots,z_{d-2\ell})$ define a diffeomorphism
 $\cU \stackrel{\sim}{\longrightarrow} \Gtwo \times C^{d-\ell}_\epsilon$ for some $\epsilon >0$, and $y_0$ corresponds to $(e,0)$;
 \item the Poisson structure can be written in terms of these coordinates as
 \be
P_{\cM}\big|_{\cU} = \sum_{i=1}^{\ell} \frac{\partial}{\partial \theta_i} \wedge \frac{\partial}{\partial p_i}
+ \sum_{\substack{a,b=1\\a<b}}^{d-2\ell} f_{ab}(z)\,\frac{\partial}{\partial z_a} \wedge \frac{\partial}{\partial z_b}\,,
 \label{app:AA2}
 \ee
 for some smooth functions $f_{ab}$ depending only on $z_1,\ldots,z_{d-2\ell}$.
\end{enumerate}
\end{thm}

\begin{proof}
Let $\ccM$ denote the submanifold of principal orbit type for the $G^2$-action.
Recall that $\pi: \ccM \to \ccM/G^2$ is a locally trivial principal fiber bundle.
By choosing a suitable open subset $\varSigma \subset \ccM/G^2$ and a local section over $\varSigma$
that passes through $y_0$,
we can trivialize the restricted bundle.
Thus, we obtain a  $G^2$-equivariant `trivializing diffeomorphism'
\be
T: \cW := \pi^{-1}(\varSigma) \to G^2 \times \varSigma,
\label{AC3}\ee
where $G^2$ acts by left-multiplications of the first factor of $G^2 \times \varSigma$.
Since the  momentum map is $G^2$-invariant, its shifted components  can be regarded as functions on $\varSigma$,  still denoted as $p_i$.
By choosing $\varSigma$, we can extend these $\ell$ functions to a coordinate system on $\varSigma$ that
yields a diffeomorphism
\be
\psi: \varSigma \to  C_r^{d-\ell},
\label{AA4}\ee
in such a way
that $\psi(\pi(y_0))$ is
the center of the cube.
Let us call the resulting  coordinates  $(p_i, x_a)$.
Let us also introduce coordinates $\vartheta_1,\ldots,\vartheta_\ell$ on $\Gtwo$
by means of the presentation
\be
\Gtwo=\mathrm{U}(1)^{\ell_1} \times \bR^{\ell_2}
=\{(e^{\ri \vartheta_{1}},\ldots,e^{\ri \vartheta_{\ell_1}},\vartheta_{\ell_1+1},\ldots,\vartheta_{\ell})\}\,.
\label{AAtriv}
\ee
By combining the above, we arrive at  the diffeomorphism
\be
\Psi = \id \times \psi: G^2 \times \varSigma \to \mathrm{U}(1)^{\ell_1} \times \bR^{\ell_2}  \times C_r^{d-\ell},
\label{AA6}\ee
whereby $\Psi\circ T$  equips $\cW$ with the coordinates\footnote{Strictly speaking, one should cover each
$\mathrm{U}(1)$ factor with two coordinate patches using angle coordinates that agree modulo $2\pi$ on the overlap.
 One can check that the relations \eqref{app:AA2} and \eqref{AA3} are then valid in all of the resulting coordinate systems.}
$(\vartheta_i, p_j, x_a)$,
and $y_0\in \cW$ corresponds to $(e,0)\in \Gtwo \times C_r^{d-\ell}$.

Since $(p_1,\ldots,p_\ell)$ is a momentum map and the $(p_j,x_a)$ are $\Gtwo$-invariant, one sees that the Poisson structure restricted to $\cW$ must satisfy
\begin{equation}
 \begin{aligned}
  \{\vartheta_i,p_j\}&=\delta_{ij}, \qquad \{p_i,p_j\}=0, \quad &&\{p_i,x_a\}=0, \\
\{\vartheta_i,\vartheta_j\}&=g_{ij}, \qquad \{\vartheta_i,x_b\}=h_{ib}, \quad &&\{x_a,x_b\}=F_{ab},
\label{AA3}
 \end{aligned}
\end{equation}
for some smooth functions $F_{ab},g_{ij},h_{ib}$ depending only on the coordinates $(p_j,x_a)$.

By the generalization of the Carathéodory--Jacobi--Lie theorem given in \cite[Thm.~2.1]{LGMV},
we can find a coordinate system on a (not necessarily $\Gtwo$-invariant)
open neighbourhood $\cV\subset \cW$ of $y_0$
such that the Poisson structure takes the form \eqref{app:AA2}.
These coordinates can be collected into a diffeomorphism
\be
\Phi: \cV \to C_{r_1}^d,
\label{AA8}\ee
for some $0<r_1<\pi$,
in such a manner that  $\theta_i, p_j, z_a$ denote the  component functions of $\Phi$,   $\Phi(y_0)=0 \in \bR^d$,
and the $p_i$ represent the  shifted momentum map restricted on $\cV$.

We can choose an open subset
\be
\cU' \subset  \cV \subset \cW
\label{AA9}\ee
whose image by  the map $\Psi\circ T$ is a  hypercube $C_{r_2}^d$,
with some $0< r_2 < \mathrm{min}(r,\pi)$.
Then, we have two coordinate systems on $\cU'$, and the change of coordinates gives
a diffeomorphism
\be
 \Phi \circ (\Psi \circ T)^{-1}:   C_{r_2}^d \to \cD,
 \label{AA10}\ee
 where $\cD$ is some open subset of $C_{r_1}^d = \mathrm{Im}(\Phi)$.
In terms of the respective coordinates, this can be described  as  a  map
\be
(\vartheta_i, p_j, x_a) \mapsto (\theta_k, p_m, z_b).
\label{AA11}\ee
By the construction of the two coordinate systems on $\cU'$,  if we express $(\theta_n, p_m, z_b)$ as functions
of  $(\vartheta_i, p_j, x_a)$, then the fundamental Poisson brackets in \eqref{AA3} must imply
the Poisson brackets of the form encoded by the Poisson tensor in \eqref{app:AA2}.
Observing that $\{\theta_i - \vartheta_i, p_j\} =0$ on $\cU'$, we see that
\be
\theta_i = \vartheta_i + F_i(\vec{p}, \vec{x})
\label{AA12}\ee
must hold for the change of coordinates map \eqref{AA11}, with some smooth functions $F_i$.
 Moreover, because $\{z_a, p_i\}=0$ by \eqref{app:AA2}, we must have
\be
z_a = \cZ_a(\vec{p}, \vec{x})
\label{AA13}\ee
for some smooth functions $\cZ_a$.
Because the change of coordinates is a diffeomorphism,  we may conclude that after
restriction on $C_{r_3}^{d-\ell}$ for some $0< r_3 \leq r_2$,
 the map
\be
\zeta: (\vec{p}, \vec{x}) \mapsto (\vec{p}, \vec{\cZ}(\vec{p}, \vec{x}))
\label{AA14}\ee
is a diffeomorphism from $C_{r_3}^{d-\ell}$ onto some open neighbourhood
of the origin in $\bR^{d-\ell}$, which we denote $\cD_{\proj}'$.
(We here used that the Jacobian determinant of the map \eqref{AA11} is nowhere zero.
Note also that $\cD_{\proj}'$ is an open subset of the obvious projection on $\bR^{d-\ell}$ of $\cD$ in \eqref{AA10}.)

By simply forgetting the restrictions $- r_2 < \vartheta_j < r_2$  on the coordinates on $\cU'$,
the formulae \eqref{AA12}, \eqref{AA13}  give rise to a diffeomorphism
 \be
G^2 \times C_{r_3}^{d-\ell} \to G^2 \times \cD_{\proj}'.
\label{AA15}\ee
In detail, this map operates according to
\be
(e^{\ri \vartheta_{1}},\ldots,e^{\ri \vartheta_{\ell_1}},\vartheta_{\ell_1+1},\ldots,\vartheta_{\ell}, \vec{p}, \vec{x})
\mapsto
(e^{\ri \theta_{1}},\ldots,e^{\ri \theta_{\ell_1}},\theta_{\ell_1+1},\ldots,\theta_{\ell},
{\zeta}(\vec{p},\vec{x}))
\label{A15+1}
\ee
using equations \eqref{AA12} and \eqref{AA14}.
The non-zero Poisson brackets of the functions $(\theta_i, p_j, z_a)$ as given by this mapping using the relations \eqref{AA3}
are still of the form
\be
\{ \theta_i, \theta_j \} =0,
\quad
\{\theta_i, p_j\} = \delta_{ij},
\quad
\label{AA16}\{ z_a, z_b\} = f_{ab}
\ee
with some functions $f_{ab}$ of $\vec{z}$.
Indeed, if this holds for $\vec{\vartheta} \in C_{r_2}^\ell$, then it
holds everywhere, as is easily seen from the formulae \eqref{AA12}, \eqref{AA14}.

The coordinates $(\theta_i, p_j, z_a)$  constructed by the above change of variables
cover an open subset $\pi^{-1}(\varSigma_0') \subset \cW$,
where $\varSigma_0'$ correspond to
$\cD_{\proj}'$ in terms of the coordinates $(p_i, z_a)$.
We can further restrict to an open subset $\varSigma_0\subset \varSigma_0'$ that corresponds
to a hypercube $C_\epsilon^{d-\ell}$ in terms of the coordinates $(p_i, z_a)$.
We have shown that the restriction of the Poisson tensor $P_\cM$  on $\cU := \pi^{-1}(\varSigma_0)$ acquires the form  \eqref{app:AA2}, and
hence the theorem is proved.
\end{proof}

\begin{rem}
Following \cite{LGMV}, we call the variables $p_i, \theta_j$
(generalized) action-angle coordinates
and the $z_a$ transversal coordinates.
Their construction out of the obvious
coordinates $\vartheta_i, p_j, x_a$ involves two basic ingredients:
 i) shrinking the base $\varSigma$ and
changing coordinates on it according to \eqref{AA14}; ii) changing the trivialization of the bundle $\pi^{-1}(\varSigma_0)$
as is induced by the formula \eqref{AA12}.
The term `generalized' indicates that $G^2$ is not necessarily compact.
\end{rem}

\begin{cor} \label{Cor:AA}
Consider an integrable system with generalized  action variables in the
sense of Definitions \ref{defn:int} and \ref{defn:genact}.
Then, there exist generalized action-angle and transversal coordinates on
a $G^2$-invariant open
neighbourhood $\cU$ of every point of $y_0\in \widetilde \cM$ (see Definition \ref{defn:genact})
having trivial $G^2$ isotropy group, and
the action coordinates $p_i$ can be taken to be restrictions
of the functions $H_i$.
Moreover, $\cU$ can be chosen in such a manner that the action coordinates $p_i$
and transversal coordinates $z_a$
can be expressed in terms of restrictions of elements of the Abelian Poisson algebra $\fH$ and  its
constants of motion $\fF$, respectively.
\end{cor}
\begin{proof}
The statement follows directly from the definitions and Theorem \ref{thm:AA}.
The only point that is  not immediately obvious
is that the transversal coordinates can be expressed in terms of elements of
$\fF\subset C^\infty(\cM)$.
To see this, one may apply the arguments in the proof of Proposition \ref{prop:intM}
 to show that for  any $y_0$ with trivial $G^2$ isotropy group
there exist $\dim(\cM) - \ell$  elements of $\fF$
that are independent on a suitable $G^2$-invariant open neighbourhood of $y_0$.
\end{proof}

\section{Examples based on cotangent bundles of compact Lie groups} \label{Sec:Cot}

Here, we present examples for which
the manifold $M$ of Scenario \ref{scen:G11} will be
a cotangent bundle.
First, we need to recall a few facts \cite{DK,Knapp,Sam} about compact Lie groups associated with simple Lie algebras
that will be utilized in our examples.

\subsection{Useful facts about compact simple Lie algebras and  Lie groups}

Let $\fk$ be a compact simple Lie algebra with a fixed maximal Abelian subalgebra $\ft$
and consider also their complexifications $\fk^\bC$  and $\ft^\bC$.
Equip $\fk^\bC$ with the normalized Killing form $\langle - ,- \rangle$, given by
\be
\langle Z_1,  Z_2 \rangle = c\, \tr (\ad_{Z_1} \circ \ad_{Z_2}), \quad Z_1, Z_2\in \fk^\bC,
\label{J1}\ee
where $c$ is some convenient, positive constant.
Choose a system of  positive roots, $\fR^+$,  with respect to the Cartan subalgebra $\ft^\bC$,
and let
\be
\Delta = \{ \alpha_1, \dots, \alpha_\ell\}
\label{J2}\ee
be the corresponding set of simple roots, where $\ell$ is the rank of $\fk$.
Next, denote by $\theta$ the complex conjugation on $\fk^\bC$ with respect to  $\fk$, and  select root vectors $e_{\pm \alpha}$ for all $\alpha \in \fR^+$ in such a way that
$\langle e_\alpha, e_{-\alpha} \rangle = 2/\vert \alpha \vert^2$ and $e_{-\alpha} = - \theta(e_\alpha)$.
Letting
\be
h_{\alpha_j}:= [e_{\alpha_j}, e_{-\alpha_j}] \quad \hbox{ for}
\quad \alpha_j \in \Delta,
\label{J3}\ee
a Weyl-Chevalley basis of $\fk^\bC$ is provided by the elements
\be
e_\alpha,\,\, e_{-\alpha},\,\, h_{\alpha_j}
\quad \hbox{with}\quad \alpha \in \fR^+,\, j=1,\dots, \ell.
\label{J4}\ee
Using such a basis, we have
\be
\fk = \span_{\bR}\{\ri h_{\alpha_j},\, (e_\alpha -e_{-\alpha}),\, \ri(e_\alpha + e_{-\alpha}) \mid \alpha_j\in \Delta,\, \alpha\in \fR^+\},
\label{J5}\ee
and we also have the `Borel' subalgebra
\be
\fb = \span_{\bR}\{ h_{\alpha_j},\, e_\alpha,\, \ri e_\alpha  \mid \alpha_j\in \Delta,\, \alpha\in \fR^+\}
\label{J6}\ee
of the realification of $\fk^\bC$.

Let $K^\bC$ be a connected and simply connected Lie group corresponding to $\fk^\bC$.
We regard $K^\bC$ as a \emph{real} Lie group, and denote by
$K$ and $B$  its connected  subgroups corresponding to $\fk$ and $\fb$, respectively.
These are closed subgroups and are also simply connected ($B$ is topologically trivial).
The compact Lie group $K$  acts on itself as well as on $\fk$ by  conjugations\footnote{We will denote the
adjoint action of $K$ on $\fk$ also as conjugation, for example, in \eqref{J20}  $\eta J \eta^{-1}$ denotes
$ \mathrm{Ad}_\eta(J)$.}, and
we shall need the submanifolds of principal orbit type, $K^\reg$ and $\fk^\reg$,
on which the isotropy groups are the maximal tori of $K$.

Let $\cC \subset \ri \ft$ be the open Weyl chamber defined as follows:
\be
\cC = \{  X \in \ri \ft \mid 0<\alpha_j(X),\,\, \forall j=1,\dots ,\ell\}.
\label{J7}\ee
Define also the open Weyl alcove $\cA\subset \cC \subset \ri \ft$:
\be
 \cA = \{ X \in \ri \ft \mid 0 < \alpha_j(X),\,\, \forall j=1,\dots,\ell, \,\,\hbox{and}\,\, \vartheta(X)< 2\pi \},
\label{J8}\ee
where $\vartheta$ is the highest root with respect to the base $\Delta$.
Both $\cC$ and $\cA$  are convex domains, whose closures are denoted $\bar \cC$ and $\bar \cA$, with $\bar \cA$ compact.
It is well known that $\fk^\reg/K$ is diffeomorphic to $\cC$ and
$\fk/K$ is homeomorphic to $\bar \cC$.
Furthermore,  $K^\reg/K$ is diffeomorphic to $\cA$ and $K/K$ is homeomorphic to $\bar \cA$.
For later use, we below elaborate the pertinent diffeomorphisms.

\begin{defn}\label{defn:J1}
Introduce the mappings $\underline{\varphi}: \fk^\reg \to \cC$  and $\underline{\chi}: K^\reg \to \cA$ by the following recipes:
\be
\underline{\varphi}(J)  =\xi  \quad \hbox{if}\quad  \ri \xi = \Ad_{\Gamma_1}(J) \equiv \Gamma_1(J) J \Gamma_1(J)^{-1}
\quad \hbox{for some}\quad  \Gamma_1(J) \in K,
\label{J9}\ee
\end{defn}
and
\be
\underline{\chi}(g) = \xi
 \quad \hbox{if}\quad   e^{\ri \xi} = \Gamma_2(g) g \Gamma_2(g)^{-1}\quad \hbox{for some}\quad \Gamma_2(g)\in K.
\label{J10}\ee

For regular $J$ and $g$, $\Gamma_1(J)$ and $\Gamma_2(g)$ are unique up to \emph{left-multiplication} by arbitrary elements
of the maximal torus $\bT < K$ whose Lie algebra is $\ft$. Their equivalence classes  give well-defined
mappings onto the homogeneous space $\bT\backslash K $.
We now quote a standard result that can be found in many sources \cite{DK}. Here, we denote $[\Gamma_1(J)]= \bT \Gamma_1(J) \in \bT\backslash K$ and
$[\Gamma_2(g)] = \bT \Gamma_2(g) \in \bT\backslash K$.
\begin{lem}\label{lem:J2}
The combined mappings
\be
([\Gamma_1], \underline{\varphi}): \fk^\reg \to \bT\backslash K \times \cC
\quad
\hbox{and}\quad
([\Gamma_2], \underline{\chi}): K^\reg \to  \bT \backslash K  \times \cA
\label{J11}\ee
 are real-analytic    diffeomorphisms.
Consequently,  the formulae
 \be
\varphi_j := \langle h_{\alpha_j}, \underline{\varphi} \rangle
\quad\hbox{and}\quad
\chi_j := \langle h_{\alpha_j}, \underline{\chi} \rangle,
\qquad \forall j=1,\dots, \ell,
\label{J12}\ee
define real-analytic, $K$-invariant real functions $\varphi_j$ on $\fk^\reg$ and $\chi_j$ on $K^\reg$,  respectively.
\end{lem}

For real functions $\varphi \in C^\infty(\fk)$ and $\chi\in C^\infty(K)$, introduce the $\fk$-valued derivatives
$d \varphi$ and $\nabla \chi$ by
\be
 \langle Z, d \varphi (J) \rangle  := \dt \varphi( J + t Z),
 \quad\hbox{and}\quad
 \langle Z, \nabla \chi(g) \rangle := \dt \chi(e^{tZ} g ), \quad \forall
  Z \in \fk.
\label{J13} \ee
 It is easily seen that the derivatives of  invariant functions $\varphi \in C^\infty(\fk)^K$ and $\chi \in C^\infty(K)^K$ are
 equivariant with respect to the respective actions of $K$. Moreover,  at regular elements we have
 \be
 \fk_J = \operatorname{span}\{ d \varphi(J) \mid \varphi \in C^\infty(\fk)^K \}
 \quad \hbox{and}\quad
 \fk_g= \operatorname{span}\{ \nabla \chi(g) \mid \chi \in C^\infty(\fk)^K \},
 \label{J14}\ee
 where $\fk_J$ and $\fk_g$ are the Lie algebras of the maximal tori fixing $J\in \fk^\reg$ and $g\in K^\reg$, respectively.

 \begin{lem}\label{lem:J3}
 The derivatives of $\varphi_{j} \in C^\infty(\fk^\reg)^K$ and $\chi_{j} \in C^\infty(K^\reg)^K$ defined in \eqref{J12} are given by
 \begin{equation}
  \begin{aligned} \label{J15}
d \varphi_{j} (J) &\,= -\Gamma_1(J)^{-1} \ri h_{\alpha_j}  \Gamma_1(J),\,\, \forall J\in \fk^\reg ; \\
\nabla \chi_{j}(g) &\,= -\Gamma_2(g)^{-1} \ri h_{\alpha_j}  \Gamma_2(g),\,\, \forall g\in K^\reg.
  \end{aligned}
 \end{equation}
 \end{lem}
 \begin{proof}
Since the derivatives of the invariant function are equivariant, it is enough to calculate  $d \varphi_{j}$ at the points of $\ri \cC$ with the Weyl chamber $\cC$,
and $\nabla \chi_j$ at the points of $e^{\ri \cA}$ with the Weyl alcove $\cA$. Also by the equivariance, the derivatives
at those points are $\ft$-valued.
It is easy to verify that
\be
 d \varphi_{j} (\ri \xi) = - \ri h_{\alpha_j} ,\,\, \forall \xi\in \cC ,
\quad \hbox{and}\quad
\nabla \chi_{j}(e^{\ri \xi}) = - \ri h_{\alpha_j}  ,\,\, \forall \xi\in \cA,
 \label{J16}\ee
 from which the claim follows.
\end{proof}

\begin{rem}\label{rem:J4}
We shall use the existence of pairs of maximal tori $\bT$ and $\bT'$ for which
\be
\bT \cap \bT' = \cZ(K),
\label{J17}\ee
where $\cZ(K)<K$ is the center, and their Lie algebras  $\ft$ and $\ft'$  are orthogonal to each other.
Recall that $\cZ(K)$ is the intersection of all maximal tori, and see Remark 4.5 in \cite{Fe23}  for an example in $\mathrm{SU}(n)$.
For completeness, and since later we shall also need the underlying auxiliary information on Coxeter elements,
we here elaborate on the existence of such tori.

Let $w_* \in N_K(\bT')/\bT'$ be a Coxeter element of the Weyl group of the pair $(K,\bT')$ for some maximal torus $\bT'$.
 It is known that $(w_* -\id)$ yields an invertible linear transformation
on $\ft'$  and the fixed point set for the action of $w_*$ on $\bT'$ is the center $\cZ(K)$ \cite{DW}.
Suppose that $g_* \in  N_K(\bT')\cap \bT^\reg$ is a representative of $w_*$.
 Since $\bT<K$ is the fixed point set of $\Ad_{g_*}$, it follows that \eqref{J17} holds.
 Consider the orthogonal decomposition  $\fk= \ft + \ft^\perp$.    By using that $w_* = [g_*]$ with $g_*\in \bT$, we obtain  $(w_* - \id)(X)\in  \ft^\perp$   for  every $X\in \ft'$,
which entails $\ft' \subset \ft^\perp$.

Take the `principal element' $g_* = \exp(2\pi \ri \rho^\vee/N)\in \bT^\reg$,  where $N$ is the Coxeter number
 and $\rho^\vee = \frac{1}{2} \sum_{\alpha\in \fR^+} h_\alpha$.
In effect,
 Kostant \cite{Ko} constructed a maximal torus $\bT'$  `in apposition' to $\bT$ on which
the principal element acts as a Coxeter element.  We see that this torus satisfies \eqref{J17} and $\ft' $ is orthogonal to $\ft$.
For a review of Kostant's pertinent construction and other related results,   see also \cite[\S2]{Me}.
\end{rem}

\subsection{Integrable systems from Poisson reduction of \texorpdfstring{$T^*K$}{T*K}}

Below, $K$ is a compact, connected and simply connected Lie group associated with the simple Lie algebra $\fk$.
We identify the cotangent bundle $T^*K$ with the manifold
\be
M := K \times \fk = \{(g,J)\mid g\in K,\, J\in \fk\}
\label{J18}\ee
by using right-translations and the Killing form.
Then, for any $F,H\in C^\infty(M)$, the canonical (symplectic) Poisson bracket of the cotangent bundle can be written as
\be
\{ F, H\}(g,J) =
\langle \nabla_1 F, d_2 H\rangle - \langle \nabla_1 H, d_2 F \rangle + \langle J, [d_2 F, d_2 H]\rangle,
\label{J19}\ee
where the derivatives are taken with respect to the first and second variable, respectively, and are evaluated  at $(g,J)$.
The group $K$ acts on $M$ by the mapping
\be
K\times M\ni (\eta, (g,J)) \mapsto (\eta g \eta^{-1}, \eta J \eta^{-1})\in M.
\label{J20}\ee
This is a Hamiltonian action generated by the momentum map
\be
\Phi(g,J) = J - g^{-1} J g,
\label{J21}\ee
where we identified $\fk^*$ with $\fk$.

Let $p_\fk: M\to \fk$ and $p_K: M\to K$ denote the projections, for which
\be
p_\fk: (g,J) \mapsto J \quad \hbox{and} \quad  p_K: (g,J) \mapsto g,
\label{J22}\ee
and consider the $K$-invariant functions
\be
\varphi\circ p_\fk,\,\, \forall \varphi \in C^\infty(\fk)^K
\quad \hbox{and}\quad
\chi \circ  p_K,\,\, \forall \chi \in C^\infty(K)^K.
\label{J23}\ee
For an arbitrary initial value $(g_0, J_0) \in M$, the integral curve of the evolution equation generated
by the Hamiltonian  $\varphi \circ p_\fk$ is
\be
(g(\tau), J(\tau)) = (\exp\left( \tau d \varphi(J_0)\right) g_0, J_0),\quad \forall \tau \in \bR,
\label{J24}\ee
and for the Hamiltonian $\chi \circ  p_K$ it is
\be
(g(\tau), J(\tau)) = (g_0, J_0 - \tau \nabla \chi(g_0)), \quad \forall \tau \in \bR.
\label{J25}\ee
Thus,  $M$ carries the  following two Abelian Poisson subalgebras of $C^\infty(M)^K$ having complete flows,
\be
\fH := p_\fk^*(C^\infty(\fk)^K)
\quad\hbox{and}\quad
\tilde \fH:=p_K^*(C^\infty(K)^K).
\label{J26}\ee

\begin{rem}\label{rem:J5}
Note that the pair $(g(\tau)^{-1} J(\tau) g(\tau), J(\tau))$ is constant along the integral curves \eqref{J24} and the pair $(g(\tau), J(\tau) - g(\tau)^{-1} J(\tau) g(\tau))$ is constant
along the integral curves \eqref{J25}.
\end{rem}

Consider the $K$-invariant dense open, connected submanifolds of  $M$,
\be
Y:= K \times \fk^\reg
\quad
\hbox{and}\quad \tilde Y := K^\reg \times \fk.
\label{J27}\ee
The connectedness follows from Lemma \ref{lem:J2}.
Then, define the $K$-invariant mappings
\be
 (H_1,\dots, H_\ell): Y \to \bR^\ell
\quad \hbox{and}
\quad
(\tilde H_1,\dots, \tilde H_\ell) : \tilde Y \to \bR^\ell
\label{J28}\ee
by introducing the functions
\be
H_j := \varphi_j \circ p_\fk
\quad\hbox{and}\quad
\tilde H_j := \chi_j \circ p_K, \qquad j=1,\dots, \ell,
\label{J29}\ee
where \eqref{J12} is used  and $p_\fk$ and $p_K$ now denote the corresponding restrictions of the projections.
Introduce also the \emph{diffeomorphisms}
\be
\cT: \bR^\ell \to \ft
\quad \hbox{and}\quad T: \bR^\ell / (2 \pi \bZ)^\ell \to \bT
\label{J30}\ee
 by the formulas
 \be
 \cT(\taU):= -\ri \sum_{j=1}^\ell \tau_j h_{\alpha_j}
 \quad \hbox{and}\quad
T(\taU):= \exp\bigl( -\ri \sum_{j=1}^\ell \tau_j h_{\alpha_j}\bigr),
\quad \forall \taU=(\tau_1,\dots, \tau_\ell) \in \bR^\ell.
\label{J31}\ee

\begin{lem}\label{lem:J6}
The map $(H_1,\dots, H_\ell)$ \eqref{J28}  is the momentum map for the free Hamiltonian action of the maximal torus $\bT<K$ on $Y$ \eqref{J27}
that works according to the formula
\be
(T(\taU), (g,J)) \mapsto (\Gamma_1(J)^{-1}T(\taU) \Gamma_1(J) g, J), \quad
\forall  \taU\in \bR^\ell,\,\, (g,J) \in Y.
\label{J32}\ee
The map $(\tilde H_1,\dots, \tilde H_\ell)$ \eqref{J28} serves as the momentum map generating the free and proper Hamiltonian action of
$\bR^\ell$ on $\tilde Y$ \eqref{J27} that operates as
\be
(\taU, (g, J)) \mapsto (g, J - \Gamma_2(g)^{-1}\cT(\taU)\Gamma_2(g)), \quad \forall \taU\in \bR^\ell, \,\, (g,J) \in \tilde Y.
\label{J33}\ee
These Abelian group actions commute with the $K$-action \eqref{J20} restricted on $Y$ and on $\tilde Y$, respectively.
\end{lem}
\begin{proof}
The fact that the maps \eqref{J28} and \eqref{J29} serve as momentum maps generating the above actions follows by combining the formulas
\eqref{J24} and \eqref{J25}  with the derivatives of the functions $\varphi_j$ and $\chi_j$ determined in Lemma \ref{lem:J3}.
It is clear that these actions of $\bT$ on $Y$ and $\bR^\ell$ on $\tilde Y$ are free and commute with the respective $K$-action.

To show  that the $\bR^\ell$-action  \eqref{J33} is proper,
consider a sequence $(g_n, J_n)$ in $\tilde Y$ that converges to $(g_*, J_*)\in \tilde Y$.
 Suppose that $\taU_n$ is a sequence in $\bR^\ell$ such that the sequence
 $(g_n, J_n - \Gamma_2(g_n)^{-1}\cT(\taU_n)\Gamma_2(g_n))$ converges to $(\hat g, \hat J) \in \tilde Y$.
 Then, the sequence  $\Gamma_2(g_n)^{-1}\cT(\taU_n)\Gamma_2(g_n)$ is also convergent, and so is
 the sequence  $\cT(\taU_n)$ since  $ K^\reg \ni g \mapsto [\Gamma_2(g)] \in \bT\backslash K$ is a smooth map (by Lemma \ref{lem:J2}).
\end{proof}

\begin{lem}\label{lem:J7}
The principal isotropy group for the $K$-action on $Y$ as well as for the $K$-action on $\tilde Y$ is given
by the center $\cZ(K)$ of the group $K$.
\end{lem}
\begin{proof}
We can choose $(g,J) \in Y \cap \tilde Y$ in such a way that $g\in  \bT^\reg$ and $J\in \ft'^\reg$, where
$\bT$ and $\bT'$ are two maximal tori whose intersection is $\cZ(K)$.  (Recall Remark \ref{rem:J4}.)
The $K$ isotropy group of this point $(g,J)$
is $\cZ(K)$. Since $\cZ(K)$ is contained in every isotropy group, the statement follows.
\end{proof}

\begin{lem}\label{lem:J8}
The principal isotropy group for the combined action of $\bG := K \times \bT$
  on $Y$ is $\cZ(K) \times \{ e\}$.
\end{lem}
\begin{proof}
Let us choose $(g,J)\in Y$ in such a way that the joint isotropy group of $g^{-1} J g$ and $J$ in $K$ is the center $\cZ(K)$.
One sees from Remark \ref{rem:J5}, and can also check directly, that
the equality
\be
(g,J) = (\eta \Gamma_1(J)^{-1} T(\taU) \Gamma_1(J) g \eta^{-1}, \eta J \eta^{-1}), \quad  (\eta, T(\taU)) \in \bG,
\label{J34}\ee
entails the equality $(g^{-1} J g, J)=  (\eta (g^{-1} J g)\eta^{-1}, \eta J \eta^{-1})$.
With our choice of $(g,J)$,  it follows that $\eta \in \cZ(K)$, and then \eqref{J34}  implies  $T(\taU) = e$.
The same argument that was used in the previous lemma concludes the proof.
\end{proof}

\begin{lem}\label{lem:J9}
The principal isotropy group for the combined action of $\tilde \bG := K \times \bR^\ell$
on $\tilde Y$ is $\cZ(K) \times \{ 0\}$.
\end{lem}
\begin{proof}
Let us choose $g\in \exp(\ri \cA) \subset \bT^\reg$ and $J\in {\ft'}^\reg$ so that $\bT \cap  \bT' = \cZ(K)$ and
the maximal Abelian subalgebras $\ft$ and $\ft'$ of $\fk$ are orthogonal to each other.
In this case $\Gamma_2(g)^{-1} \cT(\taU) \Gamma_2(g) = \cT(\taU)$,
and the equality
\be
(g,J) = (\eta g \eta^{-1}, \eta (J - \cT(\taU)) \eta^{-1}),  \quad  (\eta, \cT(\taU)) \in \tilde \bG,
\label{J35}\ee
implies that
\be
\langle J, J \rangle = \langle J - \cT(\taU), J - \cT(\taU) \rangle = \langle J, J \rangle + \langle \cT(\taU), \cT(\taU) \rangle,
\label{J36}\ee
which in turn implies that $\cT(\taU) = 0$.  Putting this back into \eqref{J35}, we get that $\eta \in \cZ(K)$.
Since we have ${\tilde \bG}_{(g,J)} = \cZ(K) \times \{0\}$,  the claim of the lemma follows.
\end{proof}

\begin{prop}\label{pr:J10}
Theorem \ref{thm:G12} is applicable to  the Abelian Poisson algebra $\fH$ \eqref{J26} with
 $M=T^* K$,
$Y\subset M$ \eqref{J27}, $(H_1,\ldots,H_\ell)$ \eqref{J29} and
$\Gone=K$, $\Gtwo=\bT$ acting as described above.
Consequently,  $\fH$  engenders integrable systems of rank $\ell$ whose action variables arise from $H_1,\dots, H_\ell$.
Moreover, the analogous statements hold if we replace $(\fH,Y,H_1,\dots, H_\ell, K\times \bT)$ by $(\tilde \fH, \tilde Y,
 \tilde H_1,\dots, \tilde H_\ell,  K\times \bR^\ell)$.
\end{prop}

\begin{proof}
We check that the assumptions of Scenario \ref{scen:G11} are satisfied by $(M,P_M)$,
$\fH$ \eqref{J26} and  $Y$ \eqref{J27}
with $\Gone= K$ and $\Gtwo= \bT$ as described above.
We already noticed that $\fH$ is Abelian with complete flows, proving item \ref{itScen:1}.
Then, any $\varphi\in C^\infty(\fk)^K$ can be expressed in terms of the $\varphi_j$ \eqref{J12} on $\fk^\reg$, hence
we can express any element of $\fH$ in terms of the $H_j$ on $Y$.
It follows from the equalities in \eqref{J14} and \eqref{J15} that $H_1,\dots, H_\ell$ are independent at every point of $Y$,
and item \ref{itScen:3} holds.
The vector fields $V_{H_j}$  generate the $\bT$-action \eqref{J32}, which is obviously proper and effective,
and we also have $\{H_i,H_j\}=0$; thus item \ref{itScen:2} is verified.
Finally,  by Lemma \ref{lem:J8}, the action of $K\times \bT$ has principal isotropy type $\cZ(K)\times \{e\}$;
so item \ref{itScen:4} of Scenario \ref{scen:G11} holds, too.
Therefore, Theorem \ref{thm:G12} is directly applicable.
We similarly arrive at the same conclusions for $\tilde Y$, $\tilde \fH$ and $\Gtwo= \bR^\ell$.
\end{proof}

\begin{rem}\label{GSrem}
The construction of the torus action \eqref{J32} is a special case of a general construction
of $2\pi$-periodic Hamiltonian flows on (an open subset
of) a symplectic $K$ manifold $M$ equipped with an equivariant momentum map, here denoted $\Psi: M \to \fk^* \simeq \fk$.
This is due to Guillemin and Sternberg \cite{GS} who applied it, for example, for obtaining action variables
of Gelfand-Cetlin systems.
Specifically, the functions $H_j$ \eqref{J29} arise from Theorem 3.4 of \cite{GS}
by using the momentum map $\Psi(g,J)=J$ that generates the cotangent lift of the action
of $K$ on itself by left multiplications.
We learned this connection with \cite{GS}  after  completion of the present paper, and thank R.~Fioresi
for bringing it to our attention.
\end{rem}

\begin{rem}\label{rem:oldres}
The systems $(M,P_M, \fH)$ and $(M, P_M, \tilde \fH)$ with the constants of motion
mentioned in Remark \ref{rem:J5} are degenerate integrable systems.
It was shown in the papers \cite{FeCONF,Re1}  that the Poisson reductions of these systems inherit
degenerate integrability on \emph{generic symplectic leaves} of (the smooth part of)
the quotient space $M/K$
 where they can be interpreted as  trigonometric
spin Sutherland systems and their rational Ruijsenaars--Schneider  type
duals.
Our Proposition \ref{pr:J10} gives considerably stronger results,
which we derived by a method different from those  used in \cite{FeCONF,Re1}.
 For a spinless special case that occurs for $K={\mathrm{SU}}(n)$, see also \cite{FeA}.
\end{rem}

\section{Generalization to Heisenberg doubles  of compact Lie groups} \label{Sec:Heis}

In what follows we can restrict ourselves to a brief exposition,
since the omitted verifications can be extracted from \cite{Fe23,Fe24}.

The Heisenberg double \cite{STS85,STS} is a Poisson--Lie analogue of the cotangent bundle $T^*K$.
It is given by a (symplectic) Poisson structure on the real group manifold
\be
M:= K^\bC,
\label{H1}\ee
whose definition relies on the non-degenerate bilinear form
\be
\langle Z_1, Z_2\rangle_\bI := \Im \langle Z_1,Z_2 \rangle,
\qquad
\forall Z_1, Z_2 \in \fk^\bC,
\label{H2}\ee
and on the real vector space decomposition
\be
\fk^\bC = \fk + \fb.
\label{H3}\ee
We also need the $\fk^\bC$-valued derivatives
$DF$ and $D'F$ of the real functions $F\in C^\infty(M)$ that satisfy
\be
 \langle Z_1, D F(X) \rangle_\bI + \langle Z_2, D' F(X) \rangle_\bI := \dt F(e^{tZ_1} X e^{tZ_2}), \quad \forall X\in M, \, Z_1,Z_2 \in \fk^\bC.
 \label{H4}\ee
Then, the Poisson bracket on $C^\infty(M)$ is defined by
\be
\{ F, H\}: = \langle D F, \varrho D H \rangle_\bI +  \langle D' F, \varrho D' H \rangle_\bI
\quad\hbox{with}\quad \varrho := \frac{1}{2}\left( \pi_{\fk} - \pi_{\fb}\right),
\label{H5}\ee
where $\pi_\fk$ and $\pi_\fb$ are the projections from $\fk^\bC$ onto the isotropic subalgebras $\fk$ and $\fb$.

According to the Iwasawa decomposition \cite{Knapp}, every element $X\in K^\bC$ admits the unique decompositions
\be
X = g_L b_R^{-1} = b_L g_R^{-1} \quad \hbox{with}\quad g_L, g_R \in K,\, b_L, b_R \in B.
\label{H6}\ee
These decompositions yield the smooth maps $\Xi_L, \Xi_R: M\to K$ and $\Lambda_L, \Lambda_R: M\to B$,
\be
\Xi_L(X) := g_L,\quad \Xi_R(X):= g_R,\quad \Lambda_L(X):= b_L,\quad \Lambda_R(X):= b_R.
\label{H7}
\ee
Then,
\be
\Lambda := \Lambda_L \Lambda_R
\label{H8}\ee
is a group valued Poisson--Lie momentum map \cite{Lu}, which generates
the so-called quasi-adjoint action of $K$ on $M$. The pertinent action map \cite{Klim} operates according to
\be
K \times M \ni (\eta, X) \mapsto \eta X \Xi_R( \eta \Lambda_L(X)),
\qquad
\forall (\eta,X) \in K \times M.
\label{H9}\ee
With respect to this $K$-action, $C^\infty(M)^K$ is closed under the Poisson bracket.

Before proceeding further,
it is useful to introduce the closed submanifold
\be
\fP := \exp(\ri \fk) \subset K^\bC,
\label{H16}\ee
and  the diffeomorphism
\be
\cP: B \to \fP \quad \hbox{for which} \quad \cP(b):= b b^\dagger,
\label{H17}\ee
where $ b^\dagger := \Theta(b^{-1})$ with the Cartan involution
$\Theta: K^\bC \to K^\bC$.
For matrix Lie groups, one can identify $b^\dagger$ with the usual adjoint of $b$.
Using the dressing action of $K$ on $B$, for which $\eta \in K$ acts on $B$ by
 \be
\Dress_\eta (b) := \Lambda_L(\eta b),
\quad
\forall \eta \in K,\, b\in B,
\label{H11}\ee
 we have
\be
\cP(\Dress_\eta(b)) = \eta \cP(b) \eta^{-1},
\quad \forall \eta \in K,\, b\in B.
\label{H18}\ee
Using also that $\fP$ is diffeomorphic to $\ri \fk$ by the exponential map, we see that
the principal isotropy type submanifolds of $\fk$, $\fP$ and $B$,  with respect to the conjugation and dressing actions of $K$,
are related according to
\be
\exp(\ri \fk^\reg) = \fP^\reg =  \cP(B^\reg).
\label{H19}\ee

Now, we introduce  two Abelian Poisson subalgebras  $\fH$ and $\tilde \fH$ of $C^\infty(M)^K$ as follows:
\be
\fH = \Lambda_R^*( C^\infty(B)^K)
\quad \hbox{and}\quad
\tilde \fH = \Xi_R^* ( C^\infty(K)^K).
\label{H10}\ee
 The $K$-invariance in \eqref{H10} refers to the dressing action
 \eqref{H11} and to the conjugation action of $K$ on itself.
 The pull-backs are invariant with respect to the action \eqref{H9} since $X \mapsto \eta X \Xi_R( \eta \Lambda_L(X))$ implies
 \be
 (g_R, b_R) \mapsto (\tilde \eta g_R \tilde \eta^{-1}, \Dress_{\tilde \eta} (b_R)) \quad \hbox{with}\quad
 \tilde \eta = \Xi_R(\eta b_L)^{-1},
 \label{messy}\ee
 using equations \eqref{H6} and \eqref{H7}.

In order to describe the integral curves \cite{Fe23},  we introduce the $\fk$-valued derivative of any
$\phi \in C^\infty(B)$ by the requirement
\be
\langle Z, D\phi(b) \rangle_\bI = \dt \phi (e^{t Z} b), \quad \forall Z\in \fb.
\label{H12}\ee
The integral curve of the Hamiltonian $\Lambda_R^*(\phi) \in \fH$
through the initial value $X_0$ has the form
\be
X(\tau) = X_0 \exp\left( - \tau D\phi(\Lambda_R(X_0))\right), \qquad \forall \tau \in \bR.
\label{H13}\ee
Next, take any $\chi\in C^\infty(K)^K$, and consider the Iwasawa decomposition
for fixed $g \in K$
\be
\exp(\ri \tau \nabla \chi(g)) = \beta_\chi(\tau,g) \gamma_\chi(\tau,g), \quad
\beta_\chi(\tau,g) \in B,\,\, \gamma_\chi(\tau,g) \in K.
\label{H14}\ee
With this, the integral curve of the Hamiltonian $\Xi_R^*(\chi) \in \tilde \fH$  turns out  to be
\be
X(\tau) = X_0 \beta_\chi(\tau,\Xi_R(X_0)), \qquad \forall \tau\in \bR.
\label{H15}\ee
We see that the elements $\fH$ and $\tilde \fH$ generate complete flows.
Taking a closer look at the above formulas, one can also confirm that $\fH$ and $\tilde \fH$
both are Abelian Poisson algebras.

By using the relations \eqref{H19}, we define the functions $\phi_j \in C^\infty(B^\reg)^K$ by setting
\be
\phi_j(b):= \frac{1}{2} \varphi_j(\ri \log (bb^\dagger)), \qquad \forall j=1,\dots, \ell,
\label{H20}\ee
with  $\varphi_j\in C^\infty(\fk^\reg)^K$  introduced in Lemma \ref{lem:J2}.

\begin{lem}\label{lem:H1}
With the map $[\Gamma_1]: \fk^\reg \to \bT\backslash K $ in  \eqref{J11}, the derivative of  $\phi_j \in C^\infty(B^\reg)^K$ is given by
\be
D\phi_j(b) = \Gamma_1(\ri \log(bb^\dagger))^{-1} \ri h_{\alpha_j} \Gamma_1(\ri \log(bb^\dagger)).
\label{H21}\ee
\end{lem}
\begin{proof}
For any $\phi\in C^\infty(B)^K$,  one has $D\phi (\Dress_\eta(b)) = \eta (D\phi(b)) \eta^{-1}$, as proved in \cite{Fe23}.
 It follows from \eqref{H18} and \eqref{H19} that $b\in B^\reg$
can be transformed into a unique element $e^\xi$ for some $\xi \in \cC$ by the dressing action, and $D\phi (e^\xi) \in \ft$ must hold because
$\Dress_\eta(e^\xi) = e^\xi$ for $\eta \in \bT$.  It is easy to verify that
$D\phi_j (e^\xi) = \ri h_{\alpha_j}$,
and then the claim follows from the equivariance property of $D\phi_j$ using also \eqref{H18} and \eqref{H19}.
\end{proof}

Similarly\footnote{The similarity could be made more manifest by replacing $M=K^\bC$
by its model  $K \times B$ relying on the map $(\Xi_R,\Lambda_R)$.}
  to the cotangent bundle case, we introduce the following connected, dense open
submanifolds of the Heisenberg double $M$:
\be
Y:= \Lambda_R^{-1}(B^\reg) \quad \hbox{and} \quad \tilde Y:= \Xi_R^{-1}(K^\reg).
\label{H22}\ee
These are invariant submanifolds for the $K$-action \eqref{H9}, since transformation of $b_R$ and $g_R$ obeys \eqref{messy}.
Then, we define the $K$-invariant maps
\be
 (H_1,\dots, H_\ell): Y \to \bR^\ell
\quad \hbox{and}
\quad
(\tilde H_1,\dots, \tilde H_\ell) : \tilde Y \to \bR^\ell,
\label{H23}\ee
now using the functions
\be
H_j := \phi_j \circ \Lambda_R
\quad\hbox{and}\quad
\tilde H_j := \chi_j  \circ \Xi_R, \qquad j=1,\dots, \ell,
\label{H24}\ee
with  $\phi_j$ in \eqref{H20} and $\chi_j$  prepared in Lemma \ref{lem:J2}.

For any $g\in K^\reg$ and $\taU\in \bR^\ell$, let $\beta(\taU, g) \in B$ be the unique solution of the
following factorization problem:
\be
\Gamma_2(g)^{-1}  \exp\bigl( \sum_{j=1}^\ell \tau_j h_{\alpha_j}\bigr) \Gamma_2(g)
= \beta(\taU,g) \gamma(\taU,g),
\qquad
\beta(\taU,g)\in B,\,\,  \gamma(\taU,g)\in K.
\label{H25}\ee
The main new result of this section is the subsequent generalization of Lemma \ref{lem:J6}.

\begin{lem}\label{lem:H2}
The map $(H_1,\dots, H_\ell)$ \eqref{H24} is the momentum map for the free Hamiltonian action of the torus $\bT$ \eqref{J30} on $Y$ \eqref{H22}
that operates according to
\be
(T(\taU), X) \mapsto X  \Gamma_1 (\ri \log (b_R b_R^\dagger))^{-1} T(\taU) \Gamma_1(\ri \log(b_R b_R^\dagger)), \quad
\forall  \taU\in \bR^\ell,\,\, X = g_L b_R^{-1} \in Y,
\label{H26}\ee
with $T(\taU)$ in \eqref{J31}.
The map $(\tilde H_1,\dots, \tilde H_\ell)$ \eqref{H24} is the momentum map for the free and proper Hamiltonian action of the Abelian group
$ \bR^\ell$ on $\tilde Y$ \eqref{H22} given by
\be
(\taU, X) \mapsto X \beta(\taU, g_R), \quad \forall \taU\in \bR^\ell, \,\, X= b_L g_R^{-1} \in \tilde Y,
\label{H27}\ee
with $ \beta(\taU, g_R)$ defined in \eqref{H25}.
These Abelian group actions commute with the $K$-action \eqref{H9} restricted on $Y$ and on $\tilde Y$, respectively.
\end{lem}
\begin{proof}
It is easily seen that the action maps \eqref{H26} and \eqref{H27}  are  the joint flows of the respective  Hamiltonians defined in \eqref{H24}.
The flows stay in $Y$ and in $\tilde Y$, respectively, since
 $\Lambda_R$ is constant along any integral curve \eqref{H13} and one has
 $\Xi_R(X(\taU))=\gamma(\taU,\Xi_R(X_0)) \Xi_R(X_0) \gamma(\taU,\Xi_R(X_0))^{-1}$
along the integral curves in \eqref{H15}. (For a verification, see \cite{Fe23}.)
The properness of the $\bR^\ell$-action \eqref{H27}  can be shown similarly to the analogous claim of Lemma \ref{lem:J6}.
These Hamiltonian actions commute with the $K$-action \eqref{H9}, since  any $K$-action generated by a Poisson--Lie momentum map commutes with the flow
generated by any $K$-invariant Hamiltonian.
\end{proof}

\begin{prop} \label{pr:H3}
Lemmas \ref{lem:J7}, \ref{lem:J8}, \ref{lem:J9} and Proposition \ref{pr:J10}
remain word for word valid, replacing the various notions
defined for the cotangent bundle with their Heisenberg double counterparts as defined above.
\end{prop}
\begin{proof}
Since $M= K^\bC$ is diffeomorphic to $K\times B$ by the map $(\Xi_R, \Lambda_R)$, we can choose $X\in Y\cap \tilde Y$ so that the joint
isotropy group of $g_R$ and $\cP(b_R)$ under the conjugation action of $K$ is $\cZ(K)$. Note that $\cZ(K) = \cZ(K^\bC)$ and that we used Remark \ref{rem:J4}.
By \eqref{messy}, $\eta \in K_X$  with respect to the action \eqref{H9} if and only if
  $\Xi_R(\eta b_L)^{-1} = c \in \cZ(K)$. But then  $\eta = c$ follows from the equality
$ \eta b_L = \Lambda_L(\eta b_L) c = c \Lambda_L (\eta b_L)$.
 Therefore $K_X = \cZ(K)$, which implies that the statement of Lemma \ref{lem:J7} is valid in our case.

 The proof of Lemma \ref{lem:J8} can be generalized to the present case by noticing that for the map $\cP:B\to \fP$ \eqref{H17}, the $\fP$-valued functions
 \be
 \cP \circ \Lambda_R
 \quad \hbox{and}\quad
 (\Xi_R)^{-1} (\cP \circ \Lambda_R) \Xi_R
\label{H29} \ee
 are constant along the  integral curves \eqref{H13} as well as along the orbits of the $\bT$-action \eqref{H26}.

 To verify the statement of Lemma \ref{lem:J9}, let us choose $X\in \tilde Y$ so that $g_R \in \exp(\ri \cA)$ \eqref{J8}
 and the isotropy group $K_X$ is $\cZ(K)$.
 In this case, the factorization problem \eqref{H25} is trivial, and it gives
 \be
 \beta(\taU, g_R) = \exp\bigl(\sum_{j=1}^\ell \tau_j h_{\alpha_j}\bigr), \,\, \gamma(\taU, g_R) = e.
 \ee
 The requirement that $(\eta, \taU) \in \bG=K \times \bR^\ell$ is in the isotropy group $\bG_X$ is now equivalent to
 \be
 g_R = \tilde \eta g_R \tilde \eta^{-1}
 \quad \hbox{and}\quad
 b_R = \Dress_{\tilde \eta} (\beta(\taU, g_R)^{-1} b_R).
 \label{H31}\ee
 Here, $\tilde \eta^{-1} = \Xi_R(\eta \Lambda_L(X))$ as in \eqref{messy}
  and we used that  $\Xi_R$ \eqref{H7} is constant on the $\bR^\ell$-orbit of the chosen  element $X\in \tilde Y$.
 The first equality enforces that $\tilde \eta \in \bT$. Consequently,  the second equality becomes
 \be
 b_R = \tilde \eta \beta(\taU, g_R)^{-1} b_R \tilde \eta^{-1}.
 \ee
As is easily seen, this equality implies that $\beta(\taU, g_R)$ is the identity element of $B$, which
means that $\taU= 0\in \bR^\ell$. Then,  $\tilde \eta \in \cZ(K)$ results from our assumption on $X$,
and this is equivalent to $\eta \in \cZ(K)$.
Having proved that $\bG_X = \cZ(K) \times \{0\}$, the statement of Lemma \ref{lem:J9} follows.

The fact that we can adapt  Proposition \ref{pr:J10} is a straightforward consequence of Scenario \ref{scen:G11}.
\end{proof}

 \begin{rem}
 One can verify that $W(X):= b_L b_R g_L^{-1} = b_L g_R b_L^{-1}$, for $X$ written as in \eqref{H6}, is constant along the integral curves
 $X(\tau)$  \eqref{H15}.
 This, and the conserved quantities \eqref{H29} along the integral curves \eqref{H13},
 were used in the previous papers \cite{Fe23,Fe24}
 to discuss the integrability of the pertinent reduced systems,  which are spin Ruijsenaars--Schneider type
deformations of trigonometric spin Sutherland systems and their duals.
However, no action variables appeared in \cite{Fe23,Fe24}
and the proof of the integrability of the `dual system' based on the Hamiltonians
$\tilde \fH$ \eqref{H10} has not been completed previously.
For important spinless special cases,  see also \cite{FK}.
\end{rem}

\section{Examples of integrable systems on  moduli spaces of flat connections} \label{Sec:qPoiss}

Arthamonov and Reshetikhin \cite{ARe} obtained  powerful general results
concerning  degenerate integrable (superintegrable) systems living on moduli spaces of flat connections.
Their approach for describing the Hamiltonian flows builds on the `$r$-matrix method' of Fock and Rosly \cite{FoR},
which makes use of an $r$-matrix of \emph{factorizable}\footnote{A `factorizable classical $r$-matrix'
$r\in \cG \otimes \cG$ for a Lie algebra $\cG$  is a solution of the classical Yang--Baxter equation
whose symmetric part  $r+ r_{21}$ is equal to  $T_a \otimes T^a$ for dual bases $\{T_a\}$, $\{T^a\}$
of $\cG$ defined by a non-degenerate, symmetric, invariant bilinear form.  Such $r$-matrices underlie
factorizable Lie bialgebras and Poisson--Lie groups as defined in \cite{ReSTS,STS}. \label{ft:factor}}
type.
This is not directly applicable to compact Lie groups,  since they do not admit such $r$-matrices \cite{KS}.
Of course, results pertaining to compact Lie groups can be obtained by exploring the real forms of the holomorphic systems.
Still, we believe it is worthwhile to directly study the compact case by using the quasi-Poisson approach \cite{AKSM} in which no $r$-matrix, but only a non-degenerate bilinear form is needed.

Here, we  present examples for which the method
developed in Section \ref{Sec:2} is immediately applicable.
For these specific examples, our method implies integrability on arbitrary symplectic leaves
of a dense open subset of the moduli space and the existence of action-angle coordinates,  complementing previous results obtained in \cite{ARe,Fe23}.
First, we need to recall the required background material.

\subsection{Quasi-Poisson preliminaries}  \label{ss:qPprelim}

We follow the general construction of \cite{AKSM,AMM}.
In this subsection, we let $K$  be an arbitrary connected, compact, non-Abelian Lie group with a fixed non-degenerate invariant
bilinear form $\langle - , - \rangle$  on its Lie algebra $\fk$.
We choose a  pair of dual bases, $\{E_a\}$ and  $\{E^a\}$, of $\fk$, that satisfy
\be
\langle E_a, E^b \rangle =\delta_{a}^b, \quad
a,b= 1,\dots, \dim(K).
\label{W1}\ee
We denote by $E_a^L$ and $E_a^R$, and similarly by $E^a_L$ and $E^a_R$, the left-invariant and right-invariant vector fields on $K$
that reduce to $E_a$ and $E^a$, respectively, at the unit element $e\in K$.

A quasi-Poisson $K$-manifold with momentum map $\Phi$ is given by the data $(M, A, P_M, \Phi)$, where $A$ is the action map
\be
A: K \times M \to M,
\label{W2}\ee
$P_M$ is a $K$-invariant bivector and  $\Phi: M \to K$ is an equivariant map (with respect to $A$ and $\Ad$)
subject to the following conditions.
First, the Schouten bracket $[P_M,P_M]$ obeys the identity
\be
[P_M, P_M]= - \frac{1}{12} \langle E_a, [E_b,E_c]\rangle \, E^a_M \wedge E^b_M \wedge E^c_M,
\label{W3}\ee
where summation over coinciding indices is applied, and $E^a_M$ is the vector field on $M$
generating the action of the one-parameter subgroup $\exp(t E^a)$ of $K$.
Second, the momentum map and the `quasi-Poisson bracket' defined by the bivector are related according to
\be
\{f, F\circ \Phi \} = \frac{1}{2} E^a_M[f] \, ( E_a^L + E_a^R)[F] \circ \Phi,\qquad \forall f\in C^\infty(M),\,\, F \in C^\infty(K),
\label{W4}\ee
where $E_a^L[F]$ and $E^a_M[f]$ denote the respective derivatives of the functions $F$ and $f$.
We should warn the reader that our convention
for the  vector fields $E^a_M$ is the opposite of that in \cite{AKSM}. This explains the extra sign in \eqref{W3} and
the fact that the momentum map appears in the second argument in \eqref{W4}.
\begin{rem}\label{rem:51}
 If one rescales the invariant bilinear form $\langle - , - \rangle$ by a factor $c\in \mathbb{R}^\ast$,
 then the same map  $\Phi$ gives the momentum map for the quasi-Poisson bivector $c^{-1} P_M$ instead of $P_M$.
\end{rem}

The definition of quasi-Poisson manifold guarantees that $C^\infty(M)^K$ is a genuine Poisson algebra, and this is identified as the Poisson algebra
of smooth functions on $M/K$. Observe also from \eqref{W4} that the momentum map is constant along the integral curves of the `quasi-Hamiltonian vector field' $V_f$ of any
  invariant function $f\in C^\infty(M)^K$.
Thus, the invariant functions of the form $\chi \circ \Phi$ with $\chi \in C^\infty(K)^K$  are in the center of the Poisson algebra $C^\infty(M)^K$.
This entails \cite{AKSM,SL} that $\Phi^{-1}(\cO)/K$ is a (stratified) Poisson subspace of $M/K$ for any conjugacy class $\cO \subset K$
for which $\Phi^{-1}(\cO)$ is not empty.

Now, we state an elementary result, that will prove to be useful.
It uses the notation $\xi_M$  for  $\xi:M\to \fk$, defined as the vector field satisfying
 $\xi_M[f](x)=(\xi(x))_{M,x}[f]$ for any $f\in C^\infty(M)$.
 Here, $(\xi(x))_{M,x}$ is the tangent vector at $x$ of the orbit of
the one-parameter group $\exp(t \xi(x))$.
\begin{lem} \label{lem:Wmomap}
For any $\chi \in C^\infty(K)^K$ the quasi-Hamiltonian vector field $V_{\Phi^\ast\chi}$ satisfies
\be
V_{\Phi^\ast\chi}[f] = (\nabla \chi\circ\Phi)_M[f], \quad \forall f\in C^\infty(M).
\label{Vmomap}\ee
Consequently, the integral curve  of $V_{\Phi^\ast\chi}$
with arbitrary initial value $p\in M$ is given by
\be
\tau \mapsto  A\big(\exp( \tau \, \nabla \chi\circ\Phi(p)), p\big), \qquad \tau \in \bR,
\label{VmomapInt}\ee
and $\Phi$ is constant along the integral curve.
\end{lem}
\begin{proof}
For any $Z\in \fk$, the invariance of $\chi$ yields
$Z^L[\chi]=\langle Z,\nabla \chi\rangle =Z^R[\chi]$. Thus \eqref{W4} entails
\be
V_{\Phi^\ast\chi}[f]
= E^a_M[f] \,\langle E_a,\nabla \chi\circ \Phi\rangle\,, \quad f\in C^\infty(M).
\ee
This directly yields \eqref{Vmomap},  and  we easily deduce \eqref{VmomapInt} from \eqref{Vmomap}.
By the equivariance of the momentum map and the invariance of $\chi$, $\Phi$ is constant along the flow of $V_{\Phi^\ast\chi}$.
\end{proof}

If $(M_i, A_i, P_{M_i}, \Phi_i)$ for $i=1,2$ are quasi-Poisson $K$-manifolds, then the direct product $M= M_1 \times M_2$
becomes such a manifold by applying the fusion procedure.
The fused action is just the diagonal one, whose generating vector fields can be written as
\be
E^a_M(x_1,x_2) = E^a_{M_1}(x_1) + E^a_{M_2}(x_2)
\quad \hbox{using}\quad
T_{(x_1,x_2)} M = T_{x_1} M_1 + T_{x_2} M_2,\,\, \forall (x_1,x_2) \in M,
\label{W5}\ee
and the bivector of $M$ has the form
\be
P_M = P_{M_1} + P_{M_2} + P_\fus
\quad\hbox{with}\quad
P_\fus:= - \frac{1}{2} E^a_{M_1} \wedge E_a^{M_2}.
\label{W6}\ee
The momentum map $\Phi$ on $M$ is the product $\Phi(x_1, x_2) = \Phi_1(x_1) \Phi_2(x_2)$.
One denotes $M_1 \times M_2$ with the fused quasi-Poisson structure as  $M_1 \circledast M_2$.
The fusion procedure is associative.
It is not commutative, but $M_1\circledast M_2$ is  equivalent to $M_2\circledast M_1$ via a certain automorphism \cite[Prop.~5.7]{AKSM}.

\begin{lem} \label{Lem:Fus}
Consider the quasi-Poisson manifold $M:=M_1 \circledast M_2$ obtained by fusion, and denote the natural projection on the $i$-th factor, $i=1,2$, by $\pi_i:M \to M_i$.
For any $f_1 \in C^\infty(M_1)^K$ the quasi-Hamiltonian vector field $V_{\pi_1^\ast f_1}$ is such that
$V_{\pi_1^\ast f_1}[h] = 0$ for any $h\in \pi_2^\ast C^\infty(M_2)$.
In particular, the integral curve of $V_{\pi_1^\ast f_1}$
with arbitrary initial value $(p_1,p_2)\in M$ is given by
\be
\tau \mapsto  (\phi_1(\tau;p_1),p_2), \quad \tau \in \bR,
\label{W:LemFus}
\ee
where $\phi_1(\tau; p_1)$ is the integral curve of $V_{f_1}$ in $M_1$.
A similar statement holds for $f_2 \in C^\infty(M_2)^K$ and the quasi-Hamiltonian vector field $V_{\pi_2^\ast f_2}$ on $M$.
\end{lem}
\begin{proof}
The $K$-invariance of $f_1$ (before fusion) guarantees the vanishing of $P_\fus^\sharp(d\,\pi_1^\ast f_1)$.
Hence, the results readily follow from the form \eqref{W6} of the quasi-Poisson bivector on $M$.
\end{proof}

The simplest example of a quasi-Poisson manifold is $K$ with the conjugation action, the bivector
\be
P_K = \frac{1}{2} E_a^R \wedge E^a_L,
\label{W7}\ee
and the identity map as the momentum map.
This structure can be restricted on the conjugacy classes of $K$, which are quasi-Poisson submanifolds.
Another example that we need is the so-called internally fused double of $K$, denoted $\bD(K)$.
As a $K$-manifold, we have
\be
\bD(K) = K \times K = \{(A,B)\}.
\label{W8}\ee
It carries the diagonal $K$-action and the momentum map
\be
\Phi_{\bD(K)}(A, B) =   AB A^{-1} B^{-1}=: [A,B].
\label{W9}\ee
The bivector of this double is provided by
\be
  P_{\bD(K)}= \frac{1}{2} \bigl(E_a^{(1),R}\wedge E^a_{(1),L} - E_a^{(2),R}\wedge E^a_{(2),L} +
E_a^{(1),L}\wedge (E^a_{(2),L}+ E^a_{(2),R}) +
 E_a^{(1),R}\wedge (E^a_{(2),L} - E^a_{(2),R}) \bigr),
\label{W10}\ee
where $E_a^{(1),L}$ and $E^a_{(2),R}$ denote the
left-invariant and right-invariant vector fields associated with the first and second copy of $K$, respectively.

In \cite[Ex.~6.2]{AKSM}, a  quasi-Poisson description of the moduli spaces of flat connections   on
trivial principal $K$-bundles over Riemann surfaces has been developed.
Namely, for  a Riemann surface $\Sigma_{m,n}$  of genus $m\geq 0$ with  $n\geq 1$ boundary components,
the relevant moduli space can be realized as a reduction of the quasi-Poisson manifold
\be
M_{m,n} = \bD(K) \circledast \dots \circledast \bD(K) \circledast K \circledast \dots \circledast K,
\label{W11}\ee
 obtained from $m$ copies of $(\bD(K), P_{\bD(K)})$ and $n$ copies of $(K, P_K)$ by successive fusion.
Writing the elements $X\in M_{m,n}$ as
\be
X= (A_1, B_1, \dots, A_m, B_m, C_1, \dots, C_n),
\label{W12}\ee
the action of $\eta\in K$ reads
\be
X \mapsto (\eta A_1 \eta^{-1}, \eta B_1 \eta^{-1}, \dots, \eta A_m \eta^{-1}, \eta B_m \eta^{-1}, \eta C_1 \eta^{-1}, \dots, \eta C_n \eta^{-1}),
\label{etaonMmn}\ee
the momentum map of $M_{m,n}$ is
\be
\Phi_{m,n}(X) = [A_1,B_1] \cdots [A_m,B_m] C_1 \cdots C_n,
\label{W13}\ee
and the moduli space results as
\be
\Phi_{m,n}^{-1}(e)/K.
\label{W14}\ee
According to the general theory, $C^\infty(\Phi_{m,n}^{-1}(e))^K$ is a Poisson algebra.
Via the  identification
\be
C^\infty( \Phi_{m,n}^{-1}(e)/K) \equiv C^\infty(\Phi_{m,n}^{-1}(e))^K,
\label{W15}\ee
this yields the (stratified) Poisson structure of the moduli space.

Notice that $\Phi_{m,n}^{-1}(e)$ is an embedded submanifold of $M_{m,n}$, which is diffeomorphic
to $M_{m,n-1}$ by the $K$-equivariant map
\be
M_{m,n-1} \ni u \mapsto (u, \Phi_{m,n-1}(u)^{-1}) \in \Phi_{m,n}^{-1}(e).
\label{W16}\ee
Thus, the rings of invariants $C^\infty(\Phi_{m,n}^{-1}(e))^K$ and  $C^\infty(M_{m,n-1})^K$ are isomorphic.
In fact, as shown by Lemma \ref{lem:Z0} below, the isomorphism holds also at the level of Poisson algebras.
This is a variant of the shifting trick of symplectic reduction, which is valid
also in the quasi-Poisson setting \cite{AKSM}.
By using it, we will  replace the Poisson algebra \eqref{W15} by its  model
\be
C^\infty(M_{m,n-1})^K \equiv C^\infty(M_{m,n-1}/K).
\label{W17}\ee

\begin{lem}\label{lem:Z0}
The $K$-equivariant diffeomorphism \eqref{W16} gives rise to an
 isomorphism between the Poisson algebras $C^\infty(\Phi_{m,n}^{-1}(e))^K$ and $C^\infty(M_{m,n-1})^K$.
\end{lem}
\begin{proof}
Using the tautological inclusion $\iota: \Phi_{m,n}^{-1}(e) \to M_{m,n}$,
the Poisson bracket of $f_1, f_2\in C^\infty(\Phi_{m,n}^{-1}(e))^K$ can be calculated as follows.
 One arbitrarily extends $f_1, f_2$ to  $F_1,F_2 \in C^\infty(M_{m,n})^K$ for which $f_i = F_i \circ \iota$, and then one has
\be
\{ f_1,f_2\} =  (dF_1 \otimes d F_2, P_{m,n})  \circ \iota.
\label{W18}\ee
Now, we  can use the trivial extensions $F_1, F_2$ that satisfy
\be
F_i(u,C_n): = f_i(u, \Phi_{m,n-1}(u)^{-1}), \qquad \forall (u,C_n) \in M_{m,n}.
\label{W19}\ee
Since the functions $F_i$ are independent of $C_n$,  putting them into the formula \eqref{W18} using \eqref{W6} gives
\be
(dF_1 \otimes dF_2, P_{m,n}) = (d F_1 \otimes dF_2, P_{m,n-1}).
\label{W20}\ee
This function on $M_{m,n}$ is independent of $C_n$ and provides $\{f_1, f_2\}$ by means of \eqref{W18}.
On the other hand, $F_1, F_2$ can be regarded as elements of $C^\infty(M_{m,n-1})^K$, and
then the right-hand side of \eqref{W20}  represents the Poisson bracket on $C^\infty(M_{m,n-1})^K$.
Therefore, the mapping \eqref{W16}  engenders the claimed Poisson isomorphism.
\end{proof}

\begin{rem}\label{rem:W0}
We recall \cite{AKSM} that (stratified) symplectic subspaces of $\Phi_{m,n}^{-1}(e)/K$ result by replacing $M_{m,n}$ \eqref{W11}  with
\be
\widehat M_{m,n} :=
\bD(K) \circledast \dots \circledast \bD(K) \circledast \cO_1 \circledast \dots \circledast \cO_n,
\ee
using conjugacy classes $\cO_1,\dots, \cO_n$ of $K$ for which $\Phi_{m,n}^{-1}(e) \cap \widehat M_{m,n}\neq \emptyset$.
These stratified symplectic spaces are recovered by replacing $M_{m,n-1}$ with
\be
\widehat M_{m,n-1} :=
\bD(K) \circledast \dots \circledast \bD(K) \circledast \cO_1 \circledast \dots \circledast \cO_{n-1},
\ee
and then taking the reduced phase space $\hat \Phi_{m,n-1}^{-1}( \cO_n^{-1})/K$, where
$\hat \Phi_{m,n-1}$ is the restriction of $\Phi_{m,n-1}$ on $\widehat M_{m,n-1}$ and
$\cO_n^{-1}$ is the pointwise inverse of $\cO_n$.
\end{rem}
With the above preliminaries at hand,  we are going to treat a series of examples.

\subsection{The sphere with four holes}  \label{ss:Sphere4}

From now on, like in Sections \ref{Sec:Cot} and \ref{Sec:Heis},
$K$ is taken to be a connected and simply connected compact Lie group with simple Lie algebra $\fk$.

Taking advantage of Lemma \ref{lem:Z0}, the Poisson structure on the algebra of $K$-invariant functions on $\Phi_{0,4}^{-1}(e)$ is the one on $C^\infty(M_{0,3})^K$.
Thus, we start with the manifold
$M_{0,3}  =\{ (C_1,C_2,C_3)\}$ \eqref{W11}
equipped with the bivector $P_{0,3}$ obtained as the 3-fold fusion product  of the standard quasi-Poisson structure $(K,P_K)$ \eqref{W7}.

\begin{defn} \label{def:fHsphere}
Let $\fH \subset C^\infty(M_{0,3})^K$ be the set of functions of the form
\be
H(C_1,C_2,C_3) = \chi(C_1 C_2), \qquad \forall \chi \in C^\infty(K)^K.
\label{Z16}\ee
\end{defn}
Combining Lemmas  \ref{lem:Wmomap} and  \ref{Lem:Fus}, we find that the integral curve of the vector field associated with $H$ \eqref{Z16} by means of the bivector $P_{0,3}$,
with arbitrary initial value $(c_1,c_2,c_3)$, is given by
\be
(C_1(\tau), C_2(\tau), C_3(\tau)) = (e^{\tau \nabla \chi(c_1 c_2)} c_1  e^{-\tau \nabla \chi(c_1 c_2)}, e^{\tau \nabla \chi(c_1 c_2)} c_2
e^{-\tau \nabla \chi (c_1 c_2)}, c_3).
\label{Z17}\ee
To apply our general scheme, we introduce the  dense open  connected  submanifold $Y \subset M_{0,3}$ by
\be
Y:= \{ (C_1,C_2,C_3) \in M_{0,3} \mid C_1 C_2 \in K^\reg\}.
\label{Z18}\ee

Next, we introduce $\ell$ $K$-invariant functions on $Y$ as follows.
For the simple roots \eqref{J2} of $\fk^\bC$,  consider the fundamental coweights
\be
\omega^\vee_1,\ldots,\omega^\vee_\ell \in \ri \ft, \qquad
\alpha_k(\omega^\vee_j) = \delta_{jk} \quad  \forall 1\leq j,k\leq \ell.
\label{LI1}
\ee
Then, recalling the mapping $\underline{\chi}: K^\reg \to \cA$ from \eqref{J10},
define
 \be
\Xi_j :=
\langle \omega^\vee_j, \underline{\chi} \rangle,
\qquad \forall j=1,\dots, \ell.
\label{LI5}\ee
These are real-analytic, $K$-invariant real functions on $K^\reg$.
By using them, we introduce
\be
(H_1^{\ad},\dots, H_\ell^{\ad}): Y \to \bR^\ell \quad \hbox{with}\quad H_j^{\ad}(C_1,C_2,C_3) := \Xi_j(C_1 C_2).
\label{LI10}\ee
To describe the flows of these Hamiltonians, we shall use the maximal torus $\bT^\ad = \bT/\cZ(K)$
of the adjoint group $K^\ad = K/\cZ(K)$, denoting
the projection from $\bT$ to $\bT^\ad$ simply as $T \mapsto [T]$.

 \begin{rem} \label{rem:Tad}
In terms of the generators of $\fk^\bC$ \eqref{J4}, we have for $\omega_k^\vee$ \eqref{LI1} and $\tau_k \in \bR$
\be
\exp\Bigl(\ri  \sum_{k=1}^\ell \tau_k \ad_{\omega^\vee_k}\Bigr)( e_{\alpha_j})
= \Ad_{\exp(\ri  \sum_{k=1}^\ell \tau_k \omega^\vee_k )}(e_{\alpha_j})
= \exp(\ri \tau_j )\,  e_{\alpha_j} \,.
\label{LI2}
\ee
This implies that $\exp(\ri  \sum_{k=1}^\ell \tau_k\omega^\vee_k) \in K$
acts trivially in the adjoint representation of $K$, and thus belongs to the center $\cZ(K)$ of $K$,
if and only if all $\tau_k$ are equal to $0$ modulo $2\pi$.
Hence, similarly to \eqref{J30}--\eqref{J31}, we may define an isomorphism
\be
[T_{\omega^\vee}]: \bR^\ell / (2 \pi \bZ)^\ell \to \bT^\ad \quad\hbox{with}\quad
T_{\omega^\vee}(\taU):= \exp\bigl( -\ri \sum_{j=1}^\ell \tau_j \omega^\vee_j \bigr),
\quad \forall \taU=(\tau_1,\dots, \tau_\ell) \in \bR^\ell.
\label{LI3}\ee
In particular, this means that $[T_{\omega^\vee}(\taU)]$ acts on $K$ via conjugation by $T_{\omega^\vee}(\taU)$.
The maps introduced in \eqref{J30}--\eqref{J31} and in \eqref{LI3} reflect the facts \cite{DK,OV}
that $\bT$ and $\bT^\ad$ are
the quotients of $\ft$ by $2\pi\ri $-times the coroot and coweight lattices, respectively.
\end{rem}

\begin{rem} \label{rem:Cowei}
Let us define the $\ell \times \ell$ matrix $Q=(Q_{jk})$ by the
following expansion:
\be
\omega^\vee_j = \sum_{k=1}^\ell Q_{j,k} h_{\alpha_k}\,.
\label{LI4}
\ee
  The matrix $Q$ has rational entries. Indeed, applying $\alpha_k$ to \eqref{LI4} and forming the corresponding matrix, we obtain
$\1_\ell = Q \, C$, with $C$ the transposed Cartan matrix of $\fk^\bC$.
Explicit formulas for the matrix $Q$ can be found in \cite{OV}.
The definition \eqref{LI5} and the relation \eqref{LI4} entail
\be
\Xi_j = \sum_{k=1}^\ell Q_{j,k} \chi_k,
\ee
 with $\chi_j$ \eqref{J12}.
Using Lemma \ref{lem:J3} we obtain the derivatives
 \be
\nabla \Xi_{j}(g) = -\Gamma_2(g)^{-1} \ri \omega^\vee_j  \Gamma_2(g),\,\, \forall g\in K^\reg.
 \label{LI6}\ee
 \end{rem}
Combining this with the formula \eqref{Z17} proves the statement of the next lemma.

\begin{lem}\label{lem:Z5ad}
The joint flows of the Hamiltonians $H_j^{\ad}$ \eqref{LI10} yield an action of the maximal adjoint torus $\bT^{\ad}$ on $Y$.
To present it, we use $T_{\omega^\vee}(\taU)$ \eqref{LI3} and define  $g_{\taU}( C_1 C_2)\in K$ by
\be
 g_{\taU}( C_1 C_2) := \Gamma_2(C_1 C_2)^{-1} T_{\omega^\vee}(\taU) \Gamma_2(C_1 C_2),
\label{Z20}\ee
applying the `diagonalizer' $[\Gamma_2]$  \eqref{J11}.
Then, the $\bT^{\ad}$-action on $Y$ \eqref{Z18}  is given by the mapping
 \be
\bT^{\ad} \times Y \ni ([T(\taU)], (C_1,C_2,C_3)) \mapsto
 ( \Ad_{g_{\taU}(C_1C_2)}(C_1) , \Ad_{g_{\taU}(C_1C_2)}( C_2),  C_3) \in Y.
\label{Z21}\ee
\end{lem}

 To avoid any possible confusion, it might be worth noting that we use $\Ad$ to denote both the adjoint representation of $K$ on $\fk$
and its respective action on $K$.

It is easy to see from \eqref{Z21} that the
principal isotropy group for the $\bT^{\ad}$-action on $Y$ is trivial, i.e., $\bT^{\ad}$ acts effectively
on $Y$.
Clearly, the principal isotropy type for the $K$-action on $Y$ is given by $\cZ(K)$.
Because the flow of any invariant function commutes with the $K$-action, the actions
of $K$ \and $\bT^{\ad}$ combine into an action of the direct product group $K \times \bT^{\ad}$ on $Y$.

\begin{lem}\label{lem:Z6}
The principal isotropy group for the combined action of $\bG:= K \times \bT^{\ad}$
on  $Y$ \eqref{Z18} is $\cZ(K) \times \{e\}$.
\end{lem}
 \begin{proof}
Since $\cZ(K) \times \{e\}$ is obviously contained in every
isotropy group, it is sufficient to exhibit a single point having this isotropy group.
For this purpose, let us take a maximal torus $\bT'$ for which $\bT' \cap \bT = \cZ(K)$,
and consider a triple $(C_1, C_2, C_3)$ satisfying the following conditions:
\be
C_2, C_3 \in \bT'\cap K^\reg, \quad
C_1 C_2 \in \exp(\ri \cA) \subset \bT^\reg,
\label{Z22}\ee
where $\cA$ is the Weyl alcove \eqref{J8}.  For such a point $g_{\taU}(C_1 C_2) = T_{\omega^\vee}(\taU)$, cf. \eqref{LI3} and \eqref{Z20}.
Furthermore,
$(\eta, [T_{\omega^\vee}(\tau)]) \in K \times \bT^{\ad}$ fixes $(C_1,C_2, C_3)$ if and only if we have the equality
\be
(\eta^{-1} C_1 \eta, \eta^{-1} C_2 \eta,\eta^{-1} C_3 \eta)
= (T_{\omega^\vee}(\taU) C_1 T_{\omega^\vee}(\taU)^{-1},  T_{\omega^\vee}(\taU) C_2 T_{\omega^\vee}(\taU)^{-1}, C_3).
\label{Z23}\ee
 Since $C_1 C_2 \in \exp(\ri \cA)$, this equality implies
\be
\eta^{-1} C_1 C_2 \eta = C_1 C_2.
\label{Z24}\ee
Combining this with $\eta^{-1} C_3 \eta = C_3$, we get $\eta \in \bT \cap \bT' = \cZ(K)$.
Then \eqref{Z23} gives $C_2 = T_{\omega^\vee}(\taU) C_2 T_{\omega^\vee}(\taU)^{-1}$.
Consequently, $T_{\omega^\vee}(\taU) \in \bT\cap \bT'$, i.e.  $T_{\omega^\vee}(\taU)\in \cZ(K)$ and its class in $\bT^{\ad}$ is the identity.
 \end{proof}

\begin{prop}\label{pr:Zsphere}
Theorem \ref{thm:G12} is applicable to  the Abelian Poisson algebra $\fH$ introduced in Definition \ref{def:fHsphere} with
 $M=M_{0,3}$,
$Y\subset M$ \eqref{Z18}, $(H_1^{\ad},\ldots,H_\ell^{\ad})$ \eqref{LI10} and
$\Gone=K$, $\Gtwo=\bT^{\ad}$ acting as described above.
Consequently,  $\fH$  engenders integrable systems of rank $\ell$ whose action variables arise from
$H_1^{\ad},\ldots,H_\ell^{\ad}$.
\end{prop}
\begin{proof}
Using Lemmas \ref{lem:Z5ad} and \ref{lem:Z6}, one can prove as in Proposition \ref{pr:J10} that all assumptions of Scenario \ref{scen:G11} are satisfied.
Therefore  Theorem \ref{thm:G12} can be applied.
\end{proof}

In this simple case, we can easily exhibit the algebra of $K$-invariant constants of motion of $\fH$.
To do so, define the surjective maps
\begin{equation*}
 \begin{aligned}
\Psi_{1,2}: M_{0,3} \to K \times K, \qquad
&\Psi_{1,2}(C_1,C_2, C_3) := (C_1,C_2),
\\
\Psi_{(12),3}: M_{0,3} \to K \times K, \qquad
&\Psi_{(12),3}(C_1,C_2,C_3) := (C_1 C_2, C_3).
\end{aligned}
\end{equation*}
Clearly, these are also quasi-Poisson maps if we view the target as $K\circledast K$.
It is obvious using \eqref{Z21} that they induce the constants of motion:
\be
\Psi_{1,2}^*( C^\infty(K \times K)^K)
\quad
\hbox{and}\quad
\Psi_{(12),3}^* (C^\infty( K\times K)^K).
\ee
The intersection of these two sets of constants of motion is $\fH$.
Moreover,
\be
\ddim \left(\Psi_{1,2}^*( C^\infty(K \times K)^K)\right) = \ddim\left(\Psi_{(12),3}^*( C^\infty(K \times K)^K)\right)= \dim(K).
\ee
These constants of motion generate all constants of motion,
given by the centralizer $\fF^K$ of $\fH$ in $C^\infty(M_{0,3})^K$.
The functional dimensions add up to
\be
\ddim(\fF^K) + \ddim(\fH) = 2\dim(K) = \dim((M_{0,3})_\ast/K),
\ee
in agreement with equation \eqref{212}, which implies \eqref{prop213}.

Let us emphasize that we have obtained a \emph{degenerate} integrable system on
the Poisson manifold $(M_{0,3})_*/K$.
Indeed,
after fixing the values of the Casimir functions obtained from the momentum map on $M_{0,3}$ and from $C_1,C_2,C_3$, generically
one obtains a symplectic leaf of dimension $2\dim(K)-4\ell$ in $(M_{0,3})_*/K$.
Except for $\operatorname{SU}(2)$,
$
\ddim(\fH) = \ell < \dim(K) - 2\ell.
$
Together with Lemma \ref{lem:214}, this shows degenerate integrability on $(M_{0,3})_\ast/K$ in the sense of Definition \ref{defn:int}.

\subsection{The torus with one hole}  \label{ss:Torus1}

In the present subsection, our  `unreduced phase space' $M$ is the internally fused double $M:=\bD(K)$ described around
equations \eqref{W8}--\eqref{W10}.
Using that as a manifold $M = K \times K = \{(A,B)\}$,
we let $p_1$ and $p_2$ denote the projections from $M$ onto the first and second $K$ factors.
Then, we introduce the Poisson subalgebras  $\fH$ and $\tilde \fH$ of the Poisson algebra  $C^\infty(M)^K$ by
\be
\fH:= p_1^*(C^\infty(K)^K)
\quad
\hbox{and}
\quad
\tilde \fH := p_2^*(C^\infty(K)^K).
\label{Z5}\ee

For any $\chi\in C^\infty(K)^K$, the integral curves of the quasi-Hamiltonian vector
field $V_{p_1^*\chi}$ defined by means of the bivector \eqref{W10}
have the form
\be
 (A(\tau), B(\tau)) = (A_0, B_0 \exp( - \tau \nabla \chi(A_0))),
 \label{Z6}\ee
and the integral curves of the vector field $V_{p_2^*\chi}$ have the form
\be
 (A(\tau), B(\tau)) = (A_0 \exp( \tau \nabla \chi(B_0)), B_0).
 \label{Z7}
\ee
These formulae  show that the pair $(A, B A B^{-1})$ is constant along the integral curves  \eqref{Z6}
and the pair $(A B A^{-1}, B)$ is constant along the integral curves \eqref{Z7}.
It can be shown \cite{FeCONF} using $K$-invariant functions of these constants of motion that
the Abelian Poisson algebras $\fH$ and $\tilde \fH$ descend
to degenerate integrable systems on $\bD(K)_*/K$, where $\bD(K)_*\subset \bD(K)$ is the principal isotropy type submanifold
for the $K$-action by simultaneous conjugation.
Now, we briefly explain  how this and stronger results follow from the method developed in the present paper.

Proceeding along the path familiar from Sections \ref{Sec:Cot} and \ref{Sec:Heis}, we define
 \be
Y:= p_1^{-1}(K^\reg) \quad \hbox{and} \quad \tilde Y:= p_2^{-1}(K^\reg),
\label{Z8}\ee
which are dense open connected submanifolds,
and the $K$-invariant maps
\be
 (H_1,\dots, H_\ell): Y \to \bR^\ell
\quad \hbox{and}
\quad
(\tilde H_1,\dots, \tilde H_\ell) : \tilde Y \to \bR^\ell
\label{Z9}\ee
with
\be
H_j := \chi_j \circ p_1
\quad\hbox{and}\quad
\tilde H_j := \chi_j  \circ p_2, \qquad j=1,\dots, \ell.
\label{Z10}\ee
The $\chi_j$ are the functions introduced in  Lemma \ref{lem:J2}.

\begin{lem}\label{lem:Z1}
The joint flows of the quasi-Hamiltonian vector fields of the components of
 $(H_1,\dots, H_\ell)$ \eqref{Z10} yield a free  action of the torus $\bT$ \eqref{J30} on $Y$ \eqref{Z8}, whose
action map is given by
\be
(T(\taU), (A,B)) \mapsto (A,  B \Gamma_2(A)^{-1}   T(-\taU) \Gamma_2(A)),
\label{Z11}\ee
with $T(\taU)$ in \eqref{J31}.
In the same way, the map $(\tilde H_1,\dots, \tilde H_\ell)$ \eqref{Z10} generates the free action of
$\bT$ \eqref{J30} on $\tilde Y$ \eqref{Z8} that
operates according to
\be
(T(\taU), (A,B)) \mapsto (A \Gamma_2(B)^{-1} T(\taU) \Gamma_2(B),  B) .
\label{Z12}\ee
These torus actions commute with the $K$-action \eqref{etaonMmn} restricted on $Y$ and on $\tilde Y$, respectively.
The principal isotropy group for the $K$-action on $Y$ as well as for the $K$-action on $\tilde Y$ is $\cZ(K)$,
and the principal isotropy group for the action of $\bG = K \times \bT$ is $\cZ(K) \times \{e\}$ in both cases.
\end{lem}
\begin{proof}
This is essentially the same as the proof of the first part of Lemma \ref{lem:J6} and the proofs of Lemmas \ref{lem:J7} and \ref{lem:J8}.
For example, when working with $(H_1,\dots, H_\ell)$ and $Y$, the elements $A$ and $B$ play roles analogous to $J$ and $g$, respectively.
\end{proof}

The proof of the next proposition follows again from our general formalism.

\begin{prop} \label{pr:Z2}
Theorem \ref{thm:G12} is applicable to  the Abelian Poisson algebra $\fH$ \eqref{Z5} with  $M=\bD(K)$,
$Y\subset M$ \eqref{Z8}, $(H_1,\ldots,H_\ell)$ \eqref{Z9} and
$\Gone=K$, $\Gtwo=\bT$ acting as described above.
Consequently,  $\fH$  engenders integrable systems of rank $\ell$ whose action variables arise from $H_1,\dots, H_\ell$.
Moreover, the analogous statements hold if we replace $(Y,\fH)$ by $(\tilde Y,\tilde \fH)$ and use the relevant Hamiltonians from \eqref{Z9}.
\end{prop}

\begin{rem}\label{rem:Z3}
It would have been enough to consider only one of $(Y,\fH)$ and $(\tilde Y,\tilde \fH)$
since they are exchanged by the automorphism $\mathfrak{S}: (A,B) \mapsto (B^{-1}, B^{-1} A B)$ of the
internally fused double.
The map $\mathfrak{S}$  descends to the reduced phase space where it constitutes a building block of the
action of the $\mathrm{SL}(2,\bZ)$ mapping class group.
The significance of the transformation $\mathfrak{S}$ from the viewpoint of  `Ruijsenaars duality' (action-position duality)
for integrable systems on $M/K$ is explained, for example,  in  the papers \cite{Fe23,FK2}.
This also exemplifies the general fact that the action of the mapping class group of the underlying Riemann surface \cite{FM}
induces many equivalences between the integrable
systems constructed on the corresponding moduli spaces of flat connections \cite{ARe}.
\end{rem}

\begin{rem} \label{Rem:TrivPhi}
Recall that $\bD(K)$ admits the momentum map $\Phi:(A,B)\mapsto [A,B]$.
Thanks to Lemma \ref{lem:Wmomap},
for any $\chi\in C^\infty(K)^K$ the integral curves of the quasi-Hamiltonian vector
field $V_{\Phi^*\chi}$ are given as
\be
 (A(\tau), B(\tau)) = \big(\exp( \tau \nabla \chi(\Phi_0)) A_0\exp(-\tau \nabla \chi(\Phi_0)),
 \exp( \tau \nabla \chi(\Phi_0))B_0 \exp(-\tau \nabla \chi(\Phi_0)) \big),
 \label{Z-D1}
\ee
where $\Phi_0=\Phi(A_0,B_0)$.
In combination with Lemma \ref{Lem:Fus},
we shall use them in the following Subsection when constructing integrable systems on $M_{2,0}$.
There, it will be important that the open subspace $\Phi^{-1}(K^\reg)$ of $\bD(K)$ is dense and connected.
 The proof of this statement is quite technical, and is presented in Appendix \ref{App:conn}.
 \end{rem}

\subsection{The genus 2 surface with one hole}  \label{ss:gen2}

We consider the `double torus' with one boundary component.
The corresponding moduli space can be identified with $M_{2,0}/K$, where we recall
\be
M_{2,0}=\bD(K) \circledast \bD(K) .
\label{DD1}
\ee
As a manifold, $M_{2,0}=K^4=\{(A_1,B_1,A_2,B_2)\}$.

 \begin{defn}\label{defn:DD1}
 Let $\fH$ denote the linear subspace of $C^\infty(M_{2,0})^K$ spanned by arbitrary finite products and sums of
 the Hamiltonians of the form
\be
\widehat{H}^{1,\chi}(A_1,B_1,A_2,B_2) := \chi([A_1,B_1])
\quad
\hbox{and}\quad H^{2,\chi}(A_1,B_1,A_2,B_2) := \chi(A_2),
\label{DD2}
\ee
for every $\chi\in C^\infty(K)^K$.
 \end{defn}
The integral curves for the Hamiltonians in \eqref{DD2} are readily obtained from \eqref{Z-D1} and \eqref{Z6}, respectively, on each copy of $\bD(K)$ inside $M_{2,0}$.
It is clear that these flows pairwise commute, hence $\fH$ is an Abelian Poisson algebra.
We also introduce the subspace
\be
Y=\{(A_1,B_1,A_2,B_2) \in M_{2,0} \mid [A_1,B_1]\in K^\reg \text{ and } A_2\in K^\reg\},
\label{DD2bis}
\ee
where the arguments  of $\chi$ appearing in \eqref{DD2} are regular.
Note that $Y$ is a connected dense open submanifold due to Lemma \ref{lem:Z-D} and the connectedness of the subspace \eqref{Z8} in $M_{1,0}$.
Next, we introduce the $K$-invariant maps
\be
 (\widehat{H}_1^1,\dots, \widehat{H}_\ell^1): Y \to \bR^\ell
\quad \hbox{and}
\quad
(H_1^2,\dots, H_\ell^2) : Y \to \bR^\ell ,
\label{DD3}\ee
where, using  $\chi_j$ \eqref{J12} and $\Xi_j$ \eqref{LI5} for $j=1,\dots, \ell$, we set
\be
\widehat{H}_j^1(A_1,B_1,A_2,B_2) := \Xi_j([A_1,B_1])
\quad\hbox{and}\quad
H_j^2(A_1,B_1,A_2,B_2) := \chi_j(A_2).
\label{DD4}\ee
The choice of the $\widehat{H}_j^1$ reflects the fact that the flows of the Hamiltonians $\widehat{H}^{1,\chi}$ act
via conjugations (cf. \eqref{Z-D1}) as in the situation of Subsection \ref{ss:Sphere4}.

\begin{lem}\label{lem:DD1}
The joint flows of the quasi-Hamiltonian vector fields of the components of
\be
\vec{H}:= (\widehat{H}_1^1,\dots, \widehat{H}_\ell^1,H_1^2,\dots, H_\ell^2): Y \to \bR^{2 \ell}
\label{DD5}
\ee
yield an effective  action of the torus $\bT^{\ad}\times\bT$ on $Y$ \eqref{DD2bis}.
To present it, we use $T(\taU)$ \eqref{J31}, $T_{\omega^\vee}(\taU)$ \eqref{LI3},
and define  $g_{\taU}^1=g_{\taU}^1(A_1,B_1)\in K$ with the `diagonalizer' $[\Gamma_2]$  \eqref{J11} by
\be
 g_{\taU}^1 := \Gamma_2([A_1,B_1])^{-1} T_{\omega^\vee}(\taU) \Gamma_2([A_1,B_1]).
\label{DD6}\ee
Then, the action of $([T_{\omega^\vee}(\taU_1)],T(\taU_2))\in \bT^{\ad}\times \bT$ on $Y$  is given by the mapping
\be
(A_1,B_1,A_2,B_2) \mapsto
(\Ad_{g_{\taU_1}^1}(A_1),\Ad_{g_{\taU_1}^1}(B_1), A_2,  B_2 \Gamma_2(A_2)^{-1}   T(-\taU_2) \Gamma_2(A_2)).
\label{DD7}
\ee
This torus action commutes with the $K$-action by simultaneous conjugation restricted on $Y$.
\end{lem}
\begin{proof}
The derivation of \eqref{DD6} and \eqref{DD7} follows the usual arguments. It is then easy to see that this is an effective action\footnote{It also follows from
the next lemma that the torus action is effective.}
of the torus $\bT^{\ad}\times\bT$. This action commutes with the $K$-action since it is given by flows of $K$-invariant quasi-Hamiltonian vector fields.
\end{proof}

\begin{lem} \label{lem:DD2}
The principal isotropy group for the $K$-action on $Y$ is $\cZ(K)$,
and the principal isotropy group for the combined action of $\bG = K \times (\bT^{\ad}\times\bT)$ is $\cZ(K) \times \{e\}$.
\end{lem}
\begin{proof}
Consider an initial value for which $[A_1,B_1]$ and $A_2$ are from $\exp(\ri \cA)\subset \bT^\reg$ \eqref{J8}.
On such initial value, an arbitrary element $(\eta, [T_1], T_2)$ of $K\times \bT^{\ad} \times \bT$ acts
according to
\be
(A_1, B_1, A_2, B_2) \mapsto \eta\cdot (T_1 A_1 T_1^{-1}, T_1 B_1 T_1^{-1}, A_2, B_2 T_2^{-1}).
\ee
We  can add the following further conditions on the initial value:
\begin{enumerate}
\item $A_1 \in \bT^\reg$;
\item $B_1$ is a representative of a Coxeter element of the Weyl group of the pair $(K,\bT)$;
\item $B_2 A_2 B_2^{-1}$ is a regular element from a maximal torus $\mathbb{T}'$ satisfying $\bT\cap \bT'= \cZ(K)$.
\end{enumerate}
The first two conditions imply that $[A_1,B_1] \in \bT$.
We note that $A_1\in \bT^\reg$ can be chosen so that $[A_1,B_1] \in \exp(\ri \cA)$
since all elements of $\bT$
can be written in the form $[A,B_1]$ with $A\in \bT$ and  $B_1$  in item (2)
(see the identity \eqref{F1}).
Now, if $(A_1,B_1, A_2,B_2)$ is fixed by $(\eta, [T_1], T_2)$, then $(A_2, B_2 A_2 B_2^{-1})$ is fixed by $\eta$,
and our choice guarantees that $\eta \in \cZ(K)$. After that, it is obvious that $T_2$ is the identity element and $T_1$
must satisfy
$T_1 B_1 T_1^{-1} = B_1$, which is equivalent to
\be
B_1 T_1 B_1^{-1} = T_1.
\ee
This implies that $T_1\in \cZ(K)$, since the fixed points for any Coxeter element acting on $\bT$
are given precisely by $\cZ(K)$ \cite{DW,Me}. In particular, the class $[T_1]\in \bT^{\ad}$ is the identity.

All in all, we have exhibited a point of $Y$ for which the $K \times (\bT^{\ad} \times \bT)$ isotropy group equals
$\cZ(K)\times \{e\}$,
which entails that this is the principal isotropy group since it is obviously contained in the isotropy group
of every element of $Y$.
\end{proof}

For the next case, introduce the basic Hamiltonians
\be  \widehat{H}^{2,\chi}(A_1,B_1,A_2, B_2) := \chi([A_2,B_2]), \quad \chi\in C^\infty(K)^K.
\label{DD10}
\ee
\begin{defn}\label{defn:DD2}
 Let $\widehat{\fH}$ denote the linear subspace of $C^\infty(M_{2,0})^K$ consisting of polynomials in the basic Hamiltonians
$\widehat{H}^{1,\chi}$ \eqref{DD2} and $\widehat{H}^{2,\chi}$ \eqref{DD10} where $\chi\in C^\infty(K)^K$.
 \end{defn}
Again, one can see that $\widehat{\fH}$ is an Abelian Poisson algebra: the flows associated with the Hamiltonians $\widehat{H}^{1,\chi}$
and  $\widehat{H}^{2,\chi}$ take the form \eqref{Z-D1} on the first and second copy of $\bD(K)$ in $M_{2,0}$, respectively.
As a consequence of Lemma \ref{lem:Z-D}, the following is a connected dense open submanifold of $M_{2,0}$:
\be
\widehat{Y}=\{(A_1,B_1,A_2,B_2) \in M_{2,0} \mid [A_1,B_1]\in K^\reg \text{ and }  [A_2,B_2]\in K^\reg\}\,.
\label{DD11}
\ee
The $K$-invariant map $(\widehat{H}_1^1,\dots, \widehat{H}_\ell^1)$ \eqref{DD3} is well-defined and smooth on $\widehat{Y}$.
Similarly, we have the smooth $K$-invariant mapping
\be
 (\widehat{H}_1^2,\dots, \widehat{H}_\ell^2): \widehat{Y} \to \bR^\ell,
\quad
\widehat{H}_j^2(A_1,B_1,A_2,B_2) := \Xi_j([A_2,B_2]),
\label{DD12}\ee
with $\Xi_j$ \eqref{LI5}. The next result is analogous to Lemma \ref{lem:DD1}.

\begin{lem}\label{lem:DD3}
The joint flows of the quasi-Hamiltonian vector fields of the components of
\be
\vec{\widehat{H}}:=
(\widehat{H}_1^1,\dots, \widehat{H}_\ell^1,\widehat{H}_1^2,\dots, \widehat{H}_\ell^2):
\widehat{Y} \to \bR^{2 \ell}
\label{DD13}
\ee
yield an effective action of the torus $\bT^{\ad}\times\bT^{\ad}$ on $\widehat{Y}$ \eqref{DD11}.
To display it, we use $T_{\omega^\vee}(\taU)$ \eqref{LI3},
and for any $\taU$  we  define $g_{\taU}^1=g_{\taU}^1(A_1,B_1)\in K$ as in \eqref{DD6} and
$g_{\taU}^2=g_{\taU}^2(A_2,B_2)\in K$ by the same formula  with $(A_2,B_2)$ in place of $(A_1,B_1)$.
Then, the action of $([T_{\omega^\vee}(\taU_1)],[T_{\omega^\vee}(\taU_2)])\in \bT^{\ad}\times \bT^{\ad}$ on
$\widehat{Y}$ is given by the mapping
\be
(A_1,B_1,A_2,B_2) \mapsto
(\Ad_{g_{\taU_1}^1}(A_1),\Ad_{g_{\taU_1}^1}(B_1), \Ad_{g_{\taU_2}^2}(A_2),\Ad_{g_{\taU_2}^2}(B_2)).
\label{DD14}
\ee
This torus action commutes with the $K$-action by simultaneous conjugation restricted on $\widehat{Y}$.
\end{lem}

\begin{lem} \label{lem:DD4}
The principal isotropy group for the $K$-action on $\widehat{Y}$ \eqref{DD11} is $\cZ(K)$,
and the principal isotropy group for the combined action of
$\widehat \bG := K \times (\bT^{\ad}\times\bT^{\ad})$ is $\cZ(K) \times \{e\}$.
\end{lem}
\begin{proof}
Fix two maximal tori of $K$ for which $\bT \cap \bT' = \cZ(K)$ and an element $g\in K$ that verifies $g \bT g^{-1} = \bT'$.
Then,  choose a point of $\widehat{Y}$ satisfying:
\begin{enumerate}
\item $[A_1,B_1] \in \exp(\ri \cA)\subset \bT^\reg$  with the Weyl alcove in \eqref{J8} and $[A_2,B_2] \in g \exp(\ri \cA) g^{-1}$;
\item $A_1$ in $\bT^\reg$ and $B_1$ is a representative of a Coxeter element of the pair $(K, \bT)$;
\item $A_2 \in \bT'^\reg$ and $B_2$ is a representative of a Coxeter element of the pair $(K, \bT')$.
\end{enumerate}
Taking the chosen point as initial value,
an arbitrary element  $(\eta,[T_1],[T_2])\in K\times (\bT^{\ad}\times \bT^{\ad})$  acts as
\be
(A_1,B_1,A_2,B_2) \mapsto \eta \cdot (T_1 A_1 T_1^{-1}, T_1 B_1 T_1^{-1}, T_2' A_2 T_2'^{-1}, T_2' B_2 T_2'^{-1}),
\label{DD15}
\ee
where $T_2' = g T_2 g^{-1}$ with $T_1, T_2 \in \bT$.
Due to our assumptions on the chosen point, \eqref{DD15} reduces to
 \be
(A_1,B_1,A_2,B_2) \mapsto  (\eta A_1 \eta^{-1}, \eta T_1 B_1 T_1^{-1}\eta^{-1}, \eta A_2 \eta^{-1}, \eta T_2' B_2 T_2'^{-1} \eta^{-1}).
\label{DD15+}\ee
Thus the combined action of $K \times (\bT^\ad \times \bT^\ad)$ fixes our chosen point only for $\eta \in \cZ(K)$.
Putting this back into \eqref{DD15+}, the condition for the isotropy group becomes
\be
 T_1 B_1 T_1^{-1} \ = B_1
\quad\hbox{and}\quad
  T_2' B_2 T_2'^{-1}  = B_2.
\ee
As in the end of the proof of Lemma \ref{lem:DD2}, we conclude that $T_1\in \cZ(K)$ and $T_2' \in \cZ(K)$; the latter condition is
equivalent to $T_2 \in \cZ(K)$. Hence, both classes $[T_1],[T_2]\in \bT^{\ad}=\bT/\cZ(K)$ are the identity, as desired.
\end{proof}

We are now in the position to derive the next result. Its proof can be deduced from the previous lemmas by applying our standard method.
\begin{prop} \label{pr:DDmain}
Theorem \ref{thm:G12} is applicable to  the Abelian Poisson algebra $\fH$ introduced in Definition \ref{defn:DD1} with
$M=M_{2,0}$, $Y\subset M$ \eqref{DD2bis}, $\vec{H}$ \eqref{DD5} and
$\Gone=K$, $\Gtwo=\bT^{\ad} \times \bT$ acting as described above.
Consequently,  $\fH$  engenders integrable systems of rank $2\ell$ whose action variables arise from $\vec{H}$.
Moreover, the analogous statements hold if
we replace $\fH, Y, \vec{H}$ and $\bT^{\ad} \times \bT$ by
$\widehat{\fH}$ from Definition \ref{defn:DD2}, $\widehat{Y}$ \eqref{DD11},  $\vec{\widehat{H}}$ \eqref{DD13}
and $\bT^{\ad} \times \bT^{\ad}$, respectively.
\end{prop}

\begin{rem} \label{Rem:gen2}
There exist other non-equivalent  choices of the initial Abelian Poisson subalgebra of $C^\infty(M_{2,0})^K$ to which our
formalism is directly  applicable.
Firstly, we may start with $\fH'$ generated by all Hamiltonians of the form
\be
H^{1,\chi}(A_1,B_1,A_2,B_2) := \chi(A_1)
\quad
\hbox{and}\quad H^{2,\chi}(A_1,B_1,A_2,B_2) := \chi(A_2),\quad \forall \chi\in C^\infty(K)^K.
\label{DD2mod}
\ee
Secondly, we may begin either with $\fH \cap \fH'$ or  with $\fH\cap \widehat \fH $ using
$\fH$ and $\widehat \fH$ from Definition \ref{defn:DD1} and \ref{defn:DD2}.
Incidentally, these choices correspond to `refinements' in the sense of \cite{ARe}.
For all these  examples, and to equivalent ones obtained by permuting the two factors of $\bD(K) \circledast \bD(K)$, our Scenario \ref{scen:G11} is
straightforwardly applicable.

One may also examine the larger Abelian subalgebra  of $C^\infty(M_{2,0})^K$ generated
by  the above $\fH'$  and $\widehat \fH$
taken together.  However,  we have not yet studied the connectedness of the open subspace where $A_1$, $A_2$,
$[A_1,B_1]$ and $[A_2, B_2]$ all belong to $K^\reg$,
and leave it to the interested reader to explore this example.
\end{rem}

\subsection{A general class of integrable systems}  \label{ss:genModuli}
Building on the ideas utilized in the preceding examples, we outline
a more general case to which our method is directly applicable.
We begin with the quasi-Poisson manifold $M_{m, n}$ \eqref{W12}  with
non-negative integers $m,n$. We assume that in the $m=0$ case $n\geq 3$.
The relevant bivector, $P_{m,n}$, results from the fusion procedure.
We write the elements of $M_{m,n}$ as
$X = (A_1, B_1, \dots, A_m, B_m, C_1,\dots, C_n)$, like before.

\noindent Now we specify an Abelian Poisson subalgebra of $C^\infty(M_{m,n})^K$.
First, we select disjoint subsets
\be
I,\widehat{I} \subset \{1,2,\dots, m\}, \, I:=\{i_1, i_2,\dots, i_p\}, \, \widehat{I}:=\{j_1, j_2,\dots, j_{\widehat{p}}\},\,
I \cap \widehat{I} = \emptyset,
\label{Z26}\ee
where  $i_1 < i_2 < \dots < i_p$ and $j_1 < j_2 < \dots < j_{\widehat{p}}$. We associate with the set $I$ the Hamiltonians of the form
\be
H_{i_\alpha}^\chi(X) := \chi(A_{i_\alpha}), \quad \forall \chi\in C^\infty(K)^K,\, \forall \alpha=1,\dots, p,
\label{Z27}\ee
and with $\widehat{I}$ the Hamiltonians defined from group commutators as
\be
\widehat{H}_{j_\gamma}^\chi(X) := \chi([A_{j_\gamma},B_{j_\gamma}]), \quad \forall \chi\in C^\infty(K)^K,\, \forall \gamma=1,\dots, \widehat{p}.
\label{Z27bar}\ee
Finally, we also fix a set of non-intersecting closed intervals
\be
J:= \{ [\lambda_1, \rho_1], [\lambda_2, \rho_2], \dots, [\lambda_q, \rho_q] \},
\label{Z28} \ee
 such that the boundaries of the intervals are integers satisfying
 \be
 1\leq \lambda_1 < \rho_1 < \lambda_2 < \rho_2 <\dots < \lambda_q < \rho_q \leq n.
 \label{Z29}\ee
 We associate with $J$ the Hamiltonians
 \be
 H_{[\lambda_\beta, \rho_\beta]}^\chi(X) := \chi(C_{\lambda_\beta} C_{\lambda_\beta + 1}\cdots C_{\rho_\beta}),
 \quad \forall \chi \in C^\infty(K)^K,\,\, \forall \beta =1,\dots, q.
 \label{Z30}\ee

 Non-empty sets of the type  $I,\widehat{I}$ and $J$ exists for generic $m,n$. For specific cases, e.g. for $m=0$ or $n=0$, only one of
 the types of sets exists.
If $m=0$, we shall further assume that $\cup_{\beta=1}^q \{\lambda_\beta,\ldots, \rho_\beta\} \subsetneq \{1,\ldots,n\}$.
If $m=1$ and $n=0$, we assume that $\widehat{I}=\emptyset$
(so as to avoid Casimir functions after reduction, see Remark \ref{Rem:TrivPhi}).

 \begin{defn}\label{defn:Z7}
 Let $\fH(I,\widehat{I},J)$ denote the linear subspace of $C^\infty(M_{m,n})^K$ generated by arbitrary finite products and sums of the Hamiltonians
 \eqref{Z27}, \eqref{Z27bar} and \eqref{Z30}.
 \end{defn}

The integral curves of the quasi-Hamiltonian vector fields of the above  Hamiltonians can be calculated.

\begin{lem}\label{lem:Z8}
Let $(a_1, b_1,\dots, a_m, b_m, c_1, \dots, c_n)$ be an arbitrary initial value, and denote by
\begin{equation*}
 X(\tau) = (A_1(\tau), B_1(\tau), \dots, A_m(\tau), B_m(\tau), C_1(\tau),\dots, C_n(\tau))
\end{equation*}
the integral curve corresponding to
one
of the Hamiltonian from \eqref{Z27}, \eqref{Z27bar} or from \eqref{Z30}.
For any Hamiltonian $H_{i_\alpha}^\chi$ \eqref{Z27}, we have
\be
B_{i_\alpha}(\tau) = b_{i_\alpha} \exp(-\tau \nabla\chi(a_{i_\alpha}))
\label{Z31}\ee
and all other components of $X(\tau)$ are constants.
For any Hamiltonian $\widehat{H}_{j_\gamma}^\chi$ \eqref{Z27bar}, we have
\begin{equation}
 \begin{aligned}
   A_{j_\gamma}(\tau) &=
\exp(\tau \nabla \chi(\Phi_{j_\gamma})) a_{j_\gamma} \exp(-\tau \nabla \chi(\Phi_{j_\gamma})) , \\
 B_{j_\gamma}(\tau) &=
 \exp( \tau \nabla \chi(\Phi_{j_\gamma}))b_{j_\gamma} \exp(-\tau \nabla \chi(\Phi_{j_\gamma})),
\label{Z31bar}
 \end{aligned}
\end{equation}
for $\Phi_{j_\gamma}:=[a_{j_\gamma},b_{j_\gamma}]$, and all other components of $X(\tau)$ are constants.
For any Hamiltonian $H_{[\lambda_\beta, \rho_\beta]}^\chi$ \eqref{Z30}, we have
\be
C_k(\tau) = \exp(\tau \nabla \chi(  c_{\lambda_\beta} c_{\lambda_\beta + 1}\cdots c_{\rho_\beta})) c_k
\exp(-\tau \nabla \chi(  c_{\lambda_\beta} c_{\lambda_\beta + 1}\cdots c_{\rho_\beta})),\quad
\forall \lambda_\beta \leq k \leq \rho_\beta,
\label{Z32}\ee
and all other components of $X(\tau)$ are constants.
 \end{lem}
 \begin{proof}
 Consider the fusion product of arbitrary quasi-Poisson $K$-manifolds $(M_j, P_{M_j})$ for $j=1,\dots, N$,
 \be
 M= M_1 \circledast \dots \circledast M_N,
 \ee
 and a function $H \in C^\infty(M)^K$ that depends only on the factor $M_i$ of $M$.
Then, the $i$-th component of the corresponding integral curve in $(M, P_M)$ is the same as
the one obtained in  $(M_i, P_{M_i})$, and the other components remain constant.
This directly follows from Lemma \ref{Lem:Fus}.
By combining this observation with the associativity of the fusion product, we easily deduce \eqref{Z31} from the integral curves \eqref{Z6} in $\bD(K)$.
Similarly, noting that $[A_{j_\gamma},B_{j_\gamma}]$ and $C_{\lambda_\beta} C_{\lambda_\beta + 1}\cdots C_{\rho_\beta}$ can be seen as momentum
 maps on the corresponding copies of $\bD(K)$ and $K^{\circledast(\rho_\beta-\lambda_\beta+1)}$ inside $M_{m,n}$,
\eqref{Z31bar} and \eqref{Z32} follow from Lemma \ref{lem:Wmomap}.
 \end{proof}

 \begin{rem}\label{rem:Z9}
One can see from Lemma \ref{lem:Z8} that $\fH(I,\widehat{I},J)$ is an Abelian Poisson algebra.
 \end{rem}

For fixed sets $I,\widehat{I}$ and $J$, let
\be
Y(I,\widehat{I},J) \subset M_{m,n}
\label{Z33}\ee
be the dense open submanifold where the arguments of all elements of $\fH(I,\widehat{I},J)$
belong to $K^\reg$.  Note that $Y(I,\widehat{I},J)$ is connected  since as a manifold
\be
Y(I,\widehat{I},J) \simeq K^{s}\times (K^\reg)^{p+q} \times (\Phi_{1,0}^{-1}(K^\reg))^{\widehat{p}}, \quad s:= 2m+n-p-2\widehat{p}-q,
\ee
where $K^\reg$ and $\Phi_{1,0}^{-1}(K^\reg)$ (cf. Lemma \ref{lem:Z-D}) are connected.
Then, define the $K$-invariant map\footnote{If $I$, $\hat I$ or $J$  are empty, then $p=0$, $\widehat{p}=0$ or $q=0$ are understood, respectively.}
\be
\vec{H} : Y(I,\widehat{I},J) \to  \bR^{p \ell+\widehat{p} \ell + q \ell}
\label{Z34}\ee
by utilizing \eqref{Z27} with $\chi_j \in C^\infty(K^\reg)^K$ \eqref{J12}, $j=1,\dots, \ell$,
and \eqref{Z27bar},\eqref{Z30} with $\Xi_j \in C^\infty(K^\reg)^K$ \eqref{LI5}, $j=1,\dots, \ell$, as follows:
The first $p\ell$ components of $\vec{H}$ are the functions
\be
 H_{i_1}^{\chi_1}, \dots,  H_{i_1}^{\chi_\ell}; H_{i_2}^{\chi_1}, \dots,  H_{i_2}^{\chi_\ell} ;\dots ,
 H_{i_p}^{\chi_1}, \dots,  H_{i_p}^{\chi_\ell};
 \label{Z35}\ee
 the next $\widehat{p}\ell$ components are
 \be
 \widehat{H}_{j_1}^{\Xi_1}, \dots,   \widehat{H}_{j_1}^{\Xi_\ell} ; \dots ;
  \widehat{H}_{j_{\widehat{p}}}^{\Xi_1}, \dots,   \widehat{H}_{j_{\widehat{p}}}^{\Xi_\ell} ;
 \label{Z35b}\ee
and the last $q\ell$  components are
\be
H_{[\lambda_1, \rho_1]}^{\Xi_1}, \dots, H_{[\lambda_1, \rho_1]}^{\Xi_\ell}; \dots ;
  H_{[\lambda_q, \rho_q]}^{\Xi_1}, \dots, H_{[\lambda_q, \rho_q]}^{\Xi_\ell}.
 \label{Z35c}\ee

Similarly to the previous examples discussed in detail, we can easily write down
 the formula for the joint flows generated by these Hamiltonians (cf. Lemmas \ref{lem:Z1}, \ref{lem:DD1} and \ref{lem:Z5ad}),
  which yield an action of the torus $\bT^{p}\times (\bT^{\ad})^{\widehat{p}} \times (\bT^{\ad})^q$, of dimension $(p+\widehat{p}+q)\ell$,  on $Y(I,\widehat{I},J)$.
Then, we can  apply our general mechanism for obtaining integrable systems based on the Abelian Poisson algebra
$\fH(I, \widehat{I}, J)$.
This is elaborated in  Appendix \ref{App:Mmn}, where the  final outcome is formulated as Proposition \ref{pr:Z11}.

\begin{rem}
On generic symplectic leaves,
 the systems obtained in Subsection \ref{ss:Torus1} for the torus with one hole represent compact variants of
  trigonometric spin Ruijsenaars--Schneider systems.
  See also \cite{FK2,FKlu} for spinless special cases on small symplectic leaves
  that arise for $K=\mathrm{SU}(n)$.
  For the sphere with four holes in Subsection \ref{ss:Sphere4},
   we expect that these systems give compactified trigonometric spin van Diejen systems, in analogy with
   their holomorphic counterpart \cite{CR}.
The cases with $(m,n)\neq(1,0),(0,3)$ yield generalizations of  the preceding models, which have not yet been related
to many-body systems.
\end{rem}

\section{Conclusion and outlook} \label{S:concl}

In this paper, we  developed a general  mechanism for producing integrable systems by reduction of a master integrable system
engendered by an Abelian Poisson  algebra $\fH$ on a higher dimensional phase space $(M,P_M)$.
 Our reduction mechanism relies on a smooth, proper action  of the direct product Lie group $ G^1 \times G^2$ on a dense open submanifold $Y\subset M$,
where $\Gone$ is a (compact) symmetry group of the unreduced system and the action  of
 $\Gtwo = \mathrm{U}(1)^{\ell_1}\times \bR^{\ell_2}$
 is generated by
`generalized action variables' $H_1,\dots, H_\ell$ defined on $Y$.
We formulated a set of convenient assumptions  in Scenario \ref{scen:G11} and then were able to prove  that $\fH$
induces integrable systems of rank $\ell$ on the Poisson manifold $M_\ast/\Gone$,
where $M_\ast\subset M$ is the principal isotropy type submanifold for the $G^1$-action on $M$,
as well as on symplectic leaves as specified  in Theorem \ref{thm:G12}.
Moreover, the functions $H_1,\dots, H_\ell$ descend to generalized action variables for the induced integrable systems.

Our main finding outlined above was  supported by an analysis of reductions of generalized Hamiltonian torus actions, summarized by Theorem \ref{thm:G7}
and its consequences formulated as Corollaries \ref{cor:G8} and \ref{cor:G9}.  Further general results that could be of interest
to the reader are given by Theorem \ref{thm:AA} and
Corollary \ref{Cor:AA} dealing with generalized action-angle coordinates.

The assumptions of Scenario \ref{scen:G11} were extracted originally from an inspection of examples.
As applications of the ensuing general results,  we here presented a  uniform  treatment of the reductions of
the different doubles of the compact  Lie groups $K$ for
which the reduced integrable systems have been identified  before as spin many-body systems of Calogero--Moser--Sutherland  or Ruijsenaars--Schneider
 type (see \cite{Fe23} and references therein).
Such systems usually come in \emph{dual pairs}, and the study of the one for which   $G^2= \bR^\ell$
was up-to-now incomplete, regarding especially the case of the Heisenberg double.
We fixed this issue as part of Section \ref{Sec:Heis}, see Proposition \ref{pr:H3}.
We also considerably strengthened  the existing results on the reductions of doubles, since previously integrability was proved either on $M_*/K$
or on  generic symplectic leaves \cite{FeCONF,Fe24,Re1,Re2}, while the application of Theorem \ref{thm:G12} gives much more.
  To our knowledge,  the identification of generalized action variables for these systems is a new result as well.

The general mechanism
was also applied in Section \ref{Sec:qPoiss} to discuss, using quasi-Poisson geometry~\cite{AKSM}, integrable systems
 on moduli spaces of flat connections (viewed as character varieties)
 for a connected and simply connected compact Lie group $K$.
We were inspired by the treatment of the complex-algebraic version due to Arthamonov and Reshetikhin \cite{ARe}.
The authors of \cite{ARe} presented a beautiful construction of  the Poisson algebras of Hamiltonians and first integrals
for  integrable systems associated with non-intersecting simple closed  curves on the corresponding
Riemann surface with boundary.  This is in principle applicable to complex as well as to compact Lie groups.
However, the description of the phase space   and the integral curves
was achieved in \cite{ARe}  by adopting the Fock--Rosly \cite{FoR} approach based on a \emph{factorizable} classical $r$-matrix
(see footnote~\ref{ft:factor}), which is not applicable directly to compact Lie groups
since compact Lie groups do not support factorizable $r$-matrices \cite{KS}.
The quasi-Poisson approach  that we applied led to a straightforward description in a number of interesting special cases, and in those
cases we also derived action variables for the relevant integrable systems.
In particular, in separate Subsections of Section~\ref{Sec:qPoiss} we characterized integrable systems  corresponding to the sphere with $4$ holes,
and the genus one and the genus two surfaces with one hole.
In the final Subsection and Appendix \ref{App:Mmn},  we dealt with  a  rich set of examples for arbitrary genus and number of holes,   
see Proposition \ref{pr:Z11}.

Let us end with an outlook on possible future research.
First of all, our method should be applicable to further develop the reduction picture behind  various spin Calogero--Moser--Sutherland
type models that were introduced in the literature, for example in the papers \cite{Fe13,FP1,KLOZ,Re3,Re4}.
The construction of some of these models uses a non-compact symmetry group \cite{Fe13,Re3}.
As was already noted at the end of the Introduction, the compactness of $G^1$ is not essential for our method: what is crucial is that the
($G^1 \times G^2$)-action must be proper.

In all the examples considered in the present paper, one has either
$\Gtwo = \mathrm{U}(1)^{\ell}$ or $\Gtwo =\bR^{\ell}$.
One may wonder how to construct examples with a genuine `toroidal cylinder'
$\Gtwo = \mathrm{U}(1)^{\ell_1}\times \bR^{\ell_2}$ with $\ell_1 \ell_2 \neq 0$.
It appears that examples of this kind may arise from reductions of geodesic systems on cotangent bundles of non-compact simple Lie groups
equipped with their natural bi-invariant pseudo-Riemanniann metric.
However, because the kinetic energy is not (positive) definite, the prospective examples do not appear to be physically relevant.

It could be interesting to extend our approach for treating the superintegrability  of (spin extensions of) Toda type models.
This would require using a direct product group for $G^1$ that includes a non-compact factor, hence it is not covered directly by our Scenario \ref{scen:G11}.
Symmetry reductions appear in the Lagrangian, variational approach to integrable systems as well, see e.g. \cite{CHSV} and references therein.
Since the systems amenable to such Lagrangian approaches include spin Calogero type models, it is natural
to inquire if there is a connection to our approach.

The reductions to (spin) many-body models could give a good testing ground for determining the set-theoretic differences in the chain
\be
Y_0/G^1 \subset Y_0^1/G^1 \subset M_*/G^1.
\ee
Since $G^2$ acts on $Y_0^1/G^1$ and this action is free on $Y_0/G^1$, the difference
$(Y_0^1/G^1) \setminus ( Y_0/G^1)$
should be a union of non-principal, typically lower dimensional,
 $G^2$-orbits.
 We know from earlier works on reductions of doubles of $\mathrm{SU}(n)$
 that this is the case for the symplectic leaves of dimension $2(n-1)$ that lie in $Y_0^1/\mathrm{SU}(n)$ and
 support the trigonometric Suthertland \cite{FeA} and
 Ruijsenaars--Schneider models \cite{FK},
and also for the so-called type I compactifications \cite{FK2,FKlu} of the latter.
The type II compactified trigonometric Ruijsenaars--Schneider models \cite{FKlu} live on symplectic
leaves of $\bD(\mathrm{SU}(n))_\ast/\mathrm{SU}(n)$ that intersect $Y_0/\mathrm{SU}(n)$ and $(Y_0^1\setminus Y_0)/\mathrm{SU}(n)$,
but contain points outside $Y/\mathrm{SU}(n)$, too, where the pertinent matrix $A$
for $(A,B) \in \bD(\mathrm{SU}(n))_\ast$ has some coinciding eigenvalues \cite{FKlu}.
The lower dimensional tori of the trigonometric Sutherland system  have been described in detail
in the recent paper \cite{LMRS}. Related results can also be extracted from  \cite{Ru}.
The structure of the reduced systems on $(M/G^1) \setminus (M_\ast/G^1)$ is wide open,
its exploration in the framework of stratified symplectic spaces \cite{SL} has just started \cite{CJRX}.

There exist several challenging open problems regarding the integrable systems on moduli spaces of flat connections.
In particular, it would be interesting to develop the link between  the constructions in Arthamonov and Reshetikhin \cite{ARe} and the quasi-Poisson approach.
We took an initial step in this direction in Subsection \ref{ss:compAR}, and hope that a complete picture will be worked out in the future.
It would also be interesting to interpret in our language the (real analogue of the) integrable systems constructed by Le Blanc~\cite{LB} for a punctured sphere.
A final problem which we hope to pursue is to further study the sphere with $4$ holes as in Subsection \ref{ss:Sphere4},
but fixing this time the conjugacy classes of $C_1,C_2,C_3$.
Using $K=\mathrm{SU}(n)$,
we  would like  to obtain  in this way a compact real form of the trigonometric van Diejen system \cite{vD}, building
on the recent derivation of the corresponding complex integrable system by Chalykh and Ryan \cite{CR}.

\bigskip

\subsubsection*{Acknowledgements}
We are grateful to S.~Arthamonov, O.~Chalykh,  R.~Fioresi, N.~Reshetikhin, M.~Semenov-Tian-Shansky and P.~Vanhaecke
for useful discussions and correspondence.
The  work of L.F. was supported in part by the NKFIH research grant K134946.
The work of M.F. was completed at IMB which is supported by the EIPHI Graduate School (contract ANR-17-EURE-0002).
We acknowledge funding from  CA21109 - COST Action CaLISTA that allowed us to
participate in the WG2 Workshop in Leeds in June 2025, where this work was presented.

\appendix

\section{The connectedness of \texorpdfstring{$\Phi^{-1}(K^\reg)\subset \bD(K)$}{Phi-1(Kreg)}}
\label{App:conn}

The proof of the following statement also gives us further useful pieces of information.

\begin{lem} \label{lem:Z-D}
 The open subspace $\Phi^{-1}(K^\reg)$ of $\bD(K)$ is dense and connected.
\end{lem}
\begin{proof}
We begin by quoting a useful identity from \cite{Go}. For this,
let  $g_*\in N_K(\bT)< K$   be a representative of a Coxeter element $w_*$ of the Weyl group $N_K(\bT)/\bT$.
We can view $w_*$ as a linear transformation on $\ri \ft$, and it is well known  that $(w_* - \id)$ is invertible.
 Then, one has the identity
\be
g_* \exp\left(\ri (w_* - \id)^{-1} h)\right) g_*^{-1}  \exp\left(-\ri (w_* - \id)^{-1} h \right) = \exp(\ri h),
\qquad \forall h\in \ri \ft.
\label{F1}\ee
This follows from
\be
g_* \exp\left(\ri (w_* - \id)^{-1} h\right) g_*^{-1}= \exp\left(\ri w_* \circ (w_* - \id)^{-1} h \right) =
\exp(\ri h) \exp\left(\ri (w_* - \id)^{-1} h \right).
\label{F2}\ee
The identity \eqref{F1} implies that the momentum map  $\Phi: \bD(K) \to K$, given by
\be
\Phi(A,B)= A B A^{-1} B^{-1} =: [A,B],
\label{F3}\ee
 is surjective \cite{Go}.

We observe that for $\exp(\ri h)\in \bT^\reg$ the element
\be
X:= (g_*, \exp\left(\ri (w_* - \id)^{-1} h\right) )
\label{F4}\ee
belongs to the principal isotropy type submanifold $\bD(K)_* \subset \bD(K)$ with respect to the $K$-action.
Indeed,   \eqref{F1} now implies that  $K_X < \bT$ and for $T\in K_X \cap \bT$ one has $T g_* T^{-1} = g_*$,
which is equivalent to
\be
g_* T g_*^{-1} = T.
\label{F5} \ee
Recalling \cite{DW,Me} that the fixed point set for any Coxeter element acting on $\bT$ is the center of $K$,
\eqref{F5} entails that $T \in \cZ(K)$.

Let $\Phi_*$ denote the restriction of $\Phi$ on the dense, open, connected submanifold $\bD(K)_* \subset \bD(K)$.
 As a consequence of the above observation,
 the image of $\Phi_*$ contains the dense open subset $K^\reg \subset K$.
 Because the action of $K/\cZ(K)$ is free on  $\bD(K)_*$,
the map $\Phi_*: \bD(K)_* \to K$ is a submersion.
Thus, we see that
 \be
  \Phi_*^{-1}(K^\reg) \subset \bD(K)_*
 \label{F6}\ee
 is a dense open subset.
 Since $\bD(K)_* \subset \bD(K)$ is dense open, we may conclude that  $\Phi_*^{-1}(K^\reg)\subset \bD(K)$ and
\be
 \Phi_*^{-1}(K^\reg) \subset \Phi^{-1}(K^\reg)
 \label{F7}\ee
are both dense open subsets.

As is easily verified, the connectedness of a dense open subspace of a topological space implies the
connectedness of the space.
Therefore,  it is enough to prove that
$\Phi_*^{-1}(K^\reg)$ is connected.
To prove this, let $\nu_0$ and $\nu_1$ be two arbitrarily chosen elements of $K^\reg$ and choose also two elements $g_0$ and $g_1$ of $K$ for which
$g_0 \nu_0 g_0^{-1}$ and $g_1 \nu_1 g_1^{-1}$ belong to $\exp(\cA)$ with our fixed Weyl alcove $\cA$ \eqref{J8}.
That is, we have
\be
g_0 \nu_0 g_0^{-1} = \exp(\ri h_0)
\quad \hbox{and}\quad
g_1 \nu_1 g_1^{-1} = \exp(\ri h_1)
\label{F8}\ee
with unique $h_0, h_1 \in \cA$.
Consider a continuous curve $h:[0,1]\to \cA$ that connects $h_0$ to $h_1$
and a  continuous curve $g:[0,1]\to K$ that connects $g_0$ to $g_1$.
 Then, define the  curves $A,B:[0,1]\to K$ by
 \be
 A(t):= g(t)^{-1} g_* g(t)
 \quad \hbox{and}\quad
 B(t) := g(t)^{-1} \exp\left(\ri (w_* - \id)^{-1} h(t)\right) g(t)
\label{F9} \ee
 for which
 \be
 [A(t), B(t)] = g(t)^{-1} \exp(\ri h(t)) g(t).
 \label{F10}\ee
 Because $X$ \eqref{F4} lies in $\Phi_*^{-1}(K^\reg)$ if $e^{\ri h} \in \bT^\reg$,
 the continuous curve
  $t\mapsto (A(t), B(t))$ lies in $\Phi_*^{-1}(K^\reg)$  and by \eqref{F7}
 \be
 [A(0), B(0)] = \nu_0,\quad [A(1), B(1)] =\nu_1.
\label{F11} \ee
 In other words,
 \be
 (A(0), B(0))
  \in \Phi_*^{-1}(\nu_0)
 \quad\hbox{and}\quad
 (A(1), B(1)) \in \Phi_*^{-1}(\nu_1).
 \label{12}\ee
 Therefore, the curve \eqref{F9} connects two points of the respective fibers $\Phi_*^{-1}(\nu_0)$ and $\Phi_*^{-1}(\nu_1)$.
 Every fiber $\Phi_*^{-1}(\nu)$ is a submanifold because $\Phi_*$ is a submersion and the fibers
 are connected by \cite[Thm.~7.2]{AMM}.
 Since for manifolds connectedness and path-connectedness are the same,  any two points of $\Phi_*^{-1}(\nu_0)$ and $\Phi_*^{-1}(\nu_1)$
 can be connected by a continuous curve.
 In conclusion, $\Phi_*^{-1}(K^\reg)$ is path-connected, and hence connected, whereby
  the proof is complete.
\end{proof}

\begin{rem}
Theorem 7.2 of \cite{AMM} states the connectedness of the fibers of the momentum map for connected
quasi-Hamiltonian spaces of connected and simply connected compact Lie groups.
This is applicable to our case since  $\bD(K)$ and its connected open submanifold $\bD(K)_*$ are
also quasi-Hamiltonian manifolds with the momentum map $\Phi$ \cite{AKSM,AMM}.
It is worth noting from the above proof that $K^\reg \subset \Phi(\bD(K)_*)$ holds.
This is probably a proper subset in general.
The examples studied in \cite{FK2,FKlu} show that this is the case for $K= \mathrm{SU}(n)$.
\end{rem}

\section{Integrability related to the general case of \texorpdfstring{$M_{m,n}$}{Mmn}}
\label{App:Mmn}

We work over $M_{m,n}$ \eqref{W11} following the treatment of Subsection \ref{ss:genModuli}.
Firstly, we will prove that one can apply the general mechanism of Theorem \ref{thm:G12} to obtain integrable systems
based on the Abelian Poisson algebra $\fH(I,\widehat{I},J)$ from Definition \ref{defn:Z7}.
Secondly, we will outline some modifications of this construction that should lead to other integrable systems.

\subsection{The systems associated with \texorpdfstring{$\fH(I,\widehat{I},J)$}{H(I,I,J)}}

Introduce the notation
\be
\vec{\bT}:=\bT^{p}\times (\bT^{\ad})^{\widehat{p}} \times (\bT^{\ad})^q
\label{vecT}
\ee
for the torus associated with the flows of the Hamiltonians $\vec{H}$ \eqref{Z34},  and define $\bG:= K \times \vec{\bT}$.
Throughout this appendix, we use the shorthand notation $Y,\fH$ for $Y(I,\widehat{I},J)$ \eqref{Z33} and $\fH(I,\widehat{I},J)$ from Definition \ref{defn:Z7}, respectively. We also recall the following.
\begin{assume}\label{assume:HIJ}
One has $p+\widehat{p}+q>0$, hence $I,\widehat{I}$ and $J$ are not all empty. Moreover,
 \begin{itemize}
 \item If $m=0$, then $n\geq 3$ and $\cup_{\beta=1}^q \{\lambda_\beta,\ldots, \rho_\beta\} \subsetneq \{1,\ldots,n\}$;
 \item If $m=1$ and $n=0$, then $\widehat{p}=0$ (i.e. $\widehat{I}=\emptyset$).
\end{itemize}
\end{assume}

\subsubsection{Principal isotropy groups}

\begin{lem}\label{lem:Preli4}
Assume that $m=0$.
Then the principal isotropy group for the $\bG$-action on  $Y$ is the direct product $\cZ(K) \times \{e\}$.
\end{lem}
\begin{proof}
Recall that, by Assumption \ref{assume:HIJ}, $n\geq 3$ and the intervals in $J$ satisfy
$\cup_{\beta=1}^q \{\lambda_\beta,\ldots, \rho_\beta\} \subsetneq \{1,\ldots,n\}$.
Thus, there exists $\nu \in \{1,\ldots,n\}$ not appearing in any of the intervals $[\lambda_\beta,\rho_\beta]$.
Put
\be
C_{(\beta)}:=C_{\lambda_\beta} C_{\lambda_\beta+1} \cdots C_{\rho_\beta}, \qquad \beta=1,\ldots,q.
\label{Prel4a}
\ee
We consider a point $X\in Y$ satisfying
\be
C_\nu,C_{\rho_1},C_{\rho_2},\ldots,C_{\rho_q}  \in \bT'\cap K^\reg, \quad
C_{(1)},C_{(2)},\ldots,C_{(q)} \in \exp(\ri \cA) \subset \bT^\reg,
\label{Prel4d}
\ee
where $\bT'$ is a maximal torus such that $\bT\cap \bT'=\cZ(K)$
and $\cA$ is the open Weyl alcove \eqref{J8}.
One can obtain from \eqref{Z32} that the torus $(\bT^\ad)^q$ acts on $Y$ through
\be
([T_{(1)}],\ldots,[T_{(q)}],X)\mapsto X(\taU)
\ee
(with
$T_{(\beta)}:=T_{\omega^\vee}(\taU)$ in \eqref{LI3} corresponding to the $\beta$-th factor in $(\bT^\ad)^q$),
where all components of $X(\taU)$ are those of $X$ except
\be
C_j(\taU) = T_{(\beta)} C_j  T_{(\beta)}^{-1}, \quad
\lambda_\beta \leq j \leq \rho_\beta, \,\, 1\leq \beta\leq q,
\label{Prel4e}
\ee
cf. Lemma \ref{lem:Z5ad} with $\Gamma_2(C_{(\beta)})=e$.
Given a tuple $(\eta,[T_{(\beta)}])_{\beta}\in \bG_X$, we deduce the equalities
\be
\eta C_{\nu} \eta^{-1} = C_{\nu}, \quad
\eta C_{(\beta)} \eta^{-1} = C_{(\beta)},  \quad \beta=1,\ldots,q.
\label{Prel4f}
\ee
The conditions \eqref{Prel4d} entail $\eta\in \cZ(K)$.
Then $T_{(\beta)} C_{\rho_\beta}  T_{(\beta)}^{-1}=C_{\rho_\beta}$ gives $T_{(\beta)}\in \cZ(K)$ by using \eqref{Prel4d} once more. Hence the class $[T_{(\beta)}]$ is trivial in $\bT^{\ad}$, and therefore $(\eta,[T_{(\beta)}])_{\beta}\in \cZ(K)\times \{e\}$.
\end{proof}

\begin{lem}\label{lem:Preli5}
Assume that $m=1$ and $n\geq 0$.
Then the principal isotropy group for the $\bG$-action on  $Y$ is the direct product $\cZ(K) \times \{e\}$.
\end{lem}
\begin{proof}
If $n=0$, Assumption \ref{assume:HIJ} imposes that $\widehat{I}=\emptyset$, and therefore this is just Lemma \ref{lem:Z1}.
So we can assume that $n\geq 1$ and write $M_{1,n}=\{(A,B,C_1,\ldots,C_n)\}$.

\underline{Case 1) $\widehat{I}=\emptyset$.}
If $n=1$, we must have $J=\emptyset$ and thus $I=\{1\}$. So the result is simply obtained as in the proof of Lemma \ref{lem:Z1}. If $n\geq 2$ and $J=\emptyset$, the result can be deduced in the same way.

Hence, we can assume that $J$ is not empty. Let us also assume that $I=\{1\}$, as the proof that we present now also works for $I=\widehat{I}=\emptyset$ by forgetting about the factor $\bT$ of the action.
Inspired by Lemmas \ref{lem:Z1} and \ref{lem:Preli4} and adopting the notation \eqref{Prel4a},
we consider a point $X\in Y$ such that $A\in \exp(\ri \cA)$ and
\be
BAB^{-1},C_{\rho_1},C_{\rho_2},\ldots,C_{\rho_q}  \in \bT'\cap K^\reg, \quad
A,C_{(1)},C_{(2)},\ldots,C_{(q)} \in \exp(\ri \cA) \subset \bT^\reg,
\label{Prel50}
\ee
where $\bT'$ is a maximal torus verifying $\bT\cap \bT'=\cZ(K)$.
One gets that $\bT \times (\bT^\ad)^q$ acts on $Y$ through
\be
(T,[T_{(1)}],\ldots,[T_{(q)}],X)\mapsto X(\taU)
\ee
where all components of $X(\taU)$ are those of $X$ except for \eqref{Prel4e} and $B(\taU) = B T^{-1}$.
Given $(\eta,T,[T_{(\beta)}])_{\beta}\in \bG_X$ and acting on $X$ yields the equalities
\be
\eta BAB^{-1}\eta^{-1}=BAB^{-1}, \quad \eta A \eta^{-1} = A,
\ee
which  imply $\eta\in \bT\cap \bT'=\cZ(K)$.  Then, looking at the action on $B$ entails $T=e$.
We conclude that the classes $[T_{(1)}],\ldots,[T_{(q)}]\in \bT^{\ad}$ are the identity element by repeating the end of the proof of Lemma \ref{lem:Preli4}.

\underline{Case 2) $\widehat{I}=\{1\}$.}
In this case, $I=\emptyset$, and our general assumption guarantees that $n\geq 1$.
The method is based on Lemmas \ref{lem:DD4} and \ref{lem:Preli4}.
We consider a point $X\in Y$ such that
\be
C_{\rho_1},C_{\rho_2},\ldots,C_{\rho_q}  \in \bT'\cap K^\reg, \quad C_{(1)},C_{(2)},\ldots,C_{(q)} \in \exp(\ri \cA) \subset \bT^\reg,
\label{Prel5a}
\ee
with the notation \eqref{Prel4a}. If $q=0$, one simply takes $C_1\in \bT^\reg$. We further assume that,
for $g\in K$ that verifies $g \bT g^{-1} = \bT'$,
$A\in \bT'\cap K^\reg$ and $B$ is a representative of a Coxeter element of the pair $(K, \bT')$ such that
$[A,B] \in g \exp(\ri \cA) g^{-1}$ (cf. the proof of Lemma \ref{lem:DD4}).
Taking the chosen point as initial value, the torus $(\bT^\ad)^{q+1}$ acts on $Y$ through
\be
([T],[T_{(1)}],\ldots,[T_{(q)}],X)\mapsto X(\taU)
\ee
where all components of $X(\taU)$ are those of $X$ except
\be
B(\taU) = T' B T'^{-1}, \quad
C_j(\taU) = T_{(\beta)} C_j  T_{(\beta)}^{-1}, \,\,
\lambda_\beta \leq j \leq \rho_\beta, \,\, 1\leq \beta\leq q,
\label{Prel5b}
\ee
for $T' = g T g^{-1}\in \bT'$ corresponding to $T \in \bT$.
Due to our assumptions on the chosen point, an element $(\eta,[T],[T_{(\beta)}])\in \bG_X$ of the isotropy group satisfies
$\eta A \eta^{-1}=A$ and $\eta C_{(\beta)} \eta^{-1} = C_{(\beta)}$
(for all $\beta$, or $\eta C_{1} \eta^{-1}=C_1$ for $q=0$), which entails $\eta\in \cZ(K)$.
Following the usual argument, one then deduces from $\bT\cap \bT'=\cZ(K)$ and \eqref{Prel5b} that all $T_{(\beta)}$ are central.
Also $T\in \cZ(K)$ because $T'\in \cZ(K)$ as it is fixed by a Coxeter element using \eqref{Prel5b}.
Therefore the respective classes $[T],[T_{(1)}],\ldots,,[T_{(q)}]$ in $\bT^{\ad}$ are the identity.
\end{proof}

\begin{lem}\label{lem:Prem2}
Assume that $m=2$.
Then the principal isotropy group for the $\bG$-action on  $Y$ is the direct product $\cZ(K) \times \{e\}$.
\end{lem}
\begin{proof}
If $n=0$, then Lemma \ref{lem:DD2} is the case $I=\{2\}$ with $\widehat{I}=\{1\}$
(whose proof is direct to adapt for $I=\{1\}$,  $\widehat{I}=\{2\}$),
and Lemma \ref{lem:DD3} is the case $I=\emptyset$ with $\widehat{I}=\{1,2\}$.
As explained in Remark \ref{Rem:gen2}, one can also easily handle the case $I=\{1,2\}$ with $\widehat{I}=\emptyset$,
or the cases where $p+\widehat{p}=1$. These are all the possibilities for $n=0$.
Let us therefore assume that $n\geq 1$.

\underline{Case 1.} \emph{$J$ is not empty or $1\in I\sqcup \widehat{I}$.} We realize the manifold $M_{2,n}$ as
 \be
 M_{2,n} \simeq M_{1,n} \times \bD(K) = \{(X_-,A_2,B_2)\}, \quad X_-=(A_1,B_1,C_1,\ldots,C_n)\,.
 \label{Prem2a}
 \ee
 \underline{Subcase 1.a.} \emph{$2\in I$.}
 The open subspace $Y\subset M_{2,n}$ \eqref{Z33} can be factored as
 $Y_- \times K^{\reg} \times K$, where $Y_-\subset M_{1,n}$ is defined as in \eqref{Z33} after omitting the index $2$ from $I$.
Furthermore, we can realize $\vec{\bT}\simeq \vec{\bT}_- \times \bT$, with $\bT$ corresponding to
 $(H_2^{\chi_1},\ldots,H_2^{\chi_\ell})$, such that the action of $\vec{\bT}_-$ restricts to $Y_-$.
By Lemma \ref{lem:Preli5}, there exists $X_-\in Y_-$ whose isotropy group with respect to the action of
$\bG':=K\times \vec{\bT}_- \times \{e\}\subset \bG$ is $\cZ(K)\times \{e\}$.
Thus, if one takes $A_2\in \exp(\ri \cA)$ and works with $X:=(X_-,A_2,B_2)\in Y$, an element
$(\eta,T_-,T)\in \bG_X$ of its isotropy group must be such that $(\eta,T_-)\in \bG'_{X_-}$,
which entails $\eta\in \cZ(K)$, $T_-=e$, and then $BT^{-1}=B$.
This last condition implies $T=e$ and we are done.

 \underline{Subcase 1.b.} \emph{$2\in \widehat{I}$.}
 The open subspace $Y\subset M_{2,n}$ \eqref{Z33} can be factored as
 $Y_- \times \Phi^{-1}(K^{\reg})$, where $Y_-\subset M_{1,n}$ is defined as in \eqref{Z33} after omitting the index $2$ from $\widehat{I}$ and $\Phi$ is the momentum map for the $\bD(K)$ factor in \eqref{Prem2a}.
 As for the previous case, we can realize $\vec{\bT}\simeq \vec{\bT}_- \times \bT^{\ad}$ such that
 there exists $X_-\in Y_-$ whose isotropy group with respect to the action of
 $K\times \vec{\bT}_- \times \{e\}\subset \bG$ is $\cZ(K)\times \{e\}$, by Lemma \ref{lem:Preli5}.
 Let us work with $X:=(X_-,A_2,B_2)\in Y$ where, for $g\in K$ verifying $g \bT g^{-1} = \bT'$,
 $A_2\in \bT'\cap K^\reg$ and $B_2$ is a representative of a Coxeter element of the pair $(K, \bT')$ such that
 $[A_2,B_2] \in g \exp(\ri \cA) g^{-1}$ (cf. the proof of Lemma \ref{lem:DD4}).
 Then $(\eta,T_-,[T])\in \bG_X$ must be such that $\eta\in \cZ(K)$, $T_-=e$, and
 $T' B_2 T'^{-1}=B_2$ for $T'=gTg^{-1}$.
We have already remarked previously that the last condition implies $T'\in \cZ(K)$, and therefore $[T]=e$ from which we can conclude.

\underline{Subcase 1.c.} \emph{$2\notin I\sqcup\widehat{I}$.}
One can just run the argument of one of the previous 2 cases since there is no regularity condition and no torus action on the pair $(A_2,B_2)$.

\underline{Case 2.} \emph{$J$ is empty and $1\notin I\sqcup \widehat{I}$.}
Since $p+\widehat{p}+q>0$, we must have $I\sqcup \widehat{I} = \{2\}$.
So it suffices to swap the pairs $(A_1,B_1)$ and $(A_2,B_2)$ in case 1.a) if $2\in I$, or in case 1.b) if $2\in \widehat{I}$.
\end{proof}

\begin{lem}\label{Lem:App1}
Working under Assumption \ref{assume:HIJ},
with $\bG:= K \times \vec{\bT}$, the principal isotropy group for the $\bG$-action on  $Y$ \eqref{Z33}
is the direct product $\cZ(K) \times \{e\}$.
In particular, the principal isotropy subgroups of the $K$ and $\vec{\bT}$ actions are $\cZ(K)$ and $\{e\}$, respectively.
\end{lem}
\begin{proof}
If $m=0,1,2$, the result is given by Lemmas \ref{lem:Preli4}, \ref{lem:Preli5} and \ref{lem:Prem2}.
So we can assume that $m\geq 3$ and obtain the result by induction on $m$ by adapting the proof of Lemma \ref{lem:Prem2}.

\underline{Case 1.} \emph{$J$ is not empty or there exists $i\in I\sqcup \widehat{I}$ with $i\neq m$.}
We realize the manifold $M_{m,n}$ as
\be
M_{m,n} \simeq M_{m-1,n} \times \bD(K) = \{(X_-,A_m,B_m)\},
\label{App1a}
\ee
with $X_-=(A_1,B_1,\ldots,A_{m-1},B_{m-1},C_1,\ldots,C_n)$.

\underline{Subcase 1.a.} \emph{$m\in I$.}
The open subspace $Y\subset M_{m,n}$ \eqref{Z33} can be factored as
$Y_- \times K^{\reg} \times K$, where $Y_-\subset M_{m-1,n}$ is defined as in \eqref{Z33} after omitting the index $m$ from $I$.
Furthermore, we can realize $\vec{\bT}\simeq \vec{\bT}_- \times \bT$ in such a way that the action of $\vec{\bT}_-$ restricts to $Y_-$.
Note that $\vec{\bT}_-$ is of dimension $(p+\widehat{p}+q-1)\ell >0$ since either $J$ or $(I\sqcup \widehat{I})\setminus \{m\}$ is not empty in this case.

By induction, there exists $X_-\in Y_-$ whose isotropy group for the action of
$\bG':=K\times \vec{\bT}_- \times \{e\}\subset \bG$ is $\cZ(K)\times \{e\}$.
Thus, if one takes $A_m\in \exp(\ri \cA)$ and works with $X:=(X_-,A_m,B_m)\in Y$,
we can argue as in the corresponding case of Lemma \ref{lem:Prem2} that $\bG_X=\cZ(K)\times \{e\}$.

\underline{Subcase 1.b.} \emph{$m\in \widehat{I}$.}
The open subspace $Y\subset M_{m,n}$ \eqref{Z33} can be factored as
$Y_- \times \Phi^{-1}(K^{\reg})$, where $Y_-\subset M_{m-1,n}$ is defined as in \eqref{Z33}
after omitting the index $m$ from $\widehat{I}$.
We can again realize $\vec{\bT}\simeq \vec{\bT}_- \times \bT^{\ad}$ in such a way that $\vec{\bT}_-$ is of positive dimension and its action restricts to $Y_-$.
Thus, by induction, there exists $X_-\in Y_-$ whose isotropy group with respect to the action of
$K\times \vec{\bT}_- \times \{e\}\subset \bG$ is $\cZ(K)\times \{e\}$.
As in the corresponding case of Lemma \ref{lem:Prem2},
we choose to work with $X:=(X_-,A_m,B_m)\in Y$ where
$A_m\in \bT'\cap K^\reg$, $B_m$ is a representative of a Coxeter element of $(K, \bT')$, and
$[A_m,B_m] \in g \exp(\ri \cA) g^{-1}$. One then deduces that $\bG_X=\cZ(K)\times \{e\}$.

\underline{Subcase 1.c.} \emph{$m\notin I\sqcup\widehat{I}$.}
We can repeat one of the previous 2 arguments with no regularity condition or torus action on $(A_m,B_m)$.

\underline{Case 2.} \emph{$J$ is empty and $I\sqcup \widehat{I} = \{m\}$.}
We realize the manifold $M_{m,n}$ as
\be
M_{m,n} \simeq M_{m-1,n} \times \bD(K) = \{(X_+,A_1,B_1)\},
\label{App1b}
\ee
with $X_+=(A_2,B_2,\ldots,A_{m},B_{m},C_1,\ldots,C_n)$,
and then we proceed similarly to Case 1.
\end{proof}

\subsubsection{Applicability of the main theorem}

We work under the conditions of Assumption \ref{assume:HIJ} for $m,n$ and $I,\widehat{I},J$.
\begin{prop}\label{pr:Z11}
Theorem \ref{thm:G12} is applicable to  the Abelian Poisson algebra $\fH(I,\widehat{I},J)$ from Definition \ref{defn:Z7}
with  $M=M_{m,n}$, $Y(I,\widehat{I},J)\subset M$ \eqref{Z33}, $\vec{H}$ \eqref{Z34},
$\Gone=K$ and $\Gtwo=\vec{\bT}$ \eqref{vecT} acting as described in the above.
Consequently,  $\fH(I,\widehat{I},J)$  engenders integrable systems of rank $(p+\widehat{p}+q)\ell$ whose action variables arise from $\vec{H}$.
\end{prop}

\begin{proof}
We need to check that all assumptions of Scenario \ref{scen:G11} are satisfied, so that Theorem \ref{thm:G12} can be applied.
We use the shorthand notation $\vec{\bT}$, $\bG:= K \times \vec{\bT}$, $Y$ and $\fH$ introduced at the beginning of the appendix.

It is clear that $\fH$ is Abelian and the integral curves its elements define are complete since $M_{m,n}$ is compact (see also Lemma \ref{lem:Z8} and Remark \ref{rem:Z9}). Thus condition \ref{itScen:1} in Scenario \ref{scen:G11} is satisfied.
The components of $\vec{H}$ \eqref{Z34} are smooth functions on $Y$ in terms of which we can express any element of $\fH$. They generate an action of the torus $\vec{\bT}\simeq \operatorname{U}(1)^{(p+\widehat{p}+q)\ell}$, automatically proper,
whose principal isotropy group is trivial.
Hence, conditions \ref{itScen:2} and \ref{itScen:3}  in Scenario \ref{scen:G11} are satisfied.
Finally, the isotropy group of $X\in Y_0$  is
\be
 \bG_X = \cZ(K) \times \{e\} = K_X \times \vec{\bT}_X\,,
\ee
by Lemma \ref{Lem:App1}. Therefore condition \ref{itScen:4} also holds.
\end{proof}

\subsection{Possible integrable systems obtained by extending \texorpdfstring{$\fH(I,\widehat{I},J)$}{H(I,I,J)}}

We briefly outline further potential applications of our method for obtaining integrable systems on $(M_{m,n})_*/K$
and symplectic leaves therein.
We leave a detailed study of these cases to the interested reader.

\subsubsection{Extra invariant functions from momentum maps}
\label{sss:Extra}

As a modification of the general construction considered so far, we can add Hamiltonians of the form
\be
\widetilde H_{[k_1,k_2]}^\chi(X) := \chi([A_{k_1},B_{k_1}]\, [A_{k_1+1},B_{k_1+1}] \cdots [A_{k_2},B_{k_2}]), \quad
\forall \chi\in C^\infty(K)^K,\,
\label{Z40}
\ee
with $1\leq k_1<k_2\leq m$
whenever $\{k_1,k_1+1,\ldots,k_2\} \cap (I\sqcup \widehat{I}) = \emptyset$, or
\be
\widetilde H_{k;\kappa}^\chi(X) := \chi([A_{k},B_{k}] \cdots [A_{m},B_{m}]\, C_1 \cdots C_\kappa),
\quad \forall \chi\in C^\infty(K)^K,
\label{Z41}
\ee
if $k>\max(i_p,j_{\widehat{p}})$ and $\kappa<\lambda_1$. The Hamiltonians in \eqref{Z40} and \eqref{Z41} are $K$-invariant functions
of  `partial momentum maps' associated with consecutive copies of quasi-Poisson factors appearing in $M_{m,n}$ \eqref{W11}.
Thus, they represent generalizations  of the Hamiltonians \eqref{Z27bar} and \eqref{Z30}.
In view of Lemmas \ref{lem:Wmomap} and \ref{Lem:Fus}, their  flows commute with all flows generated by $\fH(I,\widehat{I},J)$.
As is always the case when the Hamiltonians are obtained from momentum maps,
 one has to apply the functions $\Xi_j$ \eqref{LI5} to define the pertinent action variables.

There are two technical issues to consider in this case.
The first one is to check the connectedness of the analogue of $Y(I,\widehat{I},J)$ \eqref{Z33}, where we further require that
 the arguments appearing in the right-hand sides of \eqref{Z40} and \eqref{Z41} are all regular.
This is direct for the latter (we can replace $C_\kappa$ by the argument of \eqref{Z41} in the parametrization of $M_{m,n}$), but
 connectedness
after requiring the regularity of a product of commutators as in \eqref{Z40} needs further
justification.
The second question is to see if the principal isotropy group for the relevant $\bG$-action is still given
 by $\cZ(K)\times \{e\}$ as in the previous cases.

\subsubsection{Nested Hamiltonians}

Yet another way to add Hamiltonians to $\fH(I,\widehat{I},J)$ consists in allowing nested non-intersecting intervals of $\{1,\ldots,n\}$.
For clarity, relabel $J$ and the different indices appearing in \eqref{Z28} as $J^{(0)}$ and $\lambda_j^{(0)},\rho_j^{(0)}$ for $1\leq j \leq q$.
Next, we fix a set of non-intersecting closed intervals
\be
J^{(1)}:= \{ [\lambda_1^{(1)}, \rho_1^{(1)}], [\lambda_2^{(1)}, \rho_2^{(1)}], \dots, [\lambda_{q_1}^{(1)}, \rho_{q_1}^{(1)}] \},
\label{Z71} \ee
such that the boundaries of the intervals are integers satisfying the obvious analogue of \eqref{Z29}.
Furthermore, for all $1\leq \beta_0\leq q$, $1\leq \beta_1\leq q_1$, we assume that we are in one of the following two situations:
\begin{itemize}
 \item $[\lambda_{\beta_0}^{(0)}, \rho_{\beta_0}^{(0)}] \cap [\lambda_{\beta_1}^{(1)}, \rho_{\beta_1}^{(1)}] = \emptyset$, or
 \item $[\lambda_{\beta_0}^{(0)}, \rho_{\beta_0}^{(0)}] \cap [\lambda_{\beta_1}^{(1)}, \rho_{\beta_1}^{(1)}] =
 [\lambda_{\beta_0}^{(0)}, \rho_{\beta_0}^{(0)}] \subsetneq [\lambda_{\beta_1}^{(1)}, \rho_{\beta_1}^{(1)}]$.
\end{itemize}
In other words, an interval from $J^{(0)}$ is either properly contained in an interval from $J^{(1)}$, or it does not intersect it.
We associate with $J^{(1)}$ the Hamiltonians
 \be
 H_{[\lambda_\beta, \rho_\beta]}^{(1),\chi}(X) := \chi(C_{\lambda_\beta^{(1)}} C_{\lambda_\beta^{(1)} + 1}\cdots C_{\rho_\beta^{(1)}}),
 \quad \forall \chi \in C^\infty(K)^K,\,\, \forall \beta =1,\dots, q_1.
 \label{Z72}\ee
One can then verify that the corresponding integral curves take the form \eqref{Z32}, and that the Hamiltonians \eqref{Z72} commute with those in $\fH(I,\widehat{I},J)$.

We may iterate the procedure and select  a set of non-intersecting closed intervals $J^{(2)}$ subject to \eqref{Z29}, so that each new interval either properly contains any
given interval of $J^{(0)}$ or $J^{(1)}$, or it does not intersect it at all. Then we can define  Hamiltonians corresponding to $J^{(2)}$ as in \eqref{Z72}.
The final family of Poisson commuting Hamiltonians is made out of \eqref{Z27}, \eqref{Z27bar}, \eqref{Z30}, \eqref{Z72} and those that may arise from  further compatible
 nesting intervals $J^{(2)},J^{(3)}$ and so on.
(We call these `nested Hamiltonians' since  their arguments represent holonomies of flat connections along families of
nested curves around the $n$ boundary components; see also Subsection \ref{ss:compAR}.)
The procedure can also be modified so as to `shrink the intervals'.
This means that one can add such sets of intervals $J^{(-1)},J^{(-2)},\ldots$ for which  each interval from $J^{(k-1)}$ is either disjoint from or is  properly contained in any given
interval from $J^{(k+j)}$, $j\geq 0$.
However, as the construction rapidly becomes tedious to handle, we here do not discuss the applicability
of  Scenario \ref{scen:G11} to the ensuing ingredients.

\subsection{New systems from permutations}

So far we constructed integrable systems on $M_{m,n}$ \eqref{W11}, which is obtained by successive fusion
of the building blocks of the product manifold $\bD(K) \times \cdots \times \bD(K) \times K\times \cdots \times K$.
However, one may start from another quasi-Poisson manifold $\check{M}_{m,n}$, obtained similarly to $M_{m,n}$ by fusion of
$m$ copies of $\bD(K)$ and $n$ copies of $K$, but taken in an arbitrary order, different from that in $M_{m,n}$.
After building an integrable system on $\check{M}_{m,n}$, we can transfer it to our `canonical model'   $M_{m,n}$, where
it may give a system not covered by the previous constructions.
To this end, it suffices to read the family of commuting Hamiltonians of the integrable system on $\check{M}_{m,n}$ through a quasi-Poisson diffeomorphism between $\check{M}_{m,n}$ and  $M_{m,n}$,
and such a quasi-Poisson diffeomorphism can be derived by successive applications of \cite[Prop.~5.7]{AKSM} relating  $M_1\circledast M_2$ to $M_2\circledast M_1$.
Indeed, considering the fusion product of $N$ quasi-Poisson manifolds $M_j$, $1\leq j \leq N$,  the permutation of two consecutive factors induces the
quasi-Poisson diffeomorphism
\begin{equation}
 \begin{aligned} \label{Z50}
  \psi_j : M_1 \circledast \cdots \circledast M_j \circledast M_{j+1} \circledast \cdots \circledast M_N
  \to M_1 \circledast \cdots \circledast M_{j+1} \circledast M_{j} \circledast \cdots \circledast M_N \\
  \psi_j(m_1 ,\ldots, m_j,m_{j+1},\ldots,m_N) = (m_1 ,\ldots,  \Phi_j(m_j) \cdot m_{j+1}, m_j,\ldots,m_N),
 \end{aligned}
\end{equation}
 where $\Phi_j$ is the momentum map on $M_j$ and $\Phi_j(m_j) \cdot m_{j+1}$ uses the $K$-action on $M_{j+1}$.

We now exhibit an example of the construction just sketched.
To do this, we start with
\be
\check{M}_{2,2} =  \bD(K) \circledast K \circledast \bD(K) \circledast K,
\label{Z51}\ee
and write the points $\check{X}\in \check{M}_{2,2}$ as $\check{X}= (U_1,V_1,W_1 , U_2,V_2,W_2)$.
Then, we introduce the Abelian Poisson subalgebra $\check\fH$ of $C^\infty(\check{M}_{2,2})^K$ whose elements are polynomials in
\be
\check{H}_k^\chi(\check{X}) := \chi([U_k,V_k]\, W_k),
\quad \forall \chi\in C^\infty(K)^K, \,\, k=1,2.
\label{Z52}
\ee
For $k=1,2$, put $D_k:=[U_k,V_k]\, W_k$. Notice that
\be
\check{\Phi}_k:\check{M}_{2,2}\to K, \qquad \check{\Phi}(\check{X})=D_k,
\label{Z53}
\ee
is the momentum map of the quasi-Poisson manifold $\bD(K) \circledast K$ appearing as the $k$-th couple of factors in \eqref{Z51}.
Arguing like in the proof of Lemma \ref{lem:Z8}, we can find the integral curves corresponding to any Hamiltonian $\check{H}_k^\chi$.
Starting at $(u_1,v_1,w_1, u_2,v_2,w_2)\in \check{M}_{2,2}$, the  integral curve generated by  $\check{H}_k^\chi$,
\be
 \check{X}(\tau) = (U_1(\tau),V_1(\tau),W_1(\tau)   , U_2(\tau),V_2(\tau),W_2(\tau)),
\label{Z54}
\ee
 satisfies (for  $\ell=1,2$)
\begin{equation}
 \begin{aligned}
   U_\ell(\tau) &=
\exp(\delta_{k\ell}\,\tau \nabla \chi(d_k)) \,u_\ell\, \exp(-\delta_{k\ell}\, \tau \nabla \chi(d_k)) , \\
   V_\ell(\tau) &=
\exp(\delta_{k\ell}\,\tau \nabla \chi(d_k)) \, v_\ell\, \exp(-\delta_{k\ell}\, \tau \nabla \chi(d_k)) , \\
   W_\ell(\tau) &=
\exp(\delta_{k\ell}\,\tau \nabla \chi(d_k)) \, w_\ell\, \exp(-\delta_{k\ell}\, \tau \nabla \chi(d_k)) ,
\label{Z55}
 \end{aligned}
\end{equation}
with $d_k:=[u_k,v_k]w_k$. In particular, \eqref{Z55} gives constant elements for $\ell \neq k$.

A connected, dense open submanifold of $\check{M}_{2,2}$ is provided by
\be
\check{Y}:=\check{\Phi}_1^{-1}(K^\reg) \cap \check{\Phi}_2^{-1}(K^\reg)\,.
\label{Z55bis}
\ee
Indeed, $\check{Y}$ is diffeomorphic to $\bD(K)^2 \times (K^{\reg})^2$
since we can replace $W_k$ by $D_k$ in the parametrization of $\check{M}_{2,2}$.
Then, as is now standard, we can prove the following result.
\begin{lem}\label{lem:BZ1}
The joint flows of the quasi-Hamiltonian vector fields of the components of
\be
\vec{\check{H}}:=(\check{H}_1^{\Xi_1}, \dots,   \check{H}_1^{\Xi_\ell},
\check{H}_2^{\Xi_1}, \dots,   \check{H}_2^{\Xi_\ell}): \check{Y} \to \bR^{2\ell},
\label{Z56}
\ee
yield an effective action of the torus $(\bT^{\ad})^2$ on $\check{Y}$.
Using $T_{\omega^\vee}(\taU)$ \eqref{LI3}
and defining, for $k=1,2$, $\check{g}_{k|\taU}=\check{g}_{k|\taU}(U_k,V_k,W_k)\in K$
with the `diagonalizer' $[\Gamma_2]$  \eqref{J11} by
\be
 \check{g}_{k|\taU} := \Gamma_2([U_k,V_k]W_k)^{-1} T_{\omega^\vee}(\taU) \Gamma_2([U_k,V_k]W_k),
\label{Z57}
\ee
the action of $([T_{\omega^\vee}(\taU_1)],[T_{\omega^\vee}(\taU_2)])\in (\bT^{\ad})^2$ on $\check{Y}$  is given by the mapping
\be
(U_\ell,V_\ell,W_\ell)_{\ell=1,2} \mapsto
\Big(\Ad_{\check{g}_{\ell|\taU_\ell}}(U_\ell),\Ad_{\check{g}_{\ell|\taU_\ell}}(V_\ell) ,
\Ad_{\check{g}_{\ell|\taU_\ell}}(W_\ell)\Big)_{\ell=1,2}\,.
\label{Z58}
\ee
\end{lem}
The action \eqref{Z58} commutes with the $K$-action restricted to $\check{Y}$.
We then set $\bG:= K\times (\bT^{\ad})^2$, which also acts on $\check{Y}$.

\begin{lem}\label{lem:BZ2}
The principal isotropy group for the $\bG$-action on $\check{Y}$ is $\cZ(K) \times \{e\}$.
\end{lem}
\begin{proof}
We can adapt the proof of Lemma \ref{lem:DD4}.
Fix a second maximal torus $\bT'\subset K$ for which $\bT \cap \bT' = \cZ(K)$, and an element $g\in K$ that verifies $g \bT g^{-1} = \bT'$.
Consider a point $\check{X}\in \check{Y}$ such that for $k=1,2$,
$U_k\in \bT'\cap K^\reg$, $V_k\in \bT^\reg$,
and $W_1,W_2$ are chosen so that $[U_1,V_1]W_1\in \exp(\ri \cA)$ and $[U_2,V_2]W_2 \in g\exp(\ri \cA)g^{-1}$.

Now, we deduce from \eqref{Z58} that the action of $(\eta,[T_1],[T_2])\in \bG$ is furnished by
\be
\check{X} \mapsto \eta \cdot (T_1 U_1 T_1^{-1},T_1 V_1 T_1^{-1},T_1 W_1 T_1^{-1},
T_2' U_2 T_2'^{-1}, T_2' V_2 T_2'^{-1}, T_2' W_2 T_2'^{-1}) \,,
\label{Z59}
\ee
where $T_2'=gT_2g^{-1}$.
If $(\eta,[T_1],[T_2])\in \bG_{\check{X}}$, we have in particular
\be
\eta [U_1,V_1]W_1 \eta^{-1} = [U_1,V_1]W_1, \quad
\eta [U_2,V_2]W_2 \eta^{-1} = [U_2,V_2]W_2,
\label{Z60}
\ee
that guarantee $\eta\in \bT \cap \bT' =\cZ(K)$ by our assumptions on $\check{X}$.
Looking at the $U_1$ and $V_2$ components of \eqref{Z59} then yields $T_1,T_2'\in \cZ(K)$.
Hence, $[T_1]$ and $[T_2]$ both are equal to the identity in $\bT^{\ad}$.
\end{proof}

Similarly to Proposition \ref{pr:Z11}, we can now state that $\check{\fH}$ engenders an integrable system on $(\check{M}_{2,2})_\ast/K$
and this system  admits generalized action variables induced by $\check{\vec{H}}$ \eqref{Z56} on
$\check{Y}_0/K$ and $\check{Y}_0^1/K$ (and on relevant symplectic leaves).
To proceed,  we introduce the quasi-Poisson diffeomorphism  $\psi:M_{2,2}\to \check{M}_{2,2}$ by the formula
\be
\psi(A_1,B_1,A_2,B_2,C_1,C_2)=(A_1,B_1,\Ad_{[A_2,B_2]}(C_1),A_2,B_2,C_2),
\label{Z61}
\ee
which is a special case of \eqref{Z50}.
Combining \eqref{Z52} and \eqref{Z61}, we see that the pull-back $\psi^\ast \check{\fH}$  of
$\check{\fH}$ consists of polynomials in the Hamiltonians
\be
\psi^\ast\check{H}_1^\chi(X) = \chi([A_1,B_1] \Ad_{[A_2,B_2]}(C_1)), \,\,
\psi^\ast\check{H}_2^\chi(X)  = \chi([A_2,B_2]C_2), \,\, \forall \chi\in C^\infty(K)^K.
\label{Z62}
\ee
The reduction of the Abelian Poisson algebra $\psi^\ast \check{\fH}$  gives an integrable system on $({M}_{2,2})_\ast/K$ that
admits action variables on $\psi^{-1}(\check{Y}_0)/K$ and $\psi^{-1}(\check{Y}_0^1)/K$ (and on relevant symplectic leaves).
The resulting system can  be described by using that $\psi$ descends to a Poisson diffeomorphism between the
quotient manifolds $({M}_{2,2})_\ast/K$ and $(\check {M}_{2,2})_\ast/K$.
Based on the  form \eqref{Z62} of the Hamiltonians, we conclude that
these integrable systems do \emph{not} fall within the framework of Subsection \ref{ss:genModuli}.
We  believe that they are not equivalent to any of the previously considered systems.
The message is that
the permutation method  just illustrated  may lead to an abundance of further integrable systems
on moduli spaces of flat connections via applications of Theorem \ref{thm:G12}.

\subsection{Towards comparison with the work of  Arthamonov and Reshetikhin} \label{ss:compAR}

The approach of Arthamonov and Reshetikhin \cite{ARe} for constructing integrable systems on moduli spaces of flat connections relies
on  Abelian Poisson algebras associated to collections of  pairwise non-homotopic and non-intersecting simple closed curves on the corresponding Riemann surface.
There is a one-to-one map \cite{DFN} from free homotopy classes of closed curves to conjugacy classes of the fundamental group, which in principle underlies the connection to our approach.
However, already finding all conjugacy classes of the fundamental group that correspond to simple closed curves is, in general, a difficult problem \cite{BS}.
 Below, we outline a way to understand the Abelian Poisson algebra $\fH(I,\widehat{I},J)$ constructed in Subsection \ref{ss:genModuli} in terms of curves on  the Riemann surface $\Sigma_{m,n+1}$ of genus $m$ with $n+1$ boundary components.
This material is intended to give the gist of the comparison between our approach with the one presented in \cite{ARe}.

The fundamental group of $\pi(\Sigma_{m,n+1})$ can be presented as
\be
\pi(\Sigma_{m,n+1}) =  \langle \alpha_1,\beta_1,\ldots,\alpha_m,\beta_m,\gamma_1,\ldots,\gamma_{n+1} \mid
[\alpha_1,\beta_1]\cdots [\alpha_m,\beta_m] \, \gamma_1 \cdots \gamma_{n+1} =1\,\rangle .
\ee
The homotopy class of any closed curve on $\Sigma_{m,n+1}$ corresponds to a conjugacy class in $\pi(\Sigma_{m,n+1})$.
In this way, it is standard that $\alpha_j$ and $\beta_j$ correspond to the $\mathrm{A}$- and $\mathrm{B}$-cycle around the $j$-th hole, respectively, while $\gamma_k$ corresponds to the oriented loop going around the $k$-th boundary component.
The word $\mathtt{w}_j:=[\alpha_1,\beta_1]\cdots [\alpha_j,\beta_j]$, $1\leq j \leq m$, corresponds to the curve $\mathtt{c}_j$
separating the surface into $\Sigma_{j,1}\sqcup \Sigma_{m-j,n+2}$ after the $j$-th hole.
Hence $[\alpha_j,\beta_j]=\mathtt{w}_{j-1}^{-1}\mathtt{w}_j $ corresponds to
the composition of curves $\mathtt{c}_{j-1}^{-1}\mathtt{c}_j$.
For an example, see Figure \ref{fig:Pi} where $[\alpha_2,\beta_2]$ is depicted on $\Sigma_{2,4}$.

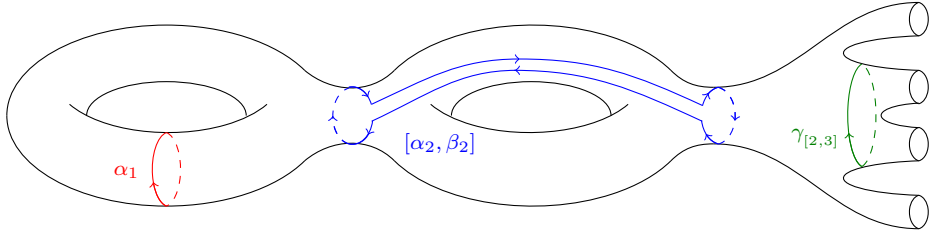
\begin{figure}[ht]
   \begin{tikzpicture}[scale=1.5]
% torus left
     \draw (-3,0) arc [start angle = 180, end angle = 30,  x radius = 14mm, y radius = 8mm]; %Up outer torus
     \draw (-3,0) arc [start angle = 180, end angle = 330,  x radius = 14mm, y radius = 8mm];%Down outer torus
     \draw (-2.3,0) arc [start angle = 180, end angle = 0,  x radius = 7mm, y radius = 3mm];%Up inner torus
     \draw (-2.45,0.1) arc [start angle = 210, end angle = 330,  x radius = 10mm, y radius = 5mm];%Down inner torus
% torus middle
     \draw (0.4,0.4) arc [start angle = 150, end angle = 30,  x radius = 14mm, y radius = 8mm]; %Up outer torus
     \draw (0.4,-0.4) arc [start angle = 210, end angle = 330,  x radius = 14mm, y radius = 8mm];%Down outer torus
     \draw (2.3,0) arc [start angle = 0, end angle = 180,  x radius = 7mm, y radius = 3mm];%Up inner torus
     \draw (2.45,0.1) arc [start angle = 330, end angle = 210,  x radius = 10mm, y radius = 5mm];%Down inner torus
     % connecting the tori to the right part
\draw[black] (-0.388,0.4) to[out=-40,in=225] (0.4,0.4);
\draw[black] (-0.388,-0.4) to[out=40,in=-225] (0.4,-0.4);
\draw[black] (2.825,0.4) to[out=-45,in=180] (3.2,0.25); \draw[black] (3.2,0.25) to[out=0,in=180] (5,1); %up
\draw[black] (2.825,-0.4) to[out=45,in=180] (3.2,-0.25);  \draw[black] (3.2,-0.25) to[out=0,in=180] (5,-1); %down
  % circles on right
\draw (5,1) arc [start angle = 100, end angle = 260,  x radius = 1mm, y radius = 1.5mm];
\draw (5,1) arc [start angle = 80, end angle = -80,  x radius = 1mm, y radius = 1.5mm]; % circle 1
\draw (5,0.4) arc [start angle = 100, end angle = 260,  x radius = 1mm, y radius = 1.5mm];
\draw (5,0.4) arc [start angle = 80, end angle = -80,  x radius = 1mm, y radius = 1.5mm]; % circle 2
\draw (5,-0.4) arc [start angle = -100, end angle = -260,  x radius = 1mm, y radius = 1.5mm];
\draw (5,-0.4) arc [start angle = -80, end angle = 80,  x radius = 1mm, y radius = 1.5mm]; % circle 3
\draw (5,-1) arc [start angle = -100, end angle = -260,  x radius = 1mm, y radius = 1.5mm];
\draw (5,-1) arc [start angle = -80, end angle = 80,  x radius = 1mm, y radius = 1.5mm]; % circle 4
% connectors to the right
\draw (5,0.705) arc [start angle = 100, end angle = 260,  x radius = 8mm, y radius = 1.55mm]; % top
\draw (5,0.11) arc [start angle = 100, end angle = 260,  x radius = 4mm, y radius = 1.1mm]; % middle
\draw (5,-0.705) arc [start angle = -100, end angle = -260,  x radius = 8mm, y radius = 1.55mm]; % bottom
%%%%%% the curves %%%%%
% A_1 cycle
     \draw[red] (-1.6,-0.15) arc [start angle = 100, end angle = 260,  x radius = 1.5mm, y radius = 3.3mm]; %Left beta
     \draw[red,->] (-1.6,-0.8) arc [start angle = 260, end angle = 200,  x radius = 1.5mm, y radius = 3.3mm]; %Left bet arrow
     \draw[red,dashed] (-1.6,-0.15) arc [start angle = 80, end angle = -80,  x radius = 1.5mm, y radius = 3.3mm];%Right beta
     \node[font=\scriptsize,red] (b) at (-1.95,-0.5) {$\alpha_1$};
%%% Phi cycle
\draw[blue] (3.1,0.09) to[out=155,in=360]   (1.5,0.5);
\draw[<-,blue] (1.5,0.5) to[out=180,in=25]  (0.2,0.1); % long middle up
\draw[>-,blue] (1.5,0.4) to[out=180,in=25]  (0.2,-0.05);
\draw[blue] (3.1,-0.05) to[out=155,in=0]  (1.45,0.4); % long middle down
% left circle
\draw[blue,dashed] (0,-0.25) arc [start angle = -100, end angle = -260,  x radius = 1.7mm, y radius = 2.5mm]; %left back1
\draw[->,blue,dashed] (0,-0.25) arc [start angle = -100, end angle = -180,  x radius = 1.7mm, y radius = 2.5mm]; % left back2
\draw[blue] (0,-0.25) arc [start angle = -100, end angle = -10,  x radius = 1.8mm, y radius = 2.5mm]; % left fr1
\draw[-<,blue] (0,-0.25) arc [start angle = -100, end angle = -30,  x radius = 1.8mm, y radius = 2.5mm]; % left fr1b
\draw[blue] (0,0.25) arc [start angle = -260, end angle = -338,  x radius = 1.8mm, y radius = 2.5mm]; % left fr2
\draw[->,blue] (0,0.25) arc [start angle = -260, end angle = -320,  x radius = 1.8mm, y radius = 2.5mm]; % left fr3
% right circle
\draw[blue,dashed] (3.24,0.25) arc [start angle = 80, end angle = -80,  x radius = 1.8mm, y radius = 2.55mm]; % right bk1
\draw[->,blue,dashed] (3.24,0.25) arc [start angle = 80, end angle = -10,  x radius = 1.8mm, y radius = 2.55mm]; % right bk2
\draw[-<,blue] (3.24,0.25) arc [start angle = 100, end angle = 140,  x radius = 1.8mm, y radius = 2.55mm]; % right front1
\draw[blue] (3.24,0.25) arc [start angle = 100, end angle = 160,  x radius = 1.8mm, y radius = 2.55mm]; % right front1b
\draw[->,blue] (3.24,-0.25) arc [start angle = -100, end angle = -140,  x radius = 1.8mm, y radius = 2.55mm]; % right fr2
\draw[blue] (3.24,-0.25) arc [start angle = -100, end angle = -170,  x radius = 1.8mm, y radius = 2.55mm]; % right fr2b
%   \draw[->,blue,densely dashed] (3.25,0.25) to[out=0,in=90]   (3.45,0); %right back upper
%   \draw[blue,densely dashed] (3.45,0) to[out=270,in=360]   (3.25,-0.25); %right back lower
     \node[font=\scriptsize,blue] (phi) at (0.8,-0.25) {$[\alpha_2,\beta_2]$};
%%% last gamma cycle
 \draw[darkgreen,dashed] (4.5,-0.45) arc [start angle = -80, end angle = 80,  x radius = 1.5mm, y radius = 4.6mm];
 \draw[darkgreen] (4.5,-0.45) arc [start angle = 260, end angle = 100,  x radius = 1.5mm, y radius = 4.6mm];
 \draw[darkgreen,->] (4.5,-0.45) arc [start angle = 260, end angle = 200,  x radius = 1.5mm, y radius = 4.6mm];
 \node[font=\scriptsize,darkgreen] (gamm) at (4.1,-0.2) {$\gamma_{_{[2,3]}}$};
  \end{tikzpicture}
  \caption{A system of curves on $\Sigma_{2,4}$ where the red, blue and green curves ($\alpha_1$, $[\alpha_2,\beta_2]$ and $\gamma_{_{[2,3]}}$) induce,  respectively, the 3 types of
  functions \eqref{Z27}, \eqref{Z27bar} and \eqref{Z30} for $I=\{1\}$, $\widehat{I}=\{2\}$ and $J=\{[2,3]\}$ on $M_{2,3}$.}
  \label{fig:Pi}
\end{figure}

For a fixed Lie group $K$, there is a natural bijection between $\operatorname{Rep}(\pi(\Sigma_{m,n+1}),K)$ and
$\Phi^{-1}_{m,n+1}(e)$.
This yields a well-defined map
\be
\rho_K:\pi(\Sigma_{m,n+1})\longrightarrow C^\infty(\Phi^{-1}_{m,n+1}(e),K), \quad
(\alpha_j,\beta_j,\gamma_k)\mapsto (A_j,B_j,C_k)\,.
\ee
Here, we assign to any $\Gamma\in \pi(\Sigma_{m,n+1})$ the evaluation function $\rho_K(\Gamma)$ defined as
$\rho_K(\Gamma)(x)=x(\Gamma)$ for a $K$-valued representation $x$ of the fundamental group.
In particular, given any closed curve $\mathtt{c}$ on $\Sigma_{m,n+1}$,
we can consider a representative $\Gamma_{\mathtt{c}}\in \pi(\Sigma_{m,n+1})$ of the corresponding conjugacy class in the fundamental group,
and then look at all the $K$-invariant elements it generates in $C^\infty(M_{m,n})\simeq C^\infty(\Phi^{-1}_{m,n+1}(e))$, cf. equation \eqref{W16}
and Lemma \ref{lem:Z0}.
These are given by
\be
\chi\circ \rho_K(\Gamma_{\mathtt{c}}), \quad \forall \chi \in C^\infty(K)^K.
\ee
From this description, one sees that the $\mathrm{A}$-cycles give rise to the functions of the form \eqref{Z27},
while the curves $[\alpha_j,\beta_j]$ and $\gamma_\lambda \gamma_{\lambda+1}\cdots \gamma_{\lambda+k}$
 correspond to the `partial momentum maps' given by $[A_j,B_j]$ and $C_\lambda C_{\lambda+1}\cdots C_{\lambda+k}$ on
 $\Phi_{m,n+1}^{-1}(e)\simeq M_{m,n}$ that  yield the functions \eqref{Z27bar} and \eqref{Z30}, respectively.

\end{document}